\documentclass{extarticle} 
\usepackage{geometry}
\usepackage{times}  %
\usepackage{helvet}  %
\usepackage{courier}  %
\usepackage[hyphens]{url}  %
\usepackage{graphicx} %
\usepackage{natbib}  %
\usepackage{caption} %
\usepackage[ruled,linesnumbered]{algorithm2e}

\makeatletter
\@ifpackageloaded{algorithmic}
 {                               %
    
 }{
   \newcommand{\algorithmtwoeFlavor}{}
 }
\makeatother

\title{Strategic Cost Selection in Participatory Budgeting}
\author {
    Piotr Faliszewski\textsuperscript{\rm 1},
    Łukasz Janeczko\textsuperscript{\rm 1},
    Andrzej Kaczmarczyk\textsuperscript{\rm 1}, \\
    Grzegorz Lisowski\textsuperscript{\rm 1},
    Piotr Skowron\textsuperscript{\rm 2},
    Stanisław Szufa\textsuperscript{\rm 1,3}\vspace{0.3cm}\\
    \textsuperscript{\rm 1} AGH University, Kraków, Poland\\
    \textsuperscript{\rm 2} University of Warsaw, Warsaw, Poland\\
    \textsuperscript{\rm 3}CNRS, LAMSADE, Université Paris Dauphine - PSL,
		Paris, France\\
    {\small
    \{faliszew,ljaneczk,andrzej.kaczmarczyk,glisowski,szufa\}@agh.edu.pl,}\\
		{\small
    p.skowron@mimuw.edu.pl}
}
\date{}

\usepackage{amsthm}

\usepackage{soul}
\usepackage[utf8]{inputenc}
\usepackage{graphicx}
\usepackage{amsmath}

\usepackage{amssymb}
\usepackage{booktabs}
\newtheorem{definition}[section]{Definition}
\newtheorem{example}{Example}

\usepackage[switch]{lineno}

\newcommand{\BibTeX}{\rm B\kern-.05em{\sc i\kern-.025em b}\kern-.08em\TeX}

\usepackage{booktabs}

\usepackage{textcomp}
\usepackage{eurosym}
\usepackage{enumitem}
\usepackage{graphicx}
\usepackage{xcolor}
\usepackage{amstext}
\usepackage{graphicx}
\usepackage{fancyhdr}
\usepackage{multicol}
\usepackage{verbatim}
\usepackage{float}
\usepackage{multirow}
\usepackage{nicefrac}
\usepackage{tikz}
\usepackage{pgffor}
\definecolor{myblue}{rgb}{0,0,0.6}
\definecolor{myred}{rgb}{0.6,0,0}
\usetikzlibrary{patterns}

\usepackage{relsize}

\tikzset{fontscale/.style = {font=\relsize{#1}}
}
\usepackage{dsfont}
\usetikzlibrary{backgrounds}
\usetikzlibrary{calc,shapes,decorations,matrix,arrows} 
\usetikzlibrary{fit}
\usepackage{etoolbox}
\usepackage{appendix}
\newcommand{\appendixproof}[3]{%
  \gappto{\appendixProofs}
  {
		\subsection{Proof of~{#1}~\ref{#2}}\label{proof:#2}
    #3
  }
}
\newcommand{\appendixcorrectness}[3]{%
  \gappto{\appendixProofs}
  {
		\subsection{Correctness of~{#1}~\ref{#2}}\label{proof:#2}
    #3
  }
}

\apptocmd{\sloppy}{\hbadness 10000\relax}{}{}

\apptocmd{\appendix}{\hbadness 10000\relax}{}{}

\DeclareMathOperator*{\pos}{pos}
\DeclareMathOperator*{\topp}{top}

\newcommand{\textbasicAV}{\ensuremath{\text{BasicAV}}}
\DeclareMathOperator*{\operAvover}{AV/cost}
\newcommand{\textAVover}{\ensuremath{\text{AV/Cost}}}
\DeclareMathOperator*{\seqPhrag}{\ensuremath{\text{Phragm\'en}}}
\newcommand{\Phragmen}{\ensuremath{\text{Phragm\'en}}}
\newcommand{\Mes}{\ensuremath{\text{MES}}}
\newcommand{\MesApr}{{MES\text{-}Apr}}
\newcommand{\MesCost}{\ensuremath{\text{{MES\text{-}Cost}}}}
\newcommand{\MESP}{\ensuremath{\text{MES\text{-}Cost/Ph}}}
\newcommand{\MESAP}{\ensuremath{\text{MES\text{-}Apr/Ph}}}
\newcommand{\NE}{{\ensuremath{\text{NE}}}}

\newcommand{\textAVoverNE}{\textAVover\text{-}\NE}
\newcommand{\PhragmenNE}{\ensuremath{\Phragmen\text{-}\NE}}

\newcommand{\MesAprNE}{\MesApr\text{-}\NE}
\newcommand{\MesCostNE}{\MesCost\text{-}\NE}
\newcommand{\Plurality}{Plurality}
\newcommand{\PartyList}{Party\text{-}List}

\newcommand{\General}{Unrestricted}
\newcommand{\baseline}{\ensuremath{d}}

\newcommand{\budget}{\ensuremath{B}}
\newcommand{\cost}{\ensuremath{\mathit{cost}}}

\newcommand{\stratprof}{\ensuremath{\mathbf{c}}}
\newcommand{\appstrat}{\ensuremath{\mathbf{ap}}}
\newcommand{\party}{{\ensuremath{\mathit{party}}}}
\newcommand{\reals}{{\mathbb{R}}}
\newcommand{\phtime}{{\ensuremath{\mathit{time}}}}
\newcommand{\response}{{\ensuremath{\mathit{br}}}}
\newtheorem{remark}{Remark}
\newtheorem{observation}{Observation}
\newtheorem{proposition}{Proposition}
\newtheorem{theorem}{Theorem}
\newtheorem{corollary}{Corollary}
\newtheorem{claim}{Claim}
\usepackage{cleveref}
\crefname{claim}{Claim}{Claims}
\DeclareRobustCommand{\abbrevcrefs}{%
  \Crefname{theorem}{Thm.}{Thms.}%
  \Crefname{example}{Ex.}{Exs.}%
  \Crefname{proposition}{Pr.}{Prs.}%
  \Crefname{corollary}{Cor.}{Cors.}%
  \Crefname{claim}{Cl.}{Cls.}
  \Crefname{equation}{Eq.}{Eqs.}
}
\DeclareRobustCommand{\Shcref}[1]{{\abbrevcrefs\Cref{#1}}}

\usepackage{breakcites}
\begin{document}

\maketitle

\begin{abstract} 
  We study strategic behavior of project proposers in the context of
  approval-based participatory budgeting (PB). In our model we assume
  that the votes are fixed and known and the proposers want to set as
  high project prices as possible, provided that their projects get
  selected and the prices are not below the minimum costs of their
  delivery. We study the existence of pure Nash equilibria (NE) in
  such games, focusing on the \textAVover{}, \Phragmen{}, and Method
  of Equal Shares rules.
  Furthermore, we report an
  experimental study of strategic cost selection on real-life PB
  election data.
\end{abstract}

\section{Introduction}

A certain city decided to implement %
participatory
budgeting~\cite{cab:j:participatory-budgeting,goe-kri-sak-ait:c:knapsack-voting,rey-mal:t:pb-survey}.
The city council fixed the budget %
and invited the citizens to propose projects that might be
implemented, letting them know that the funding decisions will be made
by voting. Soon, three groups of activists formed, one focused on
cycling infrastructure, one committed to making the city greener, and
one dedicated to improving the quality of life for senior
citizens. Each group decided to submit a single project, but they
quickly realized that they have many options to choose from. For
example, cycling enthusiasts could propose repainting the markings on
a few bike paths, which would be very cheap, or they could request one
or more new bike paths, which would cost more, or---with the full
available budget---they could greatly improve the whole biking
infrastructure while also offering free bikes for rent. In other
words, they could spend any amount of money on cycling (at least if it
were above a certain necessary minimum) and the more they could get,
the better a project they could offer. Not surprisingly, the other
groups reached similar conclusions. But how much money could they ask
for, realistically?

If the city had a long history of participatory budgeting elections,
the activists could assess the level of support for their causes based
on past experience. In our case, they simply ran a survey. While they
could not ask about the support for specific projects---after all,
they were not sure what to propose yet---they found out that many
people would simply vote for biking-, ecology-, or
seniors-oriented projects, irrespecitve of their costs; they simply
wanted improvements on these fronts (or even several of them).
The survey showed that biking enthusiasts have
the largest vote share, but other groups also had strong support.

With this information in hand, choosing a project's cost largely
depends on the voting rule used. For example, if the city used the
$\textbasicAV$ rule,\footnote{Also known as
  GreedyAV~\cite{boe-fal-jan-kac:c:pb-robustness} or
  GreedCost~\cite{rey-mal:t:pb-survey}.} which greedily selects those
projects that are approved by the largest numbers of voters
 (provided they still fit in the remaining budget),
then the biking enthusiasts
could have their dream project funded.
On the other hand, if the city were to use a proportional rule, such
as the Method of Equal Shares
(\Mes)~\cite{pet-sko:c:welfarism-mes,pet-pie-sko:c:pb-mes} or
$\Phragmen$~\cite{bri-fre-jan-lac:c:phragmen,los-chr-gro:c:phragmen-pb},
then %
each group would have to
analyze how the other ones might act, and choose their strategies
based on that.
Our goal is to analyze this game-theoretic nature of project cost
selection under various participatory budgeting rules (in particular,
our work is
different from that of
\citet{azi-guj-pad-suz-vol:pb-pooling}, who %
consider how voters pool their funds; %
see also the work of \citet{wag-mei:pb:pb-vcg}).

\paragraph{The Game.}
We analyze the following scenario. The sets of projects and voters are
fixed and each voter indicates which projects he or she approves (this
information is common knowledge). %
Each project is controlled by a different proposer choosing its cost
so that it is as high as possible, while ensuring that the project is
selected.\footnote{As in the introductory example, we view the notion
  of a project broadly. By choosing a lower cost means proposing a
  smaller-scale project and choosing a higher cost means a proposing a
  larger-scale one, with possibly extra features and higher quality.}
However, each project also has a certain delivery cost (i.e., the
lowest cost under which it can be reasonably implemented) and the
proposers prefer costs that are at least as high.
Importantly, whether a voter approves a project or not, does not
depend on its cost.
The projects are chosen according to a given 
rule, such as $\textbasicAV$, $\textAVover$, $\Phragmen$, or $\Mes$
(see \Cref{sec:prelim} for details). In other words, we consider a
game where project proposers (or, for simplicity, the projects) are
the players, project costs are their strategies, and costs of selected
projects (minus their delivery costs) are their payoffs. We analyze
whether these games have pure Nash equilibria and, if so, what costs
are reported under these equilibria.
Studying strategic cost selection is also a natural direction in market analysis and
is present, e.g.,\ in the context of the
Hotelling-Downs model~(e.g.,\ in the work of \citet{eiselt2011equilibria}).

Naturally,
our game presents a simplified view of reality. However, due to
abstracting away from excessive complexities, our study led to finding
interesting, and yet technically challenging to identify, basic
phenomena. By this our model lays the foundations and sets expectations
for more sophisticated scenarios. Multiple natural extensions of our
games will become apparent, after we explain below our most important
assumptions and their limitations:

\begin{description}[left=0pt .. 1ex, itemindent=!]
\item[The set of projects is fixed.]  In real-life PB there typically is a
	period of time during which
  one can propose projects and their costs. During this period,
  proposers are not necessarily aware of other submitted projects, but
  if two proposers do realize that their projects are similar, they
  might, for example, merge them or present them as less
	related. By fixing the set of project, our model is not capable of covering
	such dynamics. 
  
\item[Votes are common knowledge.] While analyzing imperfect or probabilistic
	knowledge scenarios is beyond the scope of our model, in reality, neither the
	accurate identities of the voters nor their votes are available to project
	proposers at the time of submitting the projects. At best, one might suspect
	what level of support a project would get based on past experience or 
	polls.

\item[Project costs do not affect approvals.] Our games are not suited to
	reflect that some voters might base their approval decisions on project costs
	and, for example, refuse to vote for projects that are either too cheap or too
	expensive. 
	However, notably, such an effect is debatable and seems to
        depend on a voting rule in
        use~\citep{goe-kri-sak-ait:c:knapsack-voting,gel-goe:c:voting-methods-in-pb}.
        Further, performing the \emph{value-for-money} analysis by
        voters is emprically cognitively costly; see e.g.,\,the
        experimental study of \citet{fai-ben-gal:c:pb-formats}.
  
\end{description}

Relaxing any of our 
assumptions 
would give rise to a more general class of games. As much as such
models are interesting, they go far beyond the the scope of our
paper. For example, allowing the proposers to cooperate, would require
to apply techniques and solution concepts from coalitional game
theory, whereas dropping the common knowledge assumption incurs
modeling the uncertainty of information.  As we initiate the study of
strategic cost selection in participatory budgeting, we feel it is
best to start with a clean, simple model, to obtain a baseline.

\paragraph{Our Contributions.}
We initiate a study of strategic behavior of project proposers in the context of
participatory budgeting. We do so by defining a new class of cost games and
analyzing the existence of pure Nash equilibria under various participatory
budgeting rules and
types of votes (such as plurality ones, where
each voter approves a single project,
or party-list, and unrestricted ones,
which have increasingly more complex structure), under various delivery
costs (considering
the cases in which
they are equal to zero or are
arbitrary), and under various tie-breaking orders.
Among others, we find that $\Mes$ with cost utilities always has equilibria irrespective of
tie-breaking,
as opposed to $\textAVover$, $\Phragmen$, and
$\Mes$ with approval utilities.

We supplement our theoretical study with an experimental analysis of
Nash equilibria in real-life scenarios from
Pabulib~\citep{fal-fli-pet-pie-sko-sto-szu-tal:c:pabulib}. One of our
findings is that $\Phragmen$ and $\Mes$ based on approval utilities
are very different from $\Mes$ based on cost utilities.  The former
two rules lead to a smoother distribution of costs, with many cheap
projects, and the latter to fewer, more expensive projects.

\section{Preliminaries}\label{sec:prelim}

\paragraph{Participatory Budgeting.}
We define a \emph{PB instance} (also sometimes referred to as a \emph{PB election}) as a tuple
$E=(P,V,\budget,\cost)$, where $P=\{p_1, \dots, p_m\}$ is a
set of \emph{projects},
$V= \{v_1, \dots, v_n \}$ is a
set of \emph{voters}, $\budget \in \reals_+$ is the available \emph{budget},
and $\cost \colon P \rightarrow \reals_+$ is a function specifying the
\emph{cost} of each project.  Each voter~$v_i$ casts a nonempty
\emph{approval ballot} $A(v_i) \subseteq P$, containing the set of
projects that he or she \emph{approves} (see, for example, the work
of~\citet{fai-ben-gal:c:pb-formats} for a discussion of other ballot
formats).
Note that a voter may approve projects whose total cost exceeds the
available budget.  We often refer to the voters and their approval
ballots as either an \emph{approval profile} or a \emph{preference profile}.
We extend the $A(\cdot)$ notation so that for a project $p_i$,
$A(p_i)$ is the set of voters that approve $p_i$. Then, $|A(p_i)|$ is
the \emph{approval score} of $p_i$. For each project $p_i$, we assume
that $|A(p_i)| \geq 1$.  Given a subset of projects $P'$, we let
$\cost(P')$ be their total cost, that is,
$\cost(P') = \sum_{p' \in P'}\cost(p')$.

Further, each PB instance comes with a tie-breaking order $\succ$ over
the projects, used by the PB rules to resolve internal ties.
While tie-breaking can have significant influence on 
voting~\cite{obr-elk:c:random-ties-matter,obr-elk-haz:c:ties-matter}
and PB~\cite{boe-fal-jan-kac:c:pb-robustness}, and
can happen fairly often in small
elections~\cite{jan-fal:c:ties-multiwinner}, in large-enough PB
instances we do not expect many of them (see, for example, the work of~\citet{xia:c:probability-of-ties} 
for a discussion of the
likelihood of ties in large ordinal elections).

\paragraph{Participatory Budgeting Rules.}
A \emph{PB rule} is a function~$f$ that for some PB
instance~$E=(P,V,\budget,\cost)$ outputs a set~$f(E) \subseteq P$
of projects, with total cost not exceeding the budget. We refer to
the projects in $f(E)$ as \emph{winning}  (or, alternatively, as
\emph{selected} or \emph{funded}).  Note that our rules are resolute,
that is, their outcomes are unique.  We focus on the following rules, denoting by~$E = (P,V,\budget,\cost)$ the instance we consider:

\begin{description}[left=0pt .. 1ex, itemindent=!]
\item[\textbasicAV{}.] It starts with $W = \emptyset$ and considers
  all the projects in the order of their nonincreasing approval scores
  (with ties broken using $\succ$), inserting a
  considered project $p$ into $W$ if
  $\cost(W \cup \{p\}) \leq \budget$. When done, it outputs~$W$.

\item[\textAVover{}.]
  It is like
  \textbasicAV{}, except that for each project $p$ it computes
  its approval-to-cost ratio~$\nicefrac{|A(p)|}{\cost(p)}$ and %
  considers the projects in the nonincreasing order of these values.

\item[\Phragmen.]
  {\Phragmen} starts with $W = \emptyset$ and fills it as
  follows. Initially, the voters have empty virtual bank accounts, but
  they continuously earn money at the rate of one unit of budget per one unit of
  time. When there is a project $p$ such that the voters
  who approve it have $\cost(p)$ funds and the project is within the
  remaining budget (that is, $\cost(W \cup \{p\}) \leq \budget$), these
  voters \emph{purchase} it, that is, $p$ is included in $W$, the bank
  accounts of voters in $A(p)$ are reset to zero, and $p$ is removed
  from consideration. If $\cost(W \cup \{p\}) > \budget$, then~$p$ is
  removed from consideration without being included in~$W$. If
  several projects
  could be purchased simultaneously, the
  rule picks one using the tie-breaking order. The rule stops and
  outputs $W$ when all projects are removed from
  consideration.\footnote{As we assume that each project is approved
    by at least one voter, upon completion there is no unfunded
    project that could be included in the outcome without exceeding
    the budget.}

\item[Method of Equal Shares (\MesCost).]
    
	First, each voter receives the same
	amount~$\nicefrac{\budget}{|V|}$ of money.
  Then we let $W = \emptyset$ and proceed iteratively: Within each
  iteration, for each project $p$ not in $W$ we compute its
  affordability coefficient $\alpha_p$ as the smallest number such
  that the following holds ($b_i$ is the money that voter $v_i$
  currently has):
  \begin{equation}\label{eq:mes}
    \textstyle
    \sum_{v_i \in A(p)} \min( b_i, \alpha_p \cdot \cost(p)) = \cost(p).
  \end{equation}
  If no such value exists (that is, the voters approving $p$ cannot
  afford it) then we set $\alpha_p = \infty$.  If $\alpha_p = \infty$
  for all projects not in $W$, then
  we terminate and output~$W$. Otherwise, we choose project $p'$ with
  the lowest affordability coefficient (using the tie-breaking order,
  if needed), include~$p'$ in~$W$, and take
  $\alpha_{p'} \cdot \cost(p')$ money from %
  each voter in $A(p')$ (or all the remaining funds, if the voter had
  less than $\alpha_{p'}\cdot \cost(p')$).  Due to the use of
  $\cost(p)$ in Equation~\eqref{eq:mes}, this variant of $\Mes$ is
  also called \emph{$\Mes$ with cost utilities}.\footnote{We note that
    under the current definition $\alpha_p$ is unique and so finding
    the minimum such value is not needed. However, we maintain this
    formulation for consistency with the literature.}

\item[{\Mes} with Approval Utilities (\MesApr).] This rule, which
  was introduced by the same authors as \MesCost{}, works in the same
  way, except that it replaces Equation~\eqref{eq:mes} with:
  \begin{equation}
    \textstyle
    \sum_{v_i \in A(p)} \min( b_i, \alpha_p ) = \cost(p).
  \end{equation}
  So, while in $\MesCost$ the affordability coefficients are
  between~$0$ and $1$, under \MesApr~they can be as large as
  $\budget$.

\end{description}

\begin{remark}\label{rem:mes}
  Both considered variants of $\Mes$ %
  might output a set of projects that can be extended without exceeding
  the budget. Following \citet{pet-pie-sko:c:pb-mes}, in our
  experiments we use \emph{\Phragmen{} completion}: When a $\Mes$ variant %
  finishes, we extend its output by running
  $\Phragmen$ with voters' bank accounts initiated with their then-current amounts of money.
  This defines the $\MESP$ and $\MESAP$ rules. %
\end{remark}

Many authors also study other PB rules (see,
for example, the works of~\citet{goe-kri-sak-ait:c:knapsack-voting,fal-tal:c:pb-framework,sre-bha-nar:c:pb-maximin});
for an overview we point to the works of \citet{rey-mal:t:pb-survey}
and \citet{lac-sko:b:approval-survey} (the latter regards multiwinner
elections, where projects have unit costs).
$\textbasicAV$ is commonly used in practice, $\MesCost$ also was
recently used by several cities (see https://equalshares.net),
whereas the other
rules are mostly studied theoretically.

\paragraph{Special Approval Profiles.}
From time to time we focus on \emph{plurality} profiles, where each voter
approves exactly one project, and on \emph{party-list} profiles, where
the projects are grouped into ``parties'' and each voter approves all
projects from a single party. Some cities require plurality profiles
(such as Wrocław, Poland; see the datasets in
Pabulib~\citep{fal-fli-pet-pie-sko-sto-szu-tal:c:pabulib}) whereas the
party-list ones are interesting theoretically.

\begin{definition}\label{def:partylist}
  Consider a set of projects~$P$ and a voter collection~$V$ with
  approval ballots over~$P$. We say that these voters have: (1)
 \emph{plurality preferences}, if $|A(v_i)| = 1$ for each voter $v_i$, (2) \emph{party-list preferences}, if either $A(v_i) = A(v_j)$ or
    $A(v_i) \cap A(v_j) = \emptyset$ for each two voters $v_i$ and $v_j$.

\end{definition}
For a party-list profile and a project $p$, by $\party(p)$ we mean the
set of projects approved by the same voters as $p$.
\Cref{fig:typesex} illustrates the types of preferences.

\begin{figure}[t]
  \centering
  \scalebox{0.4}{\begin{tikzpicture}
      
    \draw (0,0) rectangle (1,1);
    \draw (1.2,0) rectangle (2.2,1);
    \draw (2.4,0) rectangle (4.4,1);
    \draw[dotted, thick] (0.5,-0.2) -- (0.5,1.2);
    \draw[dotted, thick] (1.7,-0.2) -- (1.7,1.2);
    \draw[dotted, thick] (2.9,-0.2) -- (2.9,1.2);
    \draw[dotted, thick] (3.9,-0.2) -- (3.9,1.2);
    \node[anchor=north, fontscale=2] at (0.5,-0.2) {\Huge $v_1$};
    \node[anchor=north, fontscale=2] at (1.7,-0.2) {\Huge $v_2$};
    \node[anchor=north, fontscale=2] at (2.9,-0.2) {\Huge $v_3$};
    \node[anchor=north, fontscale=2] at (3.9,-0.2) {\Huge $v_4$};
    \node[anchor=east, fontscale=2] at (0.8,1.4) {\Huge $p_1$};
    \node[anchor=east, fontscale=2] at (1.9,1.4) {\Huge $p_2$};
    \node[anchor=east, fontscale=2] at (3.7,1.4) {\Huge $p_3$};
  \end{tikzpicture}}
  ~
  \hspace{2em}
  \scalebox{0.4}{\begin{tikzpicture}
      
    \draw (0,0) rectangle (1.9,1);
    \draw (2.1,0) rectangle (3.1,1);
    \draw (0,1.5) rectangle (1.9,2.5);
    
    \draw[dotted, thick] (0.5,-0.2) -- (0.5,2.7);
    \draw[dotted, thick] (1.5,-0.2) -- (1.5,2.7);
    \draw[dotted, thick] (2.6,-0.2) -- (2.6,1.2);
    \node[anchor=north, fontscale=2] at (0.5,-0.2) {\Huge $v_1$};
    \node[anchor=north, fontscale=2] at (1.5,-0.2) {\Huge $v_2$};
    \node[anchor=north, fontscale=2] at (2.6,-0.2) {\Huge $v_3$};
    \node[anchor=east, fontscale=2] at (0,0.5) {\Huge $p_1$};
    \node[anchor=east, fontscale=2] at (0,2) {\Huge $p_2$};
    \node[anchor=east, fontscale=2] at (4.05,0.5) {\Huge $p_3$};

  \end{tikzpicture}}
  ~
  \hspace{2em}
  \scalebox{0.4}{\begin{tikzpicture}
      
    \draw (0,0) rectangle (3,1);
    \draw (0,1.5) rectangle (1,2.5);
    \draw (2,1.5) rectangle (3,2.5);
    \draw[dotted, thick] (0.5,-0.2) -- (0.5,2.7);
    \draw[dotted, thick] (1.5,-0.2) -- (1.5,1.2);
    \draw[dotted, thick] (2.5,-0.2) -- (2.5,2.7);
    \node[anchor=north, fontscale=2] at (0.5,-0.2) {\Huge $v_1$};
    \node[anchor=north, fontscale=2] at (1.5,-0.2) {\Huge $v_2$};
    \node[anchor=north, fontscale=2] at (2.5,-0.2) {\Huge $v_3$};
    \node[anchor=east, fontscale=2] at (0.0,0.5) {\Huge $p_3$};
    \node[anchor=east, fontscale=2] at (0.0,2) {\Huge $p_1$};
    \node[anchor=east, fontscale=2] at (2.0,2) {\Huge $p_2$};
  \end{tikzpicture}}
\caption{Examples of plurality (left), party list (middle), and
  unrestricted (right) preferences. Projects are depicted as
  boxes. Each voter approves those projects that are drawn directly
  above.}
\label{fig:typesex}
\end{figure}
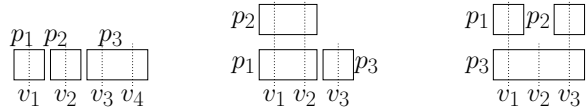

\section{Participatory Budgeting Cost Games}

In this section we define our main object of study,
\emph{participatory budgeting cost games} (\emph{PB games}).  Formally, a
PB game is a tuple~$(P,V,\budget,\baseline)$, where $P$ is a set of
projects, $V$ is a collection of voters with approval preferences over
the projects from $P$, $\budget$ is the available budget, and
$\baseline \colon P \rightarrow \reals_+$ is a function that
associates each project with its minimal \emph{delivery} cost (see the
description of the payoffs in the next paragraph). We assume that for each
project $p_i$, $\baseline(p_i) \leq \budget$.
In this game, the projects are the players and each of them needs to
report its cost. That is, a \emph{strategy profile} is a tuple
$\stratprof=(c_1, \dots, c_n)$, with a cost $c_i \in \reals_+$ for each
project $p_i$. As is standard, we write $(\stratprof_{-i},c')$ to
denote a strategy profile that is identical to~$\stratprof$ except
that project $p_i$ reports cost $c'$. We often
write $\stratprof(p_i)$ to denote the cost reported by $p_i$ under
strategy profile~$\stratprof$ (and we use strategy profiles as $\cost$
functions in PB instances).

Let us fix a PB rule $f$ and a PB game
$G = (P,V,\budget,\baseline)$. For a strategy profile $\stratprof$,
the associated PB instance is
$E(\stratprof) = (P,V,\budget,\stratprof)$ and the \emph{payoff} of
each project $p_i \in P$ is:
\[
u_{i}(\stratprof) = \begin{cases}
  \stratprof(p_i) - \baseline(p_i) & \text{ if $p_i \in f\big(E(\stratprof)\big)$,}\\
  0                             & \text{ otherwise}.
\end{cases}
\]
The interpretation %
is as follows: Each project $p$ has a minimal delivery cost
$\baseline(p)$, which is the lowest price under which  it
can be implemented. Yet, with more money the project can be better and the proposer
prefers this.
Our utility costs are defined to implement this preference. In
particular, the proposers do not receive funds equal to the utility
values.
As a special case, we %
often consider %
delivery costs %
equal to zero.

For a PB rule $f$ and a PB game $G = (P,V,\budget,\baseline)$, we
are interested whether this game has Nash equilibria (NE).  We say that a strategy profile $\stratprof$ for this game is a
\emph{Nash equilibrium} under $f$ (is an $f$-NE) if for every
project $p_i \in P$ and every cost $c' \in \reals_+$ it holds that
$u_{i}(\stratprof) \geq u_i(\stratprof_{-i},c')$. In other words, no project has
a \emph{profitable deviation}, that is, it cannot benefit by reporting a
different cost than that chosen in an $f$-NE (as long as no other project
changes its reported cost).

\begin{example}\label{ex:initial}
  Take $\textAVover$ and a PB game with projects
  $p_1$, $p_2$, voters $v_1, \ldots, v_5$, budget $10$, 
  and zero delivery costs.
  We have $A(p_1) = \{v_1,v_2\}$ and $A(p_2) = \{v_3,v_4,v_5\}$, so
  this is a plurality profile. Then, the strategy profile~$\stratprof$ where
  $\stratprof(p_1) = 4$ and $\stratprof(p_2) = 6$ is an NE. Indeed,
  according to $\textAVover$, under these costs $p_1$ and $p_2$ are
  tied, because
  $\nicefrac{|A(p_1)|}{\stratprof(p_1)} =
  \nicefrac{|A(p_2)|}{\stratprof(p_2)} = \nicefrac{1}{2}$, and both
  projects are selected under any tie-breaking order. If either of
  the projects reports a lower cost, it gets selected, but its utility
  decreases. If it %
  reports a higher cost,
  it is not selected and its utility drops to zero.
\end{example}

From time to time we use specific \emph{approval-to-delivery-cost}~(A/D)
tie-breaking orders, which favor projects with larger ratios of their
approvals to delivery costs. So, $\succ$ is A/D if projects with zero
delivery costs are ranked highest, while
$\nicefrac{A(p_i)}{\baseline(p_i)} >
\nicefrac{A(p_j)}{\baseline(p_j)}$ implies $p_i \succ p_j$ if
$\baseline(p_i) \cdot \baseline(p_j) >0$.  Since some projects may
have equal approval-to-delivery cost ratios, there may be several
different A/D tie-breaking orders in a given setting.

\begin{table}
  \newcommand{\foralltbs}{\raisebox{1px}{\ensuremath{\blacksquare}}}
	\newcommand{\forsometbs}{{\color{gray!40}\raisebox{1px}{\ensuremath{\blacksquare}}}}
	\newcommand{\noforsometbs}{\raisebox{1px}{\ensuremath{\circledcirc}}}
	\newcommand{\noforalltbs}{\raisebox{1px}{\ensuremath{\square}}}
 \centering \setlength{\tabcolsep}{2.75pt}
        \scalebox{.8}{\begin{tabular}{r|cccccc} ballots &
            \multicolumn{2}{c}{\Plurality} &
            \multicolumn{2}{c}{\PartyList}
                        & \multicolumn{2}{c}{\General}\\

		& $\baseline \equiv 0$ & $\textrm{arb. } \baseline$
		& $\baseline \equiv 0$ & $\textrm{arb. } \baseline$
		& $\baseline \equiv 0$ & $\textrm{arb. } \baseline$\\\midrule
		
		\textbasicAV{} & \foralltbs{} & \foralltbs{} 
                & \foralltbs{} & \foralltbs
                & \foralltbs{} & \foralltbs\\[-4pt]

								& \tiny\Shcref{prop:gc-ne-always-exists}
								& \tiny\Shcref{prop:gc-ne-always-exists}
								& \tiny\Shcref{prop:gc-ne-always-exists}
								& \tiny\Shcref{prop:gc-ne-always-exists}
								& \tiny\Shcref{prop:gc-ne-always-exists}
								& \tiny\Shcref{prop:gc-ne-always-exists}\\[1ex]

		\textAVover{} & \foralltbs
									& \forsometbs{} \noforsometbs{}
                  & $\foralltbs^\dagger$
								  & \forsometbs{} \noforsometbs{} 
                  & $\foralltbs^\dagger$
									& \forsometbs{} \noforsometbs{}  \\[-4pt]

								&\tiny\Shcref{prop:GC-zerobase-NE} 
                &\tiny\Shcref{thm:avcost-ne-exists}, \tiny\Shcref{ex:noSeqPhrag}
								&\tiny\Shcref{prop:GC-zerobase-NE}
								&\tiny\Shcref{thm:avcost-ne-exists}, \Shcref{ex:noSeqPhrag}
								&\tiny\Shcref{prop:GC-zerobase-NE}
								&\tiny\Shcref{thm:avcost-ne-exists}, \tiny\Shcref{ex:noSeqPhrag}\\[1ex]

                \Phragmen
                & $\foralltbs^\dagger$
							  & \forsometbs{} \noforsometbs{}
                & \foralltbs{}  
								& \forsometbs\hspace{-5px}?\hspace{3px}\noforsometbs{} 
                & \noforalltbs{}
								& \noforalltbs{} \\[-4pt]

	            	&\tiny\Shcref{cor:exclusiveNo0}
								&\tiny\Shcref{cor:exclusiveNo0}
								&\tiny\Shcref{result:phragparty}
								&\tiny\Shcref{cor:exclusiveNo0}
                &\tiny\Shcref{pro:phrag-laminar-no-ne}
								&\tiny\Shcref{pro:phrag-laminar-no-ne}
                \\[1ex]
        
		\MesApr  & \foralltbs{} 
						 & \foralltbs{} 
						 & \foralltbs{} 
						 & \forsometbs{} \noforsometbs{} 
						 & \noforalltbs{} 
						 & \noforalltbs{}\\[-4pt] 

             & \tiny\Shcref{thm:mesAprNEPlurality} 
             & \tiny\Shcref{thm:mesAprNEPlurality}
             & \tiny\Shcref{thm:mesAprNEPartyListZeroDeliveryCost}
						 & \tiny\Shcref{thm:mesAprNEPartyListDeliveryCostSpecificTBOrder}, \Shcref{ex:mesAprNoNEPartyList}  
						 & \tiny\Shcref{pro:mesAprNoNEGeneralRegardlessOfTB}
						 & \tiny\Shcref{pro:mesAprNoNEGeneralRegardlessOfTB}\\[1ex] 
  
  	\MesCost & \foralltbs
						 & \foralltbs
						 & \foralltbs
		         & \foralltbs
		         & \foralltbs
		         & \foralltbs\\[-4pt]

  	& \tiny\Shcref{thm:mesCostNEGeneral}
						 & \tiny\Shcref{thm:mesCostNEGeneral}
						 & \tiny\Shcref{thm:mesCostNEGeneral}
						 & \tiny\Shcref{thm:mesCostNEGeneral}
						 & \tiny\Shcref{thm:mesCostNEGeneral}
						 & \tiny\Shcref{thm:mesCostNEGeneral}

                      \end{tabular}}

	\caption{
          Our results regarding the existence of Nash equilibria.
          By ``${\baseline \equiv 0}$''
          and ``${\normalfont arb. }\ {\baseline}$'' we mean zero and
          arbitrary delivery costs, respectively.
          Symbols~\foralltbs{} and \forsometbs{} indicate that an NE
          always exists for an arbitrary or
					an A/D
          tie-breaking order, respectively. Symbol~\noforalltbs{} marks cases for
          which there is a game without NE for all tie-breaking
          orders.  \noforsometbs{} means that
          there is a game and a tie-breaking order for which there is
          no NE (we omit it if \noforalltbs{} applies). Results
          labeled with~$\dagger$ extend from zero delivery costs to the case where  for each project~${p}$ we have ${\baseline(p) \leq \appstrat(p)}$, where
          ${\appstrat}$ is the corresponding approval-proportional profile. ``?'' means a conjecture.%
        }                   \label{tab:results}%
\end{table}

\section{Existence of Equilibria}

Our theoretical results, summarized in \Cref{tab:results}, focus on the
existence of NE in PB games.
We mostly show that either:
(1)~There always is an equilibrium, for any tie-breaking order;
(2)~There always is an equilibrium for some %
A/D
tie-breaking order; or
(3)~There is a PB game with no equilibrium for any
tie-breaking order.
When such statements are difficult to
show, we seek instances with no equilibria %
for some tie-breaking order. Such an observation is weaker
than %
a result of type~(3), is incomparable to a result of type~(2), and
implies that a result of type~(1) does not exist.  Pragmatically, it
is most important to know whether a result of type~(1) holds or not in
a given setting; the other types of results indicate various levels of
bad news that we may be facing regarding the existence of Nash
equilibria. %

We are interested in strategy profiles that are proportional in some
way (for a discussion of fairness in PB see, for example, the works of
\citet{bri-for-lac-mal-pet:c:pb-proportionality} and
\citet{mal-rey-end-lac:c:fair-pb}, as well as the works defining
$\Mes$ and $\Phragmen$). Below we define a basic type of such a
profile, where project costs are proportional to the numbers of
approvals they get.

\begin{definition}[Approval-Proportional
  Strategy Profile]\label{def:cgwiazdka}
  Consider a PB game %
  with projects~$P = \{p_1, p_2, \ldots, p_m\}$. We say that strategy
  profile~$\appstrat$ is \emph{approval-proportional} if for each project $p$
  it holds that
  $ \appstrat(p) = \budget \cdot \frac{|A(p)|}{ \sum_{p' \in P}
    |A(p')|}$.
\end{definition}

Under plurality preferences, the approval-proportional profile is a
natural solution for rules focused on proportionality. For
more involved settings, other profiles may be more appropriate (see,
for instance, $\MesCost$ and \Cref{thm:mesCostNEGeneral}).

\subsection{\textbasicAV}

We start our analysis by considering $\textbasicAV$.
There is a simple equilibrium where the project with most
approvals (preferred in the tie-breaking order, if there is more than
one such project) requests the full budget amount.

\begin{proposition}\label{prop:gc-ne-always-exists}
  For each PB game %
  and each internal
  tie-breaking order $\succ$, there is a $\textbasicAV$-NE where the project
  with most approvals (best with respect to $\succ$) 
  reports cost $\budget$.%
\end{proposition}

\appendixproof{Proposition}{prop:gc-ne-always-exists}{%
\begin{proof}
  Take a PB game and a strategy profile $\stratprof$ as described
  in the statement, and let $p$ be the project such that
  $\stratprof(p) = \budget$. If a project other than $p$ increases
  its cost, it still is not selected so this is not a beneficial
  move. If $p$ decreases its cost then its utility drops, and if it
  increases its cost then it is not funded. In either case, its
  utility decreases. So, $\stratprof$ is a $\textbasicAV$-NE.
\end{proof}
}

In fact, we can additionally observe that the most popular project reports the cost
$\budget$ (and is selected) under every \textbasicAV-$\textit{NE}$. Hence,
\Cref{prop:gc-ne-always-exists} indicates that proposers of
projects  receiving many approvals have incentive to overstate their
costs. Furthermore, proposers of smaller, less popular projects,
have an incentive to bundle them together (indeed, in some cities one can
observe projects of the form ``improvement of infrastructure in areas
$X$, $Y$, and $Z$,'' which suggests such bundling).  So,
unsurprisingly, $\textbasicAV$ equilibria in our games are 
disproportional, highlighting known deficiencies of this rule.

\subsection{{\textAVover}}

Some %
shortcomings of  $\textbasicAV$ are %
remedied by $\textAVover$, for which
the approval-proportional strategy is an equilibrium.  To show this,
we note that in an NE, if the delivery costs are limited, %
all projects are funded, using the whole budget. Then, $\appstrat$ is the only equilibrium profile.

\begin{proposition}\label{pro:NE-uses-full-budget}
  Take a PB game %
  and the corresponding
  strategy profile $\appstrat$.  If $\baseline(p) \leq \appstrat(p)$ for each project
  $p$, then if a
  profile $\stratprof$ is an $\textAVover$-NE for this game, then
  (1)~$\sum_{p \in P} \stratprof(p) = \budget$, and (2) $\textAVover$
  funds all the projects.%
\end{proposition}
\newcommand{\won}{{\mathit{won}}}
\newcommand{\lost}{{\mathit{lost}}}

\begin{proof}[Proof Sketch]
Let each project $p$ have $\baseline(p) \leq \appstrat(p)$ and let~$\stratprof$
be a Nash equilibrium. Then, it must be that $\sum_{p \in P}\stratprof(p) \geq \budget$. Let
$\sum_{p \in P}\stratprof(p) > \budget$. As $\stratprof$ is an equilibrium, all
of the funded projects are chosen with the same score and their costs amount to
$\budget$, while being higher than those reported in $\appstrat$. So, some
non-funded project $p$ could be funded with the cost higher than
$\baseline(p)$\,---\,contradiction. Hence, $\sum_{p \in P} \stratprof(p) =
\budget$. Point~$(2)$ follows from the exhaustivity of $\textAVover$.
\end{proof}
\appendixproof{Proposition}{pro:NE-uses-full-budget}{%
\begin{proof}
  Let the notation be as in the statement of the proposition. In the beginning
  we observe that if $\sum_{p \in P}\stratprof(p) < \budget$, then
  $\textAVover$ selects all the projects. Hence, the project
  considered last benefits by reporting a higher cost (so that the sum
  of the reported costs is $\budget$). Thus such a
  strategy profile~$\stratprof$ is not an equilibrium.

  Next, let us assume that $\sum_{p \in P}\stratprof(p) >
  \budget$. Let $P_\won$ and $P_\lost$ be sets of projects that,
  respectively, are and are not funded. By our assumption, we know
  that $P_\lost$ must be nonempty, and we also note that $P_\won$ must be
  nonempty.
  Naturally, $\sum_{p \in P_\won}\stratprof(p) \leq \budget$.
  For each two projects
  $p_i$ and $p_j$ in $P_\won$, it must be the case that
  $\nicefrac{|A(p_i)|}{\stratprof(p_i)} =
  \nicefrac{|A(p_j)|}{\stratprof(p_j)}$.  For example, if we had
  $\nicefrac{|A(p_i)|}{\stratprof(p_i)} >
  \nicefrac{|A(p_j)|}{\stratprof(p_j)}$ then $\textAVover$ would
  consider $p_i$ prior to~$p_j$ and, consequently, it would be
  beneficial for $p_i$ to increase its reported cost by a small-enough
  amount so that it would still be considered (and selected) prior to
  $p_j$. This would contradict the assumption that $\stratprof$ is an
  equilibrium.
  Next, since $\sum_{p \in P}\appstrat(p) = \budget$, there must be
  some project $p'$ in $P_{\won}$ such that
  $\stratprof(p') > \appstrat(p') \geq \baseline(p')$ (otherwise, if each project $p' \in P_\won$ reported at most $\appstrat(p')$, then any project $p \in P_\lost$ would be selected after reducing its cost to $\appstrat(p)$, so \stratprof{} could not have been \NE{}). Further, by
  definition of the approval-proportional profile, for every project
  $p \in P$, we have that $\nicefrac{|A(p)|}{\appstrat(p)}$ is the
  same value and, so, for every project $q \in P_\lost$ it holds that
  $\nicefrac{|A(q)|}{\appstrat(q)} = \nicefrac{|A(p')|}{\appstrat(p')}
  > \nicefrac{|A(p')|}{\stratprof(p')}$.  This means that every
  project $q \in P_\lost$ can improve its utility by reporting cost
  $\appstrat(q)$ and being selected. This contradicts the assumption
  that $\stratprof$ is an equilibrium.

  Thus it must be the case that
  $\sum_{p \in P}\stratprof(p) = \budget$, which implies that
  $\textAVover$ selects all projects. 
\end{proof}%
}

\begin{proposition}\label{prop:GC-zerobase-NE}
  Irrespective of a tie-breaking order, 
  $\appstrat$ is the only $\operAvover$-NE if  
  $\baseline(p) \leq \appstrat(p)$ for all
  $p \in P$.
\end{proposition}

\appendixproof{Proposition}{prop:GC-zerobase-NE}{%
\begin{proof}
	Consider a PB game~$(P, V, \budget, \baseline)$ such that for every~$p_i \in P$ it is true that~$\baseline(p) \leq \appstrat(p)$.

	We argue that strategy~$\stratprof = \appstrat$
    is a \textAVover-NE.
    Note that the sum of the reported costs in~$\stratprof$ is
    exactly~$\budget$ and all projects are funded. Hence,
    for each player, decreasing the cost leads 
    to a utility loss. Hence, a profitable deviation could only be
    through increasing the reported cost.
    Towards contradiction, let us fix some player~$p_i \in
	\appstrat$ and assume that reporting a cost~$c'_i > c_i$ leads to a better payoff.
	This means that $p_i$ is funded in the modified election~$E' = (P, V, \budget,(\stratprof_{-i}, c_i') )$
 , that is, $p_i \in \operAvover(E')$.
	Additionally, since by \Cref{pro:NE-uses-full-budget} we know that~$\sum_{j \in |P|} c_j = \budget$, it
	immediately follows that $c_i' + \sum_{j \in |P|\setminus\{i\}} c_j >
	\budget > \sum_{j \in |P|\setminus\{i\}} c_j$. It is also the case that for
	each~$j \in |P|$, $\nicefrac{|A(P_j)|}{c_j} > \nicefrac{|A(P_i)|}{c'_i}$.
	Hence, in particular, in the run of~$\operAvover$ for~$E'$, all projects are
	considered before $p_i$. Since all of these project are selected, by the time the
	procedure starts considering~$p_i$, the remaining budget is smaller
	than~$c_i'$, so $p_i \not\in \operAvover(E')$; a contradiction. Note that because
	all projects that are funded are always considered by the procedure in the
	same time, the result is independent of the tie-breaking order~$\succ$.

	We now argue that every strategy profile $\stratprof{}$ other than
	$\appstrat$ is not a $\operAvover$-NE. To this end, we take such a profile
	$\stratprof{}$ and assume towards contradiction that it is a $\operAvover$-NE.
	Observe that since $\stratprof{} \neq \appstrat$,
	there exist two distinct projects~$p_i$ and~$p_j$ such that
	$\nicefrac{|A(p_i)|}{c_i} \neq \nicefrac{|A(p_j)|}{c_j}$. Hence one of the
	projects is considered earlier than the other; we assume without loss of
	generality that $p_i$ is considered before~$p_j$.
	Notably, as \stratprof{} is (by our assumption) a $\operAvover$-NE, by
	\Cref{pro:NE-uses-full-budget}, we know that $\{p_i, p_j \} \subseteq \operAvover(\stratprof{})$.
	This implies, however, that $p_i$ can report a slightly higher price~$c_i'$
	and still be selected by~$\operAvover{}$ thus obtaining a better payoff. Formally,
	there exists a $c_i' > c_i$ resulting in the profile $(c_i',
	\stratprof{}_{-i})$ in which~$\nicefrac{|A(P_i)|}{c_i'} <
	\nicefrac{|A(P_j)|}{c_j}$. So, $p_i$ is (still) considered by $\operAvover$
	before~$c_j$ and the deviation is small enough (that is, $c_i'-c_i < c_j$)
	to guarantee that $p_i\in \operAvover(( \stratprof{}_{-i}, c_i'))$. So,
	\stratprof{} is not a $\operAvover$-NE; a contradiction. Again, due
	to~\Cref{pro:NE-uses-full-budget}, in each Nash equilibrium, all projects are
	selected to be funded. Hence, our proof works for every possible tie-breaking.
\end{proof}%
}

If (some of) the delivery costs are higher than the costs implied by
the approval-proportional strategy, then the
situation becomes more complicated. As we show, it is no longer true that an 
NE exist for all internal tie-breaking orders.

\begin{example}\label{ex:noSeqPhrag}
  Take a PB game~$(P, V, \budget, \baseline)$ with plurality
  ballots, where $P=\{p_1, p_2 \}$, $|A(p_1)| = |A(p_2)|= 5$,
  $\budget = 10$, $\baseline(p_1) = 0$,
  $\baseline(p_2) = 6$, and the tie-breaking order is $p_2 \succ p_1$.
  Assume that $\stratprof$ is an
	\textAVover{}-NE. Consider two cases:
  (1) If $\stratprof(p_2) \geq 6$, then either 
    (a)~$\stratprof(p_1) \geq \stratprof(p_2)$, where it is
    better for $p_1$ to report a cost lower than $\stratprof(p_2)$
    (otherwise $p_2$ is selected before $p_1$ and there is no budget
    left for $p_1$), or
    (b)~$\stratprof(p_1) < \stratprof(p_2)$ and it is better for
    $p_1$ to increase its cost (keeping it below
    $\stratprof(p_2)$).
  (2) If $\stratprof(p_2) < 6$, then either $p_2$ is chosen and it
    prefers to report cost~$6$ (otherwise its
    payoff is $\stratprof(p_2)-\baseline(p_2) < 0$), or $p_2$ is
    not selected, meaning that
    $\stratprof(p_1) < \stratprof(p_2)$ and it is better for
    $p_1$ to increase its cost (keeping it below
    $\stratprof(p_2)$).
    Yet, if $p_1 \succ p_2$, then $(6,6)$ %
    is an equilibrium, 
    where only $p_1$ is funded: Both projects would decrease their
    payoff by lowering costs and they would not benefit by increasing
    their costs ($p_1$ %
    would no longer be selected, and $p_2$'s payoff would not change).
\end{example}

Yet, there always is a tie-breaking order yielding an NE. %

\begin{theorem}\label{thm:avcost-ne-exists}
	For each PB game %
	and A/D tie-breaking order
    there is an~\textAVoverNE{} %
		computable in polynomial time.
\end{theorem}
\begin{proof}[Proof Sketch]
  The algorithm that we propose computes the claimed profile~$\stratprof$. 
	It runs in iterations. Starting from the situation
	where each project reports its delivery cost, in each iteration the algorithm
	selects a group of projects that ``underreport'' their costs compared to their support. Then, it increases their reported costs as
	much as possible, with the projects remaining funded. Finally, if the
	funded projects after the update do not reach the budget, the procedure is
	repeated for the remaining projects whose prices have not been increased so
	far. These projects, albeit worse regarding the delivery-cost-to-support
	ratio, can still benefit from the increase. In the end the
	procedure outputs an equilibrium profile.  The full proof is in the appendix.
\end{proof}

\appendixproof{Theorem}{thm:avcost-ne-exists}{%
  \begin{algorithm}[t]\caption{Finding a Nash equilibrium for $\operAvover$.%
  }\label{alg:NE-existence}
  \ifdefempty{\algorithmtwoeFlavor}{%
   \SetCommentSty{scriptsize}
  	\KwData{set~$P = \{p_1, p_2, \ldots, p_m\}$ of projects with the delivery costs
  	function~$\baseline{}$, set~$V$ of voters with their ballots,
    budget~$\budget$, A/D tie-breaking~$\succ$.}
    \KwResult{profile~$\stratprof$
  		 that is a $\operAvover$-NE
  	 under tie-breaking~$\succ$.}
    \BlankLine
  
    $\budget^* \leftarrow \budget$\tcp*{remaining budget}
  
    $\stratprof \leftarrow (\baseline(p_1), \baseline(p_2), \ldots,
		\baseline(p_m)) = \stratprof_0$ \label{alg:init-strategy}
    \tcp*{initial strategy}
  	\tcp{let $t(p)$, for each $p \in P$, be $\nicefrac{\baseline(p)}{A(p)}$; } 
  
    \tcp{note that $\succ$ is nondecreasing w.r.t.\, values of $t$}
   
		\tcp{$\pos_X(i)$ denotes the top $i$th project of order~$\succ$ restricted
		to~$X \subseteq P$} 
  	$P_\textrm{p} \leftarrow P$\tcp*{prospective
  	projects}\label{alg:init-prospective}
  
  	\While(\tcp*[f]{are there projects to consider?}){$P_\textrm{p} \neq \emptyset$}{
  		$P' \leftarrow P_\textrm{p}$\tcp*{$P'$ is a helper variable}
  		$\baseline(\pos_{P'}(0)) \leftarrow -\inf$\tcp*{guardian ``fake'' value}
  		$\baseline(\pos_{P'}(|P'|+1)) \leftarrow +\inf$\tcp*{guardian ``fake'' value}
  		$k \leftarrow$ maximum integer $x \leq |P'|$ such that
  		$t(\pos_{P'}(x))\sum_{i \leq x-1} A(\pos_{P'}(i))
  		\leq \budget^* < t(\pos_{P'}(x+1))\sum_{i \leq x+1} A(\pos_{P'}(i))$\label{alg:selection-k}\;
  		$T \leftarrow$ maximum $x \leq
  		\nicefrac{\baseline(\pos_{P'}(k+1))}{|A(\pos_{P'}(k+1))|}$ such that
  		$x\sum_{i \in [k]} A(\pos_{P'}(i)) \leq \budget^*$\label{alg:computing-t}\;
  		\ForEach(\tcp*[f]{update~$\stratprof{}$ and $\budget^*$}){$i \leq k$}{
  			$c(\pos_{P'}(i)) \leftarrow T \cdot A(\pos_{P'}(i))$%
  			\;\label{alg:setting-t}
    	  $\budget^* \leftarrow \budget^* - c(\pos_{P'}(i))$%
  			\;\label{alg:updating-budget}
  			$P_\textrm{p} \leftarrow P_\textrm{p} \setminus
  			\{\pos_{P'}(i)\}$\;\label{alg:updating-Pp}
      }
  		\lIf{$k < |P'|$}{$P_\textrm{p} \leftarrow P_\textrm{p} \setminus \{\pos_{P'}(k+1)\}$}
  		$P_\textrm{p} \leftarrow \{p \in P_\textrm{p}: c(p) \leq
  		\budget^*\}$\;\label{alg:too-exp}
    }
  	\Return{$\stratprof$}\;
  
   }%
   {
  \renewcommand{\algorithmiccomment}[1]{\hfill {\scriptsize // #1}}
  \renewcommand{\algorithmicrequire}{\textbf{Data:}}
  \renewcommand{\algorithmicensure}{\textbf{Result:}}
  \begin{algorithmic}[1]
    \REQUIRE set~$P = \{p_1, p_2, \ldots, p_m\}$ of projects with the delivery costs
  	function~$\baseline{}$, set~$V$ of voters with their ballots,
    budget~$\budget$.
    
    \ENSURE profile~$\stratprof$ and tie-breaking~$\succ$ such that $\stratprof$
    	is both a $\seqPhrag$-NE and a $\operAvover$-NE under tie-breaking~$\succ$.
    
    \STATE  $\budget^* \gets \budget$ \COMMENT{remaining budget}
  
    \STATE $\stratprof \gets (\baseline(p_1), \baseline(p_2), \ldots,
  	\baseline(p_m)) = \stratprof_0$
    \COMMENT{initial strategy}
  
    \COMMENT{let $t(p)$, for each $p \in P$, be $\nicefrac{\baseline(p)}{A(p)}$\hfill}
   
    \STATE $\succ \gets$ arbitrary order of $P$ w.r.t. increasing values of~$t(p)$
    
    \COMMENT{$\pos_X(i)$ refers to $p \in X \subseteq P$ on pos.\,$i$ of the
   	suborder of~$\succ$ limited to~$X$} 
  	
    $P_\textrm{p} \gets \{p \in P: c(p) \leq \budget^*\}$
    \COMMENT{prospective projects}
  
    \WHILE[are there projects to consider?]{$P_\textrm{p} \neq \emptyset$}
      \STATE $P' \gets P_\textrm{p}$\COMMENT{$P'$ is a helper variable}
  	\STATE $\baseline(\pos_{P'}(0)) \gets -\inf$\COMMENT{guardian ``fake'' value}
  	\STATE $\baseline(\pos_{P'}(|P'|+1)) \gets +\inf$\COMMENT{guardian ``fake'' value}
  	\STATE $k \gets$ maximum integer $x \leq |P'|$ such that\\
  	\ \ \ $t(\pos_{P'}(x))\sum_{i \leq x-1} A(\pos_{P'}(i)) \leq \budget^*
  	< t(\pos_{P'}(x+1))\sum_{i \leq x+1} A(\pos_{P'}(i))$
      \STATE $T \gets$ maximum $x \leq
  		\nicefrac{\baseline(\pos_{P'}(k+1))}{|A(\pos_{P'}(k+1))|}$\\
      \ \ \ such that
  		$x\sum_{i \in [k]} A(\pos_{P'}(i)) \leq \budget^*$
      \FOR[update~$\stratprof{}$ and $\budget^*$]{$i=1$ \TO $k$}
        \STATE $c(\pos_{P'}(i)) \gets T \cdot A(\pos_{P'}(i))$
    	 \STATE  $\budget^* \gets \budget^* - c(\pos_{P'}(i))$
  	   \STATE $P_\textrm{p} \gets P_\textrm{p} \setminus
  			\{\pos_{P'}(i)\}$
      \ENDFOR
      \IF{$k < |P'|$}
        \STATE $P_\textrm{p} \gets P_\textrm{p} \setminus \{\pos_{P'}(k+1)\}$
      \ENDIF
  	\STATE $P_\textrm{p} \gets \{p \in P_\textrm{p}: c(p) \leq
  		\budget^*\}$\label{alg:too-exp}
    \ENDWHILE
    \RETURN $\stratprof$
  \end{algorithmic}
  }%
  \end{algorithm}%
\begin{proof}
	Our \Cref{alg:NE-existence} computes the claimed~\textAVoverNE{}.

	Let us fix a PB game $G = (P, V, B, d)$ and some corresponding A/D
	tie-breaking $\succ$, as specified in the theorem statement. We first 
	introduce helpful notation and discuss~$\succ$ in more detail. Then, we proceed
	with presenting a high-level description of~\Cref{alg:NE-existence} that finds
	the claimed~\textAVoverNE{}, which we refer to as~$\stratprof$. Eventually, we
	prove the correctness of the algorithm using~\Cref{claim:central-alg} and
	conclude with proving~\Cref{claim:central-alg} itself.

  Recall that by definition $\succ$ orders the projects~$p_i$
	nonincreasingly according to approval-to-delivery-cost ratios
	$\nicefrac{|A(p_i)|}{\baseline(p_i)}$.
	Crucially, $\succ$ is always compatible with the order in which~$\textAVover$
	considers the projects assuming they report their delivery costs, that is,
	where for all~$p_i \in P$, $c(p_i) = \baseline(p_i)$. In what follows, for
	each set of projects~$X \subseteq P$, each~$p_i \in X$, and the corresponding
	suborder~$\succ_X$ of~$\succ$, we write~$\pos_X(p_i)$ to denote the position
	of~$p_i$ in~$\succ_X$. Analogously, for some natural number~$x \leq |X|$, we
	denote by~$\topp_X(x)$ the set of top~$x$ projects according to~$\succ_X$.

	\Cref{alg:NE-existence} constructs the $\textAVover$~$\stratprof$ iteratively,
	starting from strategy~$\stratprof_0$ in which each project reports its
	delivery cost (\Cref{alg:init-strategy}). In each iteration of the while
	loop, the algorithm selects a group of projects that ``underreport'' their
	costs compared to the support that they get. As the next step, the algorithm
	increases the reported costs of these projects in strategy~$\textAVover$ as
	much as possible ensuring that the projects remain funded. If the funded
	projects after this update do not exhaust the budget, the procedure is
	repeated for the remaining projects, whose prices have not been increased so
	far. Such projects, albeit worse regarding the delivery-cost-to-support ratio,
	can still benefit from the increase. The algorithm outputs~$\stratprof$ if one
	of the updates step lead to the situation in which the whole budget is used or
	when there is no more projects underreporting their costs.

	In the following more detailed description of~\Cref{alg:NE-existence}, we use
	the notation as specified in the pseudocode and whenever the value~$k$ in an
	iteration of the loop is smaller than~$|P'|$, we define $p^* = \pos_{P'}(k+1)$
	(assuming the value of~$P'$ from the corresponding iteration).  Central
	to~\Cref{alg:NE-existence} is the while loop. Importantly, due to our basic
	assumption on the delivery costs of the projects (i.e.,\,that for each~$p \in
	P$, $\baseline(p) \leq \budget$), initially set~$P_\textrm{p}$ is not empty
	(\Cref{alg:init-prospective}), so the while loop runs at least once.
	\Cref{alg:computing-t} of the algorithm guarantees that in each iteration the
	loop sets the values of $T$ and~$k$ to:
	\begin{align*}
		\tag{S1}\label{eq:stop-one}
		&k < |P'| \;\;\textrm{and}\;\; T = \nicefrac{\baseline(p^*)}{A(p^*)} \;\;\textrm{or}\\
		\tag{S2}\label{eq:stop-two}
		&k < |P'| \;\;\textrm{and}\;\; T \cdot \sum_{p_j \in \topp_{P'}(k)} |A(p_j)| =
		\budget^*,
		\;\; \textrm {or}\\
		\tag{S3}\label{eq:stop-three}
		&k = |P'| \;\;\textrm{and}\;\; T \cdot \sum_{p_j \in \topp_{P'}(k)} |A(p_j)| =
		\budget^*.\\
	\end{align*}
	The aforementioned case distinction is crucial for the following claim.
	\begin{claim}\label{claim:central-alg}
		Right after executing~\Cref{alg:too-exp} in each iteration of the while loop
		of~\Cref{alg:NE-existence}, it jointly holds that:
	  \begin{enumerate}
	  	\item all projects from~$\topp_{P'}(k)$ are funded under strategy
	  		profile~$\stratprof$ and none of them has a profitable deviation
	  		for~$\stratprof$;\label{enum:first}
	  	\item either $k = |P'|$ or there exists a project~$p^* = \pos_{P'}(k+1)$ that
	  		is not funded and has no profitable deviation for
	  		profile~$\stratprof$;\label{enum:mid}
	  	\item no project from~$\topp_{P'}(k)$ has a profitable deviation for each
	  		strategy profile~$\stratprof'$ that differs from~$\stratprof$ only by
	  		strategies of projects in~$P_\textrm{p}$ in a way that, for each~$p\in
	  		P_\textrm{P}$, $\stratprof'(p) \geq \stratprof(p)$;\label{enum:second}
	  	\item either $k = |P'|$ or the project~$p^* = \pos_{P'}(k+1)$ has no profitable
	  		deviation for each strategy profile~$\stratprof'$ that differs
	  		from~$\stratprof$ only on strategies of projects in~$P_\textrm{p}$ such
	  		that, for each~$p\in P_\textrm{P}$, $\stratprof'(p) \geq
	  		\stratprof(p)$;\label{enum:last}
	  	\item no project in~$P \setminus P_\textrm{p}$ has a profitable deviation
	  		in~$\stratprof$.\label{enum:ultim}
	  \end{enumerate}
  \end{claim}

	Note that~\Cref{alg:NE-existence} ends when executing~\Cref{alg:too-exp}
	makes $P_\textrm{p}$ empty. So, by~\Cref{enum:ultim}
	of~\Cref{claim:central-alg} the algorithm returns profile~$\stratprof$ which
	is a Nash equilibirum. Clearly, in each iteration~$k>0$, so at least one
	element is removed from~$P_\textrm{p}$ in every iteration of the loop.
	Consequently, the algorithm always ends and thus it is correct. To conclude
	the whole proof it remains to show that~\Cref{claim:central-alg} indeed holds.

	\begin{proof}[Proof of~\Cref{claim:central-alg}]
	We provide an inductive argument over the iterations of the while
	loop~of~\Cref{alg:NE-existence}. For the sake of the argument's simplicity,
	instead of thinking of~$\textAVover$ considering projects in nonincreasing
	order of approval-to-cost ratio, we say that it considers projects in
	nondecreasing order of cost-to-approval ratio. We note that these two
	interpretations are equivalent.
	We take the first iteration of the loop as the base case and subsequently show
	that~\Cref{enum:first,enum:mid,enum:second,enum:last,enum:ultim} hold.

	\emph{\Cref{enum:first}.}
	We first show that each $p \in \topp_{P'}(k)$ is funded in
	election~$E(\stratprof{})$. Let us denote as $t(p)$ the ratio
	$\nicefrac{\stratprof(p)}{|A(p)|}$. Due to~\Cref{alg:computing-t} and the
	following
	foreach loop, we know that every project in~$\topp_{P'}(k)$ is tied to be
	considered by~$\textAVover$ due to the same cost-to-approval ratio~$T$. So,
	if~$k = |P'|$, then all projects in $P_\textrm{p} = P$ are
	tied for consideration. Since, by~\Cref{alg:computing-t} all project costs sum
	up to~$\budget^* = \budget$, the claim holds.
	We follow assuming that~$k<|P'|$. In this
	case, \Cref{alg:computing-t} shows that $p^*$ has cost-to-approval ratio~$t(p^*)
	\geq T$. Recall that~$\succ$ orders projects nondecreasingly with respect to
	their cost-to-approval ratio assuming~$\stratprof_0$ and that $\stratprof$
	differs from~$\stratprof_0$ only in strategies of projects in~$\topp_{P'}(k)$.
	Hence, $p^*$~precedes in~$\stratprof$ each project except of those
	in~$\topp_{P'}(k)$ and all projects in~$\topp_{P'}(k)$ are considered
	before~$p^*$ (the latter holds even in the case of~$t(p^*) = T$ due
	to~$\succ$). Eventually, \Cref{alg:computing-t}~guarantees that the cost of
	all projects in~$\topp_{P'}(k)$ does not exceed the budget.
	Knowing that, we show that no project in~$\topp_{P'}(k)$ has a profitable
	deviation in~$\stratprof$. To prove it by contradiction, let us
	assume that some player~$p_i \in \topp_{P'}(k)$ has a profitable deviation
	by reporting cost~$c_i' > c_i$; the opposite case of~$c_i' < c_i$
	trivially does not yield a payoff improvement for~$p_i$. As a result, $p_i$
	has to be funded in the corresponding election~$E((\stratprof{}_{-i}, c_i'))$.
	In this election, $p_i$ has cost-to-approval ratio~$t' =
	\nicefrac{c_i'}{|A(p_i)|} > T$. We further split our analysis into the three
	cases from~\Cref{eq:stop-one,eq:stop-two,eq:stop-three}.

	\noindent
	Assuming~\Cref{eq:stop-one}, project~$p^*$ is also funded in
	election~$E'$, as it is considered before~$p_i$ due to $t' > T$.  Hence, the
	funded projects cost at least 
	\begin{align*}
	T \cdot \sum_{\substack{p_j \in	\topp_{P'}(k) \\ p_j \neq p_i}} |A(p_j)| +
	t'|A(p_i)| + T|A(p^*)| > \\
 T \cdot \sum_{p_j \in	\topp_{P'}(k+1)} |A(p_j)|.
 \end{align*}
	However, due to~\Cref{alg:selection-k}, we know that $\budget^* = \budget < T
	\sum_{p_j \in \topp_{P'}(k+1)} |A(p_j)|$, which gives the contradiction with the
	fact that~$p_i$ is funded.

  \noindent
	Suppose~\Cref{eq:stop-two} or~\Cref{eq:stop-three} hold. Then, given our
	assumption that~$p_i$ is funded, the total cost of funded projects is 
 \begin{align*}
	T \cdot \sum_{\substack{p_j \in	\topp_{P'}(k)\\ p_j \neq p_i}} |A(p_j)| +
	t'|A(p_i)| > \\ T \sum_{p_j \in \topp_{P'}(k)} |A(p_j)| = B^* =  B,
\end{align*}
	Here, the equality is due to the assumption of~\Cref{eq:stop-two}
	or~\Cref{eq:stop-three}; which yields the sought contradiction.

	\emph{\Cref{enum:mid}.} If~$k = |P'|$, then the statement
	trivially holds. Otherwise, let us consider $p^*$ in the light
	of~\Cref{eq:stop-one}. Due to~\Cref{alg:selection-k}, we have that 

 \begin{align*}
	\budget < T \sum_{p_j
	\in \topp_{P'}(k+1)} |A(p_j)| = \\
 T \sum_{p_j \in \topp_{P'}(k)} |(A(p_j))| + \baseline(p^*).
\end{align*}
	Hence, since all projects in~$\topp_{P'}(k)$ are funded, $p^*$ is not funded.
	Naturally, if $p^*$ reports a cost bigger
	than~$\baseline(p^*)$ it will not be funded even more, whereas reporting a
	cost lower than~$\baseline(p^*)$ results in a worse payoff, which finishes the
	argument for, for~\Cref{eq:stop-one}. By the condition
	of~\Cref{eq:stop-two}, it immediately holds
	that all projects~$\topp_{P'}(k)$ use up the whole budget.
	As a result, again, $p^*$ is not funded and, analogously to the case
	of~\Cref{eq:stop-one}, $p^*$ has no profitable deviation. We already observed
	that~\Cref{enum:mid} trivially holds for~$k = |P'|$, which
	subsumes~\Cref{eq:stop-three}.

	\emph{\Cref{enum:last,enum:second}.}
	Observe that assuming strategy~$\stratprof$, all projects in $P_\textrm{p}$
	are considered after all projects outside of~$P_\textrm{p}$ in
	election~$E(\stratprof)$. This is because order~$\succ$
	is compliant with the order of considering the projects by~\textAVover{} and
	our modifications
	to the initial strategy profile~$\stratprof$ do not change the order of
	considering projects (note
	that we let modified projects be tied for consideration with cost-to-approval
	ratio~$T$).
	Furthermore, note that~\Cref{alg:NE-existence} never decreases the reported
	costs of the projects and never considers the same project twice, so the
	projects in~$P_\textrm{p}$ will
	never be considered before the other ones as a result of further modifications
	of~$\stratprof$.
	So, no modification of the
	profile~$\stratprof{}$ that increases the costs for projects in~$P_\textrm{p}$
	can influence the decision made for projects in~$\topp_{P'}(k)$
	or~$\topp_{P'}(k+1)$ (depending on whether~$k<|P'|$), as the latter are
	considered by~\textAVover{} earlier than projects in~$P_\textrm{p}$.

	\emph{\Cref{enum:ultim}.}
	Note that $P \setminus P_\textrm{p}$ consists only of the following projects:
	(i) those removed in the foreach loop, that is, $\topp_{P'}(k)$; (ii) $p^*$ if
	$k < |P'|$; and (iii) those removed in~\Cref{alg:too-exp}. Regarding the first
	two groups, we have already shown that they do not have a profitable deviation
	for~$\stratprof$. Let $\hat{p}$ be $\topp_{P'}(k)$ if $k < |P'|$ or,
	otherwise, let $\hat{p}$ be $p^*$. The last group~$Y$ consists of these
	projects~$p \in P'$, for which $\hat{p} \succ p$ and whose~$\baseline(p) >
	\budget^*$. Hence, by~$\succ$, these projects are considered after all
	projects in~$\topp_{P'}(k) \cup \{\hat{p}\}$. However, projects in~$\topp_{P'}(k)$
	are selected to be funded, which leaves exactly~$\budget^*$ remaining budget.
	So there is not enough budget left to fund any project in~$Y$, even if it
	reports its delivery cost. Hence, no project in~$Y$ has a profitable
	deviation.

	Thus, we have established the base case and we move on to the induction step.
	Consider $x>1$ and the $x$-th iteration of the while loop, assuming that the
	base case claims are met for the $(x-1)$-th iteration.  Due to the assumptions
	for the $(x-1)$-th iteration together with the fact that our algorithm never
	changes the order~$\succ$ in which~\textAVover{} considers projects and that
	it never considers the same project twice, we can ignore all projects
	processed in previous $(x-1)$ iterations of the while loop.  Clearly, the
	ignored projects will have no impact on the current loop iteration except for decreasing
	the value of the variable~$B^*$ representing the remaining budget.
	Consequently, the arguments
	for~\Cref{enum:first,enum:mid,enum:second,enum:last} carry on without changes
	for the $x$th iteration of the while loop. 
	The claim from~\Cref{enum:ultim}, however, needs more attention.
	Let~$Q^{(x-1)} = P \setminus P_\textrm{p}^{(x-1)}$ be the set of players
	without profitable deviations after the $(x-1)$-th iteration of the loop. Note
	that~$P' = P_\textrm{p}^{(x-1)}$. We denote by~$R$ the set of projects, that
	are removed from the initial state of~$P_\textrm{p}$ by
	\Cref{alg:updating-Pp,alg:too-exp}. That is, $R$ consists of all these
	projecst
	that---due to~\Cref{enum:first,enum:mid,enum:second,enum:last} and the
	argument for~\Cref{enum:ultim} in the base case---have no profitable deviation
	in~$\stratprof$. So, we have that~$P_\textrm{p} = P'
	\setminus R$. We now consider the set~$Q
	= P \setminus P_\textrm{p}$ of the projects for which we need to show that
	they have no profitable deviation. Putting the bits
	together, we observe the following:
	\begin{align*}
	Q = &P \setminus P_\textrm{p} = P \setminus (P' \setminus R) = \\
     &P \setminus
	(P_\textrm{p}^{(x-1)} \setminus R) = (P \cap R) \cup (P \setminus
	P_\textrm{p}^{(x-1)}).
	\end{align*}
	Since~$P \cap R = R$, we obtained that $Q = R \cup Q^{(x-1)}$ thus showing
	that the $x$-th iteration extends the set of project that do not have a
	profitable deviation with a collection of project that have no profitable
	deviation either. As a result, we proved~\Cref{enum:ultim} for the $x$-th
	iteration.%
  \renewcommand{\qed}{\hfill$\blacksquare$}%
  \end{proof}%
  \renewcommand{\qed}{\hfill$\square$}%
	Having proven~\Cref{claim:central-alg}, we completed the argument for the
	correctness of~\Cref{alg:NE-existence} and~\Cref{thm:avcost-ne-exists}.
\end{proof}
}

We conclude our analysis of \textAVover{} by noting that, as
\Cref{ex:noSeqPhrag} shows, even for plurality preferences the costs of projects
selected in an equilibrium can be far from using the entire budget. In fact,
they can be arbitrarily far from it.

\begin{proposition}\label{result:phragsmallcost}
	For each integer $\gamma >1$ there is a PB game %
	with plurality preferences and $\textAVover$-NE profile~$\stratprof$ with
	the total cost of projects funded under $\textAVover$ at most
	$\nicefrac{\budget}{\gamma}$.
\end{proposition}

\appendixproof{Proposition}{result:phragsmallcost}{%
\begin{proof}

  Take a natural $\gamma > 1$ and consider a PB game
  $(P,V,\budget,\baseline)$ where $P= \{p_1, p_2, p_3 \}$,
  $\budget = 10$, $\baseline(p_1) = \baseline(p_2) = 0$, and
  $\baseline(p_3) = 10 - \frac{1}{2 \gamma}$.  The voters have
  plurality ballots such that $|A(p_1)|= |A(p_2)| = 1$ and
  $A(p_3) = 20 \gamma -1$.  We assume tie-breaking order
  $p_1 \succ p_2 \succ p_3$.
  We claim that strategy profile $\stratprof$ such that
  $\stratprof(p_1) = \stratprof(p_2) = \frac{1}{2 \gamma}$ and
  $\stratprof(p_3) = 10 - \frac{1}{2 \gamma}$ is a $\Phragmen$-NE.

  First, we see that if the projects reports costs as in $\stratprof$
  then $\Phragmen$ selects $p_1$ and $p_2$.  Indeed, we see that at
  time moment $\frac{1}{2\gamma}$ the voters supporting each of the
  projects have exactly as much money as is need to purchase them. Due
  to tie-breaking, the singleton voters supporting $p_1$ and $p_2$ buy
  these projects and, then, there is not enough budget left for $p_3$
  and the rule finishes.

  Second, we observe that no project can benefit by changing its
  strategy under $\stratprof$. Indeed, if either $p_1$ or $p_2$
  decreased their cost, they would obtain lower payoff, and if either
  of them increased their cost, $p_3$ would be funded instead and the
  payoff of the project that increased its cost would drop to zero.
  If $p_3$ decreased its cost then it would be selected, but its
  payoff would become negative (due to the delivery cost), and if it
  increased its cost then its payoff would remain zero.
  Consequently, $\stratprof$ is $\Phragmen$-NE.

  Finally, we have
  $\stratprof(p_1) + \stratprof(c_2) = \nicefrac{1}{\gamma}$, and as
  $B = 10$, this is less than~$\nicefrac{\budget}{\gamma}$. This
  concludes the proof.
\end{proof}%
}

\subsection{{\Phragmen}}

It is easy to verify that for plurality ballots $\textAVover$ and
$\Phragmen$ are identical, that is, for every PB instance with plurality
ballots (and the same internal tie-breaking) they output the same
projects.
Thus, the next corollary translates results from the previous section
to the case of $\Phragmen$.

\begin{corollary}\label{cor:exclusiveNo0}
  For each PB game %
  with plurality ballots, (a)~if $\baseline(p) \leq \appstrat(p)$ for
  all~$p \in P$, then $\appstrat$ is the
  only $\Phragmen$-NE for every tie-breaking order;
  (b) otherwise, for %
	an A/D
	tie-breaking
  order there is a~$\Phragmen$-NE (and there are PB games
  and tie-breaking orders with no~$\Phragmen$-NE).
\end{corollary}

Naturally, \Cref{result:phragsmallcost} also holds for~\Phragmen{}.
So, there are games with equilibria for which $\Phragmen$ funds
projects whose total cost is an arbitrarily
small fraction of the budget. 
Next, we focus
on more involved preference profiles.
On the positive side, for party-list ballots (and zero delivery
costs) we always have unique equilibria, for each
tie-breaking order.  Intuitively, %
each project $p$ reports cost proportional to the number of its
approvals, divided by the size of its party.

\begin{proposition}\label{result:phragparty}
  For every PB game %
  with party-list
  ballots and zero delivery costs, there is the unique $\PhragmenNE$
  $\stratprof$, where for every project $p$ we have
  $\stratprof(p) = B \cdot \frac{|A(p)|}{|V| \cdot |\party(p)|}$.
\end{proposition}

\appendixproof{Proposition}{result:phragparty}{%
\begin{proof}
  Consider a PB game $G$ and the strategy profile $\stratprof$ as defined in the
  statement of the proposition.  We claim that $\stratprof$ is a
  $\seqPhrag$-NE. First, we observe that under $\stratprof$
  $\seqPhrag$ selects all the projects, so it is not beneficial for
  any of them to report a lower cost. On the other hand, if some
  project $p$ reports a higher cost, this project is not selected by
  $\seqPhrag$. To see why this is the case, consider some project $p$.
  Under $\stratprof$, the total cost of the projects in $\party(p)$ is
  $\nicefrac{B\cdot|A(p)|}{|V|}$. Since each of the $|A(p)|$ voters
  supporting these projects earns one unit of money per one unit of
  time, altogether they earn this money in time $\nicefrac{B}{|V|}$.
  This value is independent of $p$ so, under $\stratprof$, the last
  project of each party is funded at the same time. Hence, if some
  project increased its price, it would be considered even later and,
  by that time, there would not be enough budget left. Hence, it is
  never beneficial to increase a cost and so, $\stratprof$ is a $\seqPhrag$-NE.

  It remains to show that $\stratprof$ is the unique $\PhragmenNE$ for
  our game. To this end, let $\stratprof'$ be some arbitrary
  equilibrium for~$G$.  By \Cref{pro:NE-uses-full-budget}, we know
  that under $\stratprof'$ the reported costs of all projects sum up
  to $\budget$ and all projects are funded. Next, we observe that for
  each project $p$, all projects from $\party(p)$ report the same
  cost. Indeed, if there were two projects, $p'$ and
  $p'' \in \party(p)$, such that
  $\stratprof'(p') < \stratprof'(p'')$, then it would be beneficial
  for $p'$ to report a higher cost (but below $\stratprof'(p'')$), so
  that $\seqPhrag$ would still consider and fund it prior to $p''$
  (for which, then, there would not be enough budget left).

  Finally, if there are two projects, $p$ and $q$, such that
  $\party(p) \neq \party(q)$, then, under $\stratprof'$, the last
  project from $\party(p)$ and the last project from $\party(q)$ are
  selected by $\seqPhrag$ at the same time. Indeed, if this were not
  the case, then it would be beneficial for the one selected earlier
  to report higher cost (but so that it still is selected at an
  earlier time than the other project).

  Altogether, the only strategy profile that satisfies the properties
  described above is $\stratprof$ as defined in the statement of the~proposition.
\end{proof}%
}

We conjecture that one can adapt the algorithm
behind~\Cref{thm:avcost-ne-exists} to compute in polynomial-time a tie-breaking
order and a corresponding~\PhragmenNE{} for party-list profiles and arbitrary
delivery costs. The reported cost for each party should be spread equally to all
members (as in~\Cref{result:phragparty}). Yet, so far, the respective
correctness proof remains elusive.

Unfortunately, party-list preferences with zero delivery costs are
where the good news end.  Indeed, there are instances with no
equilibria, even for very few candidates and
voters. %

\begin{proposition}\label{pro:phrag-laminar-no-ne}
  There is a PB game %
  with six projects and six voters
  and zero delivery costs, for which there is no \PhragmenNE{}
  irrespective of the %
  tie-breaking order.
\end{proposition}

\begin{proof}[Proof Sketch]
  Let our project set be $P = P' \cup P''$, where
  $P' = \{p'_1,p'_2,p'_3\}$ and $P'' =
  \{p''_1,p''_2,p''_3\}$. Similarly, the set of voters is
  $V = V' \cup V''$, where $V' = \{v_{1}',v_{2}',v_{3}'\}$ and
  $V'' = \{v''_1, v''_2,v''_3\}$.  Project $p'_1$ is approved by
  $v'_1$, $p'_2$ by $v'_2$, and $p'_3$ by all
  voters in $V'$. The approvals for projects in $P''$ are analogous,
  but come from voters in $V''$ (hence our instance can be illustrated
  by two copies on the right-hand side picture from
  \Cref{fig:typesex} with $v_2$ and $v_3$ swapped).

  We first
  show that in all equilibria all projects report such costs that at
  the first time some voters can make a purchase, there is a %
  tie between all of them. If $p'_1 \succ p'_3$, then $p'_2$ benefits by reporting a  higher cost
  (as $\Phragmen$ funds $p'_1$ and the supporters of $p'_3$
  need time to collect money for it; meanwhile, $v'_2$
  also gets more money  %
  to spend on $p'_2$). The
  cases where $p'_2 \succ p'_3$, or either $p''_1$ or
  $p''_2$ is preferred to $p''_3$ are symmetric.
  If $p'_3$ is preferred to  other members of $P'$ and $p''_3$ %
  to other members of $P''$, then $p'_3$
  (symmetrically, $p''_3$) prefers to increase the price, as after $p'_1$, $p'_2$, and $p''_3$ are simultaneously funded, the time needed by $v''_1$ and $v''_2$ to collect money
  for $p''_1$ and $p''_2$ suffices for voters in $V'$ to collect more
  than enough money for $p'_3$. The full proof is in the appendix.
\end{proof}

\appendixproof{Proposition}{pro:phrag-laminar-no-ne}{%
\begin{proof}
  Let our project set be $P = P' \cup P''$, where
  $P' = \{p'_1,p'_2,p'_3\}$ and $P'' =
  \{p''_1,p''_2,p''_3\}$. Similarly, the set of voters is
  $V = V' \cup V''$, where $V' = \{v_{1'},v_{2'},v_{3'}\}$ and
  $V'' = \{v_{1''},v_{2''},v_{3''}\}$.  The approvals are as follows (see also
  \Cref{fig:phrag-no-ne} for illustration):
  \begin{enumerate}
  \item $p'_1$ is approved by $v_{1'}$, $p'_2$ is approved by $v_{2'}$, and
    $p'_3$ is approved by all the voters in $V'$.
  \item The approvals for the projects in $P''$ are analogous, except
    that they are approved by the voters from $V''$.
  \end{enumerate} 
  We set the delivery costs $\baseline$ to be zero for every project, and
  we set the budget $\budget$ to be $1$ (the exact value will be
  irrelevant as we will operate on times when $\Phragmen$ reaches the
  costs of particular projects rather than on directly on these
  costs). We claim that under $\Phragmen$ there are no Nash equilibria
  for the thus defined PB game $G = (P,V,\budget,d)$.

  \begin{figure}[t]\centering
  \newcommand{\robotface}[2]{
    \draw (#1+0,0) rectangle (#1+3,1);
    \draw[dashed] (#1+1,0) rectangle (#1+2,1);
    \draw (#1+0,1.5) rectangle (#1+1,2.5);
    \draw (#1+2,1.5) rectangle (#1+3,2.5);
    \draw[dotted] (#1+0.5,-0.2) -- (#1+0.5,2.7);
    \draw[dotted] (#1+1.5,-0.2) -- (#1+1.5,1.2);
    \draw[dotted] (#1+2.5,-0.2) -- (#1+2.5,2.7);
    \node[anchor=north] at (#1+0.5,-0.2) {$v_1#2$};
    \node[anchor=north] at (#1+1.5,-0.2) {$v_3#2$};
    \node[anchor=north] at (#1+2.5,-0.2) {$v_2#2$};
    \node[anchor=east] at (#1+0.2,0.5) {$p#2_3$};
    \node[anchor=east] at (#1+0.2,2) {$p#2_1$};
    \node[anchor=east] at (#1+2.2,2) {$p#2_2$};
  }
  
  \begin{tikzpicture}[scale=0.55]
    \robotface{0}{'}
    \robotface{5}{''}
  \end{tikzpicture}
  \caption{Illustration of the PB game from the proof of
    \Cref{pro:phrag-laminar-no-ne}. The projects are depicted as
    boxes. Each voter approves those projects that are drawn directly
    above him or her (and are crossed by the dotted line).}\label{fig:phrag-no-ne}
  
\end{figure}

  For the sake of contradiction, let us assume that there is
  $\PhragmenNE$ for $G$ and let it be $\stratprof^*$. For each project
  $p \in P$, let $\phtime(p)$ be the time moment (in the sense of the
  $\Phragmen$ rule) when the voters approving $p$ would collect
  exactly $\stratprof^*(p)$ amount of money (assuming that neither of
  these voters spends it on any other projects in between; hence
  $\phtime(p) = \nicefrac{\stratprof^*(p)}{|A(p)|}$). We make the
  following observations:
  \begin{enumerate}
  \item All projects in $P$ are funded for the reported costs
    $\stratprof^*$ (if some project were not funded, it would prefer
    to report a small nonzero cost that would be at most equal to
    $\budget$ and that would ensure that $\Phragmen$ considers it
    first).
  \item It must be the case that $\phtime(p'_1) = \phtime(p'_2)$.
    Indeed, if we had $\phtime(p'_1) < \phtime(p'_2)$ then it would be
    beneficial for $p'_1$ to slightly increase its cost, but so that
    $\phtime(p'_1)$ would still be below $p'_2$. Then, irrespective in
    what order are the other projects funded, $p'_1$'s voter would
    collect enough money to purchase $p'_1$ before the $p'_2$'s voter
    would, and $p'_1$ would be bought (since $p'_2$ would not be
    selected at this time yet, there would be enough budget left for
    this). This would contradict that $\stratprof^*$ is an
    equilibrium. The case $\phtime(p'_2) < \phtime(p'_1)$ is
    symmetric.

  \item It must be the case that $\phtime(p'_3) = \phtime(p'_1)$ and,
    hence, also equal to $\phtime(p'_2)$. Indeed, if we had
    $\phtime(p'_3) < \phtime(p'_1) = \phtime(p'_2)$ then it would be
    beneficial for $p'_3$ to report slightly higher cost, so that
    $\phtime(p'_3)$ would still be smaller than
    $\phtime(p'_1) = \phtime(p'_2)$, yet $p'_3$'s cost would not have
    increased by more than the total costs of $p'_1$ and $p'_2$. Then
    $p'_3$ would still be selected before $p'_1$ and $p'_2$ and there
    would be sufficient amound of budget left for it. This would
    contradict that $\stratprof^*$ is an equilibrium.  On the other
    hand, if $\phtime(p'_3) > \phtime(p'_1) = \phtime(p'_2)$ then it
    would be beneficial, e.g., for $p'_1$ to report slightly higher
    cost, but ensuring that $\phtime(p'_1) < \phtime(p'_3)$. Indeed,
    if originally we have $\phtime(p'_1) < \phtime(p'_3)$, then $p'_1$
    is funded before $p'_3$ by $\Phragmen$. After the increase, $p'_1$
    would still be purchased before $p'_3$ and, thus, there would
    still be sufficient budget for it. This would contradict that
    $\stratprof^*$ is an equilibrium.
  \item By a reasoning analogous to the one given above, it must be the
    case that $\phtime(p''_1) = \phtime(p''_2) = \phtime(p''_3)$.
  \end{enumerate}
  Given the above observations, we set:
  \begin{align*}    
   \phtime'  &= \phtime(p'_1) = \phtime(p'_2) = \phtime(p'_3), \text{ and} \\
   \phtime'' &= \phtime(p''_1) = \phtime(p''_2) = \phtime(p''_3).
  \end{align*}
  Note that this implies that
  $\stratprof^*(p'_1) = \stratprof^*(p'_2) = \phtime'$,
  $\stratprof^*(p''_1) = \stratprof^*(p''_2) = \phtime''$,
  $\stratprof^*(p'_3) = 3\cdot \phtime'$, and
  $\stratprof^*(p''_3) = 3\cdot \phtime''$.
  We claim that $\phtime' = \phtime''$. Indeed, let us consider what
  happens if $\phtime' < \phtime''$ (the case where
  $\phtime'' < \phtime'$ is symmetric). There are two cases two
  consider (the second case also splits into two).

  \paragraph{Case~A}
  First, let us assume that at least one of $p'_1$ and $p'_2$ is
  preferred to $p'_3$ by the tie-breaking order (for the sake of
  specificity, let this be $p'_1$). Then it is beneficial for $p'_2$
  to slightly increase its cost. To see why this is the case, let us
  consider how $\Phragmen$ operates after this increase. At time
  $\phtime(p'_1) = \phtime(p'_3)$, the voters have enough funds to
  purchase either $p'_1$ or $p'_3$ (and not enough to purchase $p'_2$,
  who increased its cost). The rule selects $p'_1$ due to
  tie-breaking. Consequently, voter $1'$ pays for $p'_1$ and his or
  her virtual bank account is reset to zero. Voters $2'$ and $3'$
  still have $\phtime'$ amount of money. The cost of $p'_3$ is
  $3 \cdot \phtime'$, so voters in $N'$ will have collected enough
  money for it after further $\nicefrac{1}{3} \phtime'$ amount of
  time. However, if $p'_2$ increased its cost from $\phtime'$ to an
  amount a bit below $\nicefrac{4}{3} \cdot \phtime'$ then its voter
  will have collected this amount earlier, and $p'_2$ will be
  purchased before $p'_3$. This contradicts the fact that
  $\stratprof^*$ is an equilibrium.

  \paragraph{Case~B$'$}
  The second case is that $p'_3$ is preferred by tie-breaking to both
  $p'_1$ and $p'_2$. Then it is beneficial for $p'_3$ to slightly
  increase its cost. Again, let us consider how $\Phragmen$ operates
  after such a change. If $p''_3$ is selected prior to both $p''_1$
  and $p''_2$, then at time $\phtime''$ (after the purchase of
  $p''_3$) the voters have the following amounts of money:
  \begin{enumerate}
  \item Voter $3'$ has $\phtime''$ amount of money and the remaining
    voters in $N'$ have nonzero amounts of money (specifically, each
    of them has $\phtime''-\phtime'$, because they paid for $p'_1$ and
    $p'_2$ at time $\phtime'$).
  \item All the voters in $V''$ have empty bank accounts. 
  \end{enumerate}
  Hence, only after another $\phtime''$ amount of time will the voters
  in $N''$ have enough money to purchase $p''_1$ and $p''_2$. Yet, at
  this time the voters in $N'$ would, altogether, have more than
  $4 \phtime''$. So, if $p'_3$ increased its cost to be between
  $3\cdot \phtime'$ and $4\cdot\phtime'$ (which is smaller than
  $4\cdot \phtime''$), then $p'_3$ would be funded before $p''_1$ and
  $p''_2$, and there would be sufficient amount of budget left for
  this.

  \paragraph{Case~B$''$} On the other hand, if at least one of $p''_1$
  and $p''_2$ is preferred to $p''_3$ by tie-breaking at time
  $\phtime''$, then both $p''_1$ and $p''_2$ are funded at that time
  (because after one of them is selected due to tie-breaking, the
  voters in $V''$ no longer have enough money to purchase $p''_3$, but
  they do have enough for the other one among $p''_1$ and
  $p''_2$). Consequently, after $p''_1$ and $p''_2$ are purchased at
  time $\phtime''$, only projects $p'_3$ and $p''_3$ are not selected
  yet and the voters have the following amounts of money:
  \begin{enumerate}
  \item Voter $v_{3'}$ has $\phtime''$ amount of money and the remaining
    voters in $V'$ have nonzero amounts of money (specifically, each
    of them has $\phtime''-\phtime'$, because they paid for $p'_1$ and
    $p'_2$ at time $\phtime'$).
  \item Voter $v_{3''}$ has $\phtime''$ amount of money and the remaining
    voters in $V''$ have no money.
  \end{enumerate}
  Since voters in $V'$ have, together, more money than those in $V''$,
  but voters from both groups (jointly) earn money at the same rate
  (as $|V'| = |V''|$), if $p'_3$ increased its cost to be slightly
  below that of $p''_3$, it will be funded before $p'_3$.
  Consequently, if $\stratprof^*$ is an equilibrium then we must have
  $\phtime' = \phtime''$.\bigskip

  Finally, it suffices to show that the assumption
  $\phtime' = \phtime''$ also leads to no equilibrium.
  W.l.o.g., we
  assume that $p'_3$ precedes both $p'_1$ and $p'_2$ in the
  tie-breaking order, and that $p''_3$ precedes both $p''_1$ and
  $p''_2$ (if this were not the case, then the arguments given in
  Case~A above would still show that $\stratprof^*$ is not an
  equilibrium). However, now we can see that, e.g., it is beneficial
  for $p'_3$ to slightly increase its cost. This follows by the same
  reasoning as given in Case~B$'$. Hence, we have reached a
  contradiction. Consequently, under $\Phragmen$ there is no Nash equilibrium
  in our game.
\end{proof}%
}

\subsection{Method of Equal Shares}
The mechanics of~\Mes{} might appear similar to those of \Phragmen{}, as all these rules implement the %
approach in which voters
buy the projects they approve. However, contrary to~\Phragmen{}, in both types of \Mes{} voters have limited amounts of budget to spend. Hence, the
approval scores of projects give %
upper-bounds on the maximum costs that the projects might
report. That is, since each voter gets an equal share of the budget,
projects that cost more than the total budget of their supporters
would never be funded. %

\begin{observation}\label{obs:mesBound} For every PB game %
  and a strategy profile~\stratprof{}, if there is a
  project~$p_i \in P$ such that
  $\stratprof(p_i) > \budget \cdot \nicefrac{|A(i)|}{|V|}$, then
  $p_i$ is not funded under $\Mes$. %
\end{observation}

For %
\MesCost{}, we provide positive results. Specifically, we show that a
\MesCostNE{} always exists regardless of the tie-breaking
order, even for arbitrary delivery costs, it is polynomial-time
computable, and uses the entire budget (at least if all the projects
have sufficiently low delivery costs).

\begin{theorem}\label{thm:mesCostNEGeneral}
For each PB game, %
there is a
polynomial-time computable \MesCostNE{} profile $\stratprof$, 
unique up to the actions of the projects not selected by~\MesCost{}
under~\stratprof{}.
\end{theorem}

\begin{proof}[Proof sketch] Let~$p$ be the project with the
  largest number of approvals (and preferred in the tie-breaking
  order, should there be other projects with the same number of
  approvals). We observe that~$p$ can always request the full amount of money that its
  voters have at the beginning of the execution of $\MesCost$; the
  rule would fund~$p$ irrespective of the costs of the other projects. Thus,
  $p$~reports this amount and is funded (if the amount is below
  its delivery cost, $p$~reports the delivery cost and is
  not funded). The project has no reason to report less, and reporting
  more would cause it to not be funded. Then, we disregard all the
  voters that approve it and repeat the reasoning for the remaining
  projects and voters. We put the full proof in the appendix.
  \end{proof}

\appendixproof{Theorem}{thm:mesCostNEGeneral}{%
\begin{proof}
	Let $(P, V, \budget, \baseline)$ be a \MesCost{} PB game.

	We provide an iterative algorithm that computes~\stratprof{}. In each
	iteration, the algorithm fixes selected values of~\stratprof{}, drops the
	corresponding projects from further consideration and deletes the voters that
	have no budget left. The algorithm finishes when there is no more projects to
	consider. Before laying out the details, let us recall that in~\MesCost{},
	each voter gets the equal share of the budget and cannot spend on the projects
	more than their entitlement. 

	Our algorithm constructs the strategy~\stratprof{} step by step maintaining
	the collection~$V'$ of voters to consider and the collection $P'$ of projects to
	consider.
	The algorithm starts with setting~$V' = V$, $P' = P$
	and proceeds as follows:
  \begin{enumerate} \label{alg:mesCostNEAlgorithm}
		\item \label{stp:thm:mesCostNEGeneral:get-rid-of-impossible-to-select}
			Remove from~$P'$ all projects~$p_i \in P'$ for which~$|A(p_i) \cap V'| \cdot \nicefrac{\budget}{|V|}$ is lower than $\baseline(p_i)$ or equal to zero and let their strategy
			be~$\stratprof(p_i) = \baseline(p_i)$.
		\item \label{stp:thm:mesCostNEGeneral:selecting-project}
			For each project $p_j \in P'$, let $\alpha_j = \nicefrac{1}{|A(p_j)
			\cap V'|}$ (due to
			Step~\ref{stp:thm:mesCostNEGeneral:get-rid-of-impossible-to-select} we
			avoid dividing by zero) and let $p^*$ be the project~$p_j \in P'$ with the
			minimum $\alpha_j$-value (in case of a tie, select the first one in the
			tie-breaking order).
		\item \label{stp:thm:mesCostNEGeneral:strategy-fixing}
			Set the strategy~$\stratprof(p^*)$ to $|A(p^*) \cap V'| \cdot
			\nicefrac{\budget}{|V|}$, that is, such that~$p^*$ reports the cost equal
			to the total budget of its supporters in~$V'$.
  	\item \label{stp:thm:mesCostNEGeneral:spend-budget}
			Remove all supporters of~$p^*$ from $V'$ and remove~$p^*$ from~$P'$.
		\item Repeat Steps
			\ref{stp:thm:mesCostNEGeneral:get-rid-of-impossible-to-select} to
			\ref{stp:thm:mesCostNEGeneral:spend-budget} until $P'$ is empty. 
	\end{enumerate}
	
	Our algorithm is based on the following useful claim about the decisions made
	by~\MesCost{}.
	\begin{claim}\label{claim:approval-proportional-alpha}
		Let us fix some stage of \MesCost{} in which every voter with a nonzero
		budget has the same value of budget and let~$V'$ be these voters.
		Then, for each $\alpha'$-affordable project~$p'$ it holds that
		that $\alpha' = \nicefrac{1}{|A(p') \cap V'|}$.
	\end{claim}
	The claim is implied by the definition of~\MesCost. Observe that if a
	project~$p$ is $\alpha$-affordable, then the voters approving it have enough
	budget to buy it and $\alpha$ is the minimum such~$x$ that $|A(p') \cap V'|
	\cdot x \cdot \cost(p) \leq \cost(p)$. The sought $\alpha$ is then clearly
	$\nicefrac{1}{|A(p') \cap V'|}$.

	Let us denote by~$(p^*_1, p^*_2, \ldots, p^*_k)$ the projects considered in the
	respective iterations $1$ to $k$ of
	Steps~\ref{stp:thm:mesCostNEGeneral:get-rid-of-impossible-to-select}-\ref{stp:thm:mesCostNEGeneral:spend-budget}.
	In what follows we show that in each Nash equilibrium \MesCost{} selects
	exactly projects~$(p^*_1, p^*_2, \ldots, p^*_k)$ in the given order and
	that~$\stratprof$ is a Nash equilibrium. We apply induction over the stages
	of~\MesCost{}.

	For the base case, we argue that
	$p^*_1$ has to be selected first by~\MesCost{} and has to
	play~$\stratprof(c^*_1)$ in every Nash equilibirum. Assume for contradiction
	that some other project~$p'$ is selected first instead of~$p^*_1$ in a Nash
	equilibrium. Due to~\Cref{claim:approval-proportional-alpha} and
	since~\MesCost{} selects $\alpha$-affordable projects starting from those with
	the smallest~$\alpha$, it follows that one of the three holds: (i)
	$\nicefrac{1}{|A(p')|} < \nicefrac{1}{|A(p^*_1)|}$, (ii)
	$\nicefrac{1}{|A(p')|} = \nicefrac{1}{|A(p^*_1)|}$ and $p'$ is preferred
	to~$p^*_1$ by the tie-breaking, or (iii) $p^*_1$ reports a cost greater than
	the budget of its supporters. Cases (i) and (ii) yield a clear contradiction
	to Step~\ref{stp:thm:mesCostNEGeneral:selecting-project}, which
	defines~$p^*_1$. In Case (iii), reporting cost~$\stratprof'(p^*_1) =
	\baseline(p^*_1) + \epsilon \leq \budget \cdot \nicefrac{|A(p^*_1)|}{|V|}$
	(such an~$\epsilon$ always exists due
	to~Step~\ref{stp:thm:mesCostNEGeneral:get-rid-of-impossible-to-select}) is a
	profitable deviation for~$p^*_1$---a contradiction. Knowing that~$p^*_1$ is
	selected first in every Nash equilibrium, we observe that if~$p^*_1$ reports a
	cost~$\stratprof'(p^*_1) < \budget \cdot \nicefrac{|A(p^*_1)|}{|V|}$, then
	reporting exactly~$\stratprof(p^*_1) = \budget \cdot
	\nicefrac{|A(p^*_1)|}{|V|}$ is always a profitable deviation, which confirms
	that in every Nash equilibrium $p^*_1$'s strategy is $\stratprof(p^*_1)$ and
	it is selected first by~\MesCost{}.
	
	We move on to the inductive step and thus consider stage~$i$ of~\MesCost{} in
	which project~$p^*_i$ is selected. Imporantly, as we have shown that
	player~$p^*_1$ has to report cost~$\stratprof(p^*_1)$ equal to the total
	budget of~$p^*_1$ supporters, then we can assume that at the~$i$-th stage we
	have only voters who either spend all their budget or who did not spend their
	budget at all; we denote the latter group by~$V'$. As a result,
	\Cref{claim:approval-proportional-alpha} holds at the $i$-th stage, which we
	consider. Thanks to this, we apply an analogous exchange argument
	(pretending that voters outside of~$V'$ do not exist) to that for the
	base case to show that indeed~$p^*_i$ has to be selected at the $i$-th stage.
	Then, we directly repeat the argument regarding the optimal strategy
	for~$p^*_1$ applying it to~$p^*_i$. Eventually, we have shown that in each
	Nash equilibrium for~\MesCost{} projects~$(p^*_1, p^*_2, \ldots, p^*_k)$ are
	selected and report, respectively, costs~$(\stratprof(p^*_1),
	\stratprof(p^*_2), \ldots, \stratprof(p^*_k))$.
\end{proof}%
}

For %
\MesApr{}, the situation is more complicated, as an equilibrium may
not exist even for a party-list profile with only three projects and
one voter. Nevertheless, it always exists for plurality ballots and
for party-list ballots with zero delivery costs (the proofs are in the appendix).

\begin{theorem}\label{thm:mesAprNEPlurality}
For every PB game with plurality ballots, there is a polynomial-time computable
\MesAprNE{} regardless of the internal tie-breaking order and delivery costs. 
\end{theorem}

\appendixproof{Theorem}{thm:mesAprNEPlurality}{
\begin{proof}
  Let $(P, V, \budget, \baseline)$ be a \MesApr{} PB game with
  plurality ballots.  In \MesApr{}, like in every \Mes{} rule, each
  voter receives $\nicefrac{\budget}{|V|}$ money for the whole
  election process.  As each voter approves only one project, each
  project $p_i$ can request at most
  $M(p_i) = \nicefrac{|A(p_i)| \cdot \budget}{|V|}$ from its
  supporters, and will report this cost as they would not spend
  any money on other projects.  Thus, if $\baseline(p_i) \leq M(p_i)$,
  project $p_i$ reports $M(p_i)$ and gets selected with its maximum possible
  cost.  Otherwise, if $\baseline(p_i) > M(p_i)$,  $p_i$ cannot
  be selected with cost covering  $\baseline(p_i)$, so $p_i$
  reports $\baseline(p_i)$ and is not selected.
\end{proof}
}

\begin{theorem}\label{thm:mesAprNEPartyListZeroDeliveryCost}
For every PB game %
with party-list ballots and zero
delivery costs, there is a polynomial-time computable \MesAprNE{}
regardless of the internal tie-breaking order. 
\end{theorem}

\appendixproof{Theorem}{thm:mesAprNEPartyListZeroDeliveryCost}{%
\begin{proof}
    Let $(P, N, \budget, \baseline)$ be a \MesApr{} PB game with party-list
		ballots and zero delivery costs.
    
    At the beginning, each voter receives $\nicefrac{\budget}{|V|}$ money.
    In a party-list profile, as every two voters have either the same preferences or the disjoint ones, $A(p_i) = A(\party(p_i))$.
    The total money the supporters of $\party(p_i)$ have is $M(\party(p_i)) = |A(p_i)| \cdot \nicefrac{\budget}{|V|} = \frac{|A(p_i)| \cdot \budget}{|V|}$. 
    Therefore, the projects in $\party(p_i)$ need to somehow distribute this money between them in such a way that no project would complain and change its cost to obtain better outcome.
    
		Since we have cardinal utilities, the $\alpha$-value of project $p_i$ is
		initially $\alpha_1(p_i) = \frac{\stratprof(p_i)}{|A(\party(p_i))|}$.
    For this reason, \MesApr{} will start from the cheapest project in the party to the most expensive one (following tie-breaking order in case of a tie) and equally take the budget from the party supporters to fund the next project unless the cost exceeds the money they still have.
    
    This means that each project should ask for the equal part of the money of party supporters, that it, $\frac{M(\party(p_i))}{|\party(p_i)|}$. 
    The project asking for more money would be moved to the end of the order when its supporters have insufficient money to fund it, so no project would increase its cost. 
    Asking for less money is pointless if it is already funded.

    For this reason, for the profile \stratprof{} where each project says $\stratprof(p_i) = \frac{M(\party(p_i))}{|A(\party(p_i)|}$, we obtain \MesAprNE{}.
\end{proof}%
}

Now we show an example with only three projects and one voter for which
\MesAprNE{} does not exist (recall that each profile with one voter is
also a party-list profile).

\begin{example}\label{ex:mesAprNoNEPartyList}
  Take a PB game~$G$ with one voter approving projects $p_1$, $p_2$,
  and $p_3$. Let $\budget=6$, $\baseline(p_1) = 3$,
  $\baseline(p_2) = 0$, and $\baseline(p_3) = 0$, with
  $p_1 \succ p_2 \succ p_3$.  We show that there is no \MesApr-NE
  in~$G$ (the full proof is in the appendix).

  Sketch of an argument: Suppose that $\stratprof$ is an NE. Then
  $\stratprof(p_1) \geq 3$ (otherwise, as it is selected, $p_1$
  improves by reporting $3$) and
  $\stratprof(p_2), \stratprof(p_3) < \stratprof(p_1) = 3 =
  \nicefrac{B}{2}$ (otherwise, the more expensive one is not
  funded). Since $\stratprof(p_2) + \stratprof(p_3) < B$, $p_2$ may
  report $\nicefrac{(\stratprof(c_1) + \stratprof(c_2))}{2}$ and still
  be selected, so $\stratprof$ is not an NE.

   \appendixcorrectness{Example}{ex:mesAprNoNEPartyList}{
      Suppose towards contradiction that $\stratprof$ is such an
  equilibrium. W.l.o.g.,let
  $\stratprof(p_1) \geq \baseline(p_1) = 3$.  Otherwise, if $p_1$ would be selected under these costs (and, hence, received a negative
	payoff), then it would prefer to report $\baseline(p_1)$ and receive utility
	0.
  If it was
  not funded under $\stratprof$, then 
  $(\stratprof_{-i}, \baseline(p_1))$
  also
  would be an equilibrium. 
  Next, we observe that 
  $\stratprof(p_1) = 3$. Indeed, if $p_1$ reported a value greater
  than $3$, then one of the projects could obtain a better payoff: If
  $p_2$ or $p_3$ reported a value greater or equal to $\stratprof(p_1)$,
  then it would benefit by reporting a smaller one; if $p_2$ and
  $p_3$ reported values that sum up to more than $\budget$, then at
  least one of them would benefit by reporting a smaller cost, and if
  they reported values that sum up to at most $\budget$ then either
  (a)~one of them would benefit by reporting a larger cost, or (b) if they
  both reported $3$, $p_1$ would benefit by reporting~$3$.
  Hence, we have that $\stratprof(p_1) = 3 = \nicefrac{\budget}{2}$. Then,
  both $\stratprof(p_2) > 0$ and $\stratprof(p_3) > 0$ (the project
  reporting $0$ would benefit by reporting a larger
  number). Further, we have that
  $\stratprof(p_2) < \stratprof(p_1)$ and
  $\stratprof(p_3) < \stratprof(p_1)$. Indeed, if either $p_2$ or
  $p_3$ reported value greater or equal to $\stratprof(p_1)$ then it would not be selected and, hence, would benefit by
  reporting a smaller cost. So, 
  $\stratprof(p_2) + \stratprof(p_3) < \budget$.  But
  then either of them would benefit by reporting a
  higher cost (so that the sum of their costs would not exceed
  $\budget$). Thus $\stratprof$ is not a \MesApr-NE.

}
\end{example}

We note that the problem in Example~\ref{ex:mesAprNoNEPartyList} is
partially caused by the delivery costs and the tie-breaking order. If
$p_1$ were behind $p_2$ and $p_3$ in the tie-breaking order, then
$p_2$ and $p_3$ would have reported cost~$3$, which would be a
\MesAprNE{}. Now we show that for a party-list instance we can
efficiently compute an internal tie-breaking order and a
corresponding~\MesAprNE{} (the proof is in the
appendix). %

\begin{theorem}\label{thm:mesAprNEPartyListDeliveryCostSpecificTBOrder}
	For every PB game with party-list ballots and some corresponding A/D
	tie-breaking order, there is a~\MesAprNE{} computable in polynomial-time.
\end{theorem}

\appendixproof{Theorem}{thm:mesAprNEPartyListDeliveryCostSpecificTBOrder}{%
\begin{proof}
		Let $(P, V, \budget, \baseline)$ be a \MesApr{} PB game with party-list
		ballots.  Recall that at the beginning each voter receives
		$\nicefrac{\budget}{|V|}$ money.  Then, the total money the supporters of
		$\party(p_i)$ have is $M(\party(p_i)) = |A(p_i)| \cdot
		\nicefrac{\budget}{|V|} = \frac{|A(p_i)| \cdot \budget}{|V|}$. 

		Let~$\succ$ be the A/D tie-breaking order from the theorem statement.
		Note that for the case of party-list ballots $\succ$ orders the projects
		within the same party nondescendingly by the delivery costs. That is,
		within the same party cheaper projects preferred to the more expensive
		ones. In case of two projects within the same party with equal delivery
		costs, we assume without loss of generality that the project with a lower
		index is preferred (we can always relabel projects to achieve this
		condition). Notably, $\succ$ does not impose any specific order with
		respect to the delivery costs over any two projects of two different
		parties.
   
    Since in \MesApr{} we consider cardinal utilities, the $\alpha$-value of project $p_i$ is at the beginning: $\alpha_1(p_i) = \frac{\stratprof(p_i)}{|A(\party(p_i))|}$.
    For this reason, \MesApr{} will start from the cheapest project in the party to the most expensive one (following the specified tie-breaking order in case of a tie) and equally take the budget from the party supporters to fund the next project unless the cost exceeds the money they still have.
    
    Now we have to consider two cases:
    \begin{itemize}
        \item The last in the tie-breaking order project $p_j \in party(p_i)$ has delivery cost $\baseline(p_j)$ not exceeding $\frac{M(\party(p_i))}{|\party(p_i)|}$. Note that this case is similar to the one in Theorem \ref{thm:mesAprNEPartyListZeroDeliveryCost}, i.e., each project submits $\frac{M(\party(p_i))}{|\party(p_i)|}$, gets funded, and no project benefits from changing its cost.
        
        \item The last in the tie-breaking order project $p_j \in \party(p_i)$ has delivery cost $\baseline(p_j)$ exceeding $\frac{M(\party(p_i))}{|\party(p_i)|}$. 
        In such a case, if every project submitted cost $\baseline(p_j)$, then $p_j$ would not be selected as it is last in the tie-breaking order and cannot profitably decrease its cost. 
        For this reason, we set the cost of $p_j$ to be $\baseline(p_j)$, remove $p_j$ from consideration, and repeat this procedure as long as the product of the number of yet-remained projects and the delivery cost of the last in the tie-breaking project exceeds $M(\party(p_i))$.
        Suppose now that the procedure stopped and project $p_k$ was the last removed one. 
        Then all $y$ yet-remaining projects would say cost $min(\nicefrac{M(\party(p_i))}{y}, \baseline(p_k))$ and be given this funding.
        Saying more than $\baseline(p_k)$ would result in being overtaken by  project $p_k$ and thrown out of the outcome due to insufficient budget.
        Saying more than $\nicefrac{M(\party(p_i))}{y}$ would result in getting the last among remaining projects and asking for more money than left in the budget at that timestamp.
    \end{itemize}

    The combination of \stratprof{} profiles calculated in the above way for all parties is a \MesAprNE{} for the given instance. 
\end{proof}%
}

One might wonder whether there always is a tie-breaking order leading
to a \MesAprNE{}. Unfortunately, as in the case of \Phragmen{}, there
is a small game (depicted in the appendix)
with no \MesAprNE{}
irrespective of the
tie-breaking order.

\begin{proposition}\label{pro:mesAprNoNEGeneralRegardlessOfTB}
  There are games with zero delivery costs, for which there is no
  \MesAprNE{} regardless of the tie-breaking order, even if we have
  only $4$ projects and $16$ votes.
\end{proposition}
\appendixproof{Proposition}{pro:mesAprNoNEGeneralRegardlessOfTB}{%
\begin{proof}
    We create four projects $p_1, \ldots, p_4$ and $16$ voters $v_1, \ldots, v_{16}$. Project $p_1$ is approved by voters $v_1, \ldots, v_5$; project $p_2$ by voters $v_{13}, \ldots, v_{16}, v_1$; project $p_3$ by voters $v_5, \ldots, v_9$; and project $p_4$ by voters $v_9, \ldots, v_{13}$. We set the budget to be some positive integer $\budget$ and delivery costs to be $\baseline \equiv 0$. Please note that we do not specify tie-breaking as we prove nonexistence of \MesAprNE{} in any tie-breaking.
    
    \begin{figure}[t]
         \centering
         \scalebox{0.8}{
            \begin{tikzpicture}
                [->,shorten >=1pt,auto,node distance=1.2cm, semithick]
                \node[shape=circle,draw=black, fill=black] (AB)  {};
                \node[shape=circle,draw=black, fill=black, below of=AB] (A2)  {};
                \node[shape=circle,draw=black, fill=black, below of=A2] (A3)  {};
                \node[shape=circle,draw=black, fill=black, below of=A3] (A4)  {};
                \node[shape=circle,draw=black, fill=black, below of=A4] (AD)  {};
                
                \node[shape=circle,draw=black, fill=black, right of=AB] (B2)  {};
                \node[shape=circle,draw=black, fill=black, right of=B2] (B3)  {};
                \node[shape=circle,draw=black, fill=black, right of=B3] (B4)  {};
                \node[shape=circle,draw=black, fill=black, right of=B4] (BC)  {};
                
                \node[shape=circle,draw=black, fill=black, below of=BC] (C2)  {};              
                \node[shape=circle,draw=black, fill=black, below of=C2] (C3)  {};
                \node[shape=circle,draw=black, fill=black, below of=C3] (C4)  {};
                \node[shape=circle,draw=black, fill=black, below of=C4] (CD)  {};

                \node[shape=circle,draw=black, fill=black, right of=AD] (D2)  {};
                \node[shape=circle,draw=black, fill=black, right of=D2] (D3)  {};
                \node[shape=circle,draw=black, fill=black, right of=D3] (D4)  {};
                \node[shape=circle,draw=black, fill=black, right of=D4] (CD)  {};
                
                \node[draw,  dashed, fit=(AB) (AD) ](FIt1) {};
                \node[draw,  dashed, fit=(AB) (BC) ](FIt2) {};
                \node[draw,  dashed, fit=(BC) (CD) ](FIt3) {};
                \node[draw,  dashed, fit=(AD) (CD) ](FIt4) {};
                
                \node[shape=circle, above of =FIt1, xshift = -3em, yshift= -3.5em, fontscale=2] (X1)  {$p_2$};
                
                \node[shape=circle, above of =FIt2, xshift = -0em, yshift= -1.5em, fontscale=2] (X2)  {$p_1$};
                \node[shape=circle, right of =FIt3, xshift = -0em, yshift= -0em, fontscale=2] (X3)  {$p_3$};
                \node[shape=circle, below of =FIt4, xshift = -0em, yshift= -0em, fontscale=2] (X4)  {$p_4$};
            
            \end{tikzpicture}}
         \caption{Approval sets of projects $p_1, p_2, p_3$, and $p_4$. Here, each vertex represents a single voter.}\label{fig:mesAprNoNECircle}
    \end{figure}

    For a better visualization of the instance, please look at the Figure \ref{fig:mesAprNoNECircle}. In particular, one can see that the instance is symmetric (each project is approved by three voters supporting only it and by two voters, one shared with each neighbour).

    Now we will show that there is no \MesAprNE{} for this instance. Suppose towards contradiction that some profile $\stratprof$ is a \MesAprNE{}.

    At the beginning, each voter receives $b = \frac{\budget}{|V|} = \frac{\budget}{16}$ money. 
    Therefore, a) no project says cost exceeding $5b = \frac{5\budget}{16}$ (otherwise its supporters would never have enough money to buy it, and b) no project says less than $3b = \frac{3\budget}{16}$ (as each project has three supporters that contribute only to it, it is nonoptimal to say less than they have).
    Further, as each project has zero delivery costs, each project must be selected in \MesAprNE{}, otherwise an unselected project could have decreased its cost to $0$ and be selected.
    For this reason, $\sum_{i=1}^{4}{\stratprof(p_i)} \leq B$.
    
    For the sake of brevity, for yet-unselected project $p_i$ in iteration $j$, we denote $\alpha_j(p_i)$ as minimum value such that project $p_i$ is $\alpha_j(p_i)$-affordable according to \MesApr{} rule (or infinity, if such a value does not exist). 
    By $\epsilon \in \reals_+$ we mean some very small positive real number such that adding it does not change the given strict inequality sign.
    
    Due to symmetry, w.l.o.g. we can assume that 1) project $p_1$ says the lowest cost (i.e., $\stratprof(p_1) \leq \stratprof(p_2), \stratprof(p_3), \stratprof(p_4)$) and wins ties if there is more than more project with this cost as well as 2) $\stratprof(p_2) \leq \stratprof(p_3)$ and $p_2$ wins ties with $p_3$. If it was not the case, we can shift or rotate projects and perform analogous reasoning.

    Before we move on, let us exclude some cases in which $\stratprof$ cannot be \MesAprNE{}. 
    
    We point out that if there is no project of the same cost as $p_1$, then $p_1$ could have increased its cost by some positive number $\epsilon$ while still being considered first and getting more, so $\stratprof$ could not have been \MesAprNE{} in this case.
    Therefore, $p_1$ has the same cost either as 1) $p_4$, or as 2) $p_2$ (please note $\stratprof(p_1) = \stratprof(p_3)$ implies $\stratprof(p_1) = \stratprof(p_2)$ due to assumption that $\stratprof(p_2) \leq \stratprof(p_3)$). 

    With our assumptions, project $p_1$ is selected in the first iteration and each of its supporters pays $\frac{\stratprof(p_1)}{5}$ for $p_1$. As $b - \frac{\stratprof(p_i)}{5} \leq \frac{\budget}{16} - \frac{\frac{3\budget}{16}}{5} = \frac{\frac{2\budget}{16}}{5} < \frac{\frac{3\budget}{16}}{5} \leq \frac{\stratprof(p_j)}{5}$ for any $p_i, p_j \in P$, voter $v_1$ has insufficient money to contribute the equal share $\frac{\stratprof(p_2)}{5}$ to purchase $p_2$.
    Therefore, in the second iteration, $\alpha_2(p_4) = \frac{\stratprof(p_4)}{5}$, $\alpha_2(p_2) = \frac{\stratprof(p_2) - (b - \frac{\stratprof(p_1)}{5})}{4}$, and $\alpha_2(p_3) = \frac{\stratprof(p_3) - (b - \frac{\stratprof(p_1)}{5})}{4} \geq \alpha_2(p_2)$.
    Please note that $\alpha_2(p_2) = \frac{\stratprof(p_2) - (b - \frac{\stratprof(p_1)}{5})}{4} \geq \frac{\stratprof(p_1) - (\frac{\budget}{16} - \frac{\stratprof(p_1)}{5})}{4} = \frac{6\stratprof(p_1)}{20} - \frac{\budget}{64} = \frac{4\stratprof(p_1)}{20} + (\frac{2\stratprof(p_1)}{20} - \frac{\budget}{64}) \geq \frac{\stratprof(p_1)}{5} + (\frac{2 \cdot \frac{3\budget}{16}}{20} - \frac{\budget}{64}) = \frac{\stratprof(p_1)}{5} + \frac{6\budget-5\budget}{16 \cdot 20} = \frac{\stratprof(p_1)}{5} + \frac{\budget}{16 \cdot 20} > \frac{\stratprof(p_1)}{5} = \alpha_1(p_1)$. 
    Thus, as $\alpha_2(p_2) > \alpha_1(p_1)$, project $p_4$ would regret saying the same cost as project $p_1$. Indeed, selecting $p_1$ in the first iteration significantly increases $\alpha_2(p_2)$ and $\alpha_2(p_3)$, so if $p_4$ said $\stratprof(p_4) = \stratprof(p_1)$, then it would benefit from slightly increasing the cost by $\epsilon$ (it would still be before $p_2$, but it would earn more). 
    For this reason, it cannot be the case that $\stratprof(p_1) = \stratprof(p_4)$, so due to the former reasoning we know that $\stratprof(p_1) = \stratprof(p_2)$ and $\stratprof(p_1) < \stratprof(p_4)$.

    So we know that: $\stratprof(p_1) = \stratprof(p_2)$, $\stratprof(p_1) <
		\stratprof(p_4)$, and $\stratprof(p_2) \leq \stratprof(p_3)$. We have two cases to consider:
    \begin{itemize}
        \item Project $p_4$ is selected in the second iteration. 
            It means that $\alpha_2(p_4) \leq \alpha_2(p_2)$. 
            We observe that it must be $\alpha_2(p_4) = \alpha_2(p_2)$, otherwise it would be beneficial for $p_4$ to slightly increase its cost by $\epsilon$ in such a way that $p_4$ is still considered before $p_2$ and gains more.
            
            In this case, as $p_2$ and $p_3$ are approved by disjoint sets of voters and both $p_1$ and $p_4$ took the same amount of money from the same number of their supporters, both $p_2$ and $p_3$ should request for the same amount of money, that is, all money that their supporters have. 
            Thus $\stratprof(p_3) = \stratprof(c_2) = 3 \cdot b + (b - \frac{\stratprof(p_1)}{5}) + (b - \frac{\stratprof(p_4)}{5}) = 5b - \frac{\stratprof(p_1)}{5} - \frac{\stratprof(p_4)}{5}$ and $\stratprof(p_4) = 5 \cdot (5b - \frac{\stratprof(p_1)}{5} - \stratprof(p_2)) = 25b - 5\stratprof(p_2) - \stratprof(p_1)$.
            
            However, we know that $\stratprof(p_1) = \stratprof(p_2)$, so $\stratprof(p_1) = \stratprof(p_2) = \stratprof(p_3) = 5b - \frac{\stratprof(p_1)}{5} - \frac{\stratprof(p_4)}{5}$, which implies that $\stratprof(p_4)  = 25b - 5\stratprof(p_2) - \stratprof(p_1) = 25b - 6\stratprof(p_1)$.
            Therefore, $\alpha_2(p_4) = \alpha_2(p_2)$ means that $\frac{\stratprof(p_4)}{5} = \frac{\stratprof(p_2) - (b - \frac{\stratprof(p_1)}{5})}{4} = \frac{6\stratprof(p_1)}{20} - \frac{b}{4}$, which is equivalent to $\frac{25b - 6\stratprof(p_1)}{5} = \frac{6\stratprof(p_1)-5b}{20} \Rightarrow 100b - 24\stratprof(p_1) = 6\stratprof(p_1) - 5b \Rightarrow 30\stratprof(p_1) = 105b \Rightarrow \stratprof(p_1) = \frac{7b}{2}$. Then we have $\stratprof(p_1) = \stratprof(p_2) = \stratprof(p_3) = \frac{7b}{2}$ and $\stratprof(p_4) = 25b - 6\stratprof(p_1) = 25b - 21b = 4b$.
            
            Now imagine that project $p_1$ changes its cost from $\stratprof(p_1) = \frac{7b}{2}$ to $\stratprof'(p_1) = \frac{36b}{10}$.
            Naturally, projects $p_2$ and $p_3$ are selected in first two iterations so voters $v_1, \ldots, v_5$ supporting $p_1$ are left with $b - \frac{\stratprof(c_2)}{5}, b, b, b, b - \frac{\stratprof(p_3)}{5}$ money respectively and they have $b - \frac{\stratprof(c_2)}{5} + b + b + b + b - \frac{\stratprof(p_3)}{5} = 5b - 2 \cdot \frac{7b}{10} = \frac{36b}{10}$ money in total which they can only spend on project $p_1$ as they disapprove $p_4$.
            So $p_1$ gets selected with $\stratprof'(p_1) = \frac{36b}{10} > \frac{7b}{2} = \stratprof(p_1)$ which indicates that $\stratprof$ could not be \MesAprNE{}.
            
        \item Project $p_2$ is selected in the second iteration.
            It means that $\alpha_2(p_2) \leq \alpha_2(p_4)$. 
            We observe that it must be $\alpha_2(p_2) = \alpha_2(p_4)$, otherwise it would be beneficial for $p_2$ to slightly increase its cost in such a way that $p_2$ is still considered before $p_4$ and gains more (please note that even if $p_3$ would jump before $p_2$, it takes money from the disjoint set of voters from $p_2$'s supporters, so it does not disprove this argument).
            
            By performing analogous calculations to the previous case and using $\stratprof(p_1) = \stratprof(p_2)$ we obtain that $\alpha_2(p_2)= \frac{\stratprof(p_2) - (b - \frac{\stratprof(p_1)}{5})}{4} = \frac{6\stratprof(p_1)}{20} - \frac{b}{4}$, so $\alpha_2(p_4) = \alpha_2(p_2)  \Rightarrow  \frac{\stratprof(p_4)}{5} = \frac{6\stratprof(p_1)}{20} - \frac{b}{4}  \Rightarrow  \stratprof(p_4) = \frac{6\stratprof(p_1)-5b}{4}$.

            After $p_1$ and $p_2$ are selected, voter $v_1$ has $0$ money left, voters $v_2, \ldots, v_5$ have $b - \frac{\stratprof(p_1)}{5}$ money left, voters $v_{13}, \ldots, v_{16}$ have $b - \alpha_2(p_2) = b - (\frac{6\stratprof(p_1)}{20} - \frac{b}{4}) = \frac{25b-6\stratprof(p_1)}{20}$ money left, and voters $v_6, \ldots, v_{12}$ have $b$ money left.

            We need to consider two cases:
            \begin{itemize}
                \item $p_3$ is selected before $p_4$.
                    Then $\stratprof(p_3) = \stratprof(p_2)$, otherwise (if $\stratprof(p_2) < \stratprof(p_3)$) $p_2$ would increase its cost by $\epsilon$ in such a way that $\stratprof(p_2) + \epsilon < \stratprof(p_3)$ and still be selected in the second iteration, but with more money.
                    Because $\stratprof(p_3) = \stratprof(p_2) = \stratprof(p_1)$ and the supporters of $p_2$ and $p_3$ are disjoint and symmetric, after selecting $p_3$, voter $v_5$ be left with no money whereas voters $v_6, \ldots, v_9$ with $b - \alpha_3(p_3) = b - \alpha_2(p_2) = \frac{25b-6\stratprof(p_1)}{20}$.
                    Further, in the fourth iteration, the supporters of $p_4$ have together $3b + 2 \cdot \frac{25b-6\stratprof(p_1)}{20} = \frac{30b+25b-6\stratprof(p_1)}{10} = \frac{55b-6\stratprof(p_1)}{10}$, so $p_4$ asks for the money they still have in total. 
                    
                    Therefore, $\stratprof(p_4) = \frac{6\stratprof(p_1)-5b}{4} = \frac{55b-6\stratprof(p_1)}{10}  \Rightarrow  30\stratprof(p_1)-25b = 110b-12\stratprof(p_1)  \Rightarrow  42\stratprof(p_1) = 135b  \Rightarrow  \stratprof(p_1) = \frac{135b}{42}$.
                    Thus $\stratprof(p_4) = \frac{6\stratprof(p_1)-5b}{4} = \frac{6 \cdot \frac{135b}{42} - 5b}{4} = \frac{135b - 35b}{4 \cdot 7} = \frac{100b}{28} = \frac{150b}{42}$.

                    Imagine what happens if $p_1$ increases its cost from $\stratprof(p_1) = \frac{135b}{42}$ to $\stratprof'(p_1) = \frac{156b}{42}$. 
                    Then $p_2$ gets selected first with cost $\stratprof(p_2) = \frac{135b}{42}$ and each of its supporters pays $\frac{\stratprof(p_2)}{5} = \frac{\frac{135b}{42}}{5} = \frac{27b}{42}$. Next, each of $p_3$'s supporters will pay $\frac{27b}{42}$ to purchase $p_3$ (voters supporting $p_2$ and $p_3$ are disjoint). After that, regardless whether $p_4$ is selected before $p_1$ or not, the supporters of $p_1$ have $3b + 2 \cdot (b - \frac{27b}{42}) = \frac{3 \cdot 42 + 2 \cdot 15b}{42} = \frac{156b}{42}$, so the supporters of $p_1$ have enough money to purchase $p_1$. 
                    We showed that $p_1$ could improve its utility by changing its cost, so the proposed \stratprof{} is not \MesAprNE{} in this case.

                \item $p_4$ is selected before $p_3$. 
                    It means that $\alpha_3(p_4) \leq \alpha_3(p_3)$.
                    As we know how much money each voter has in the third iteration, $\alpha_3(p_4) = \frac{\stratprof(p_4)-\frac{25b-6\stratprof(p_1)}{20}}{4} = \frac{20\stratprof(p_4)-25b+6\stratprof(p_1)}{80}$ and $\alpha_3(p_3) = \frac{\stratprof(p_3)-\frac{\stratprof(p_1)}{5}}{4} = \frac{20\stratprof(p_3)-4\stratprof(p_1)}{80}$.
                    After inserting the exact values we obtain that $\alpha_3(p_4) \leq \alpha_3(p_3)  \Rightarrow  \frac{20\stratprof(p_4)-25b+6\stratprof(p_1)}{80} \leq \frac{20\stratprof(p_3)-4\stratprof(p_1)}{80}  \Rightarrow  20\stratprof(p_3) \geq 20\stratprof(p_4)-25b+10\stratprof(p_1)  \Rightarrow  \stratprof(p_3) \geq \stratprof(p_4) + \frac{\stratprof(p_1)}{2} - \frac{5b}{4}$.
                    However, $\stratprof(p_3) \geq \stratprof(p_4) + \frac{\stratprof(p_1)}{2} - \frac{5b}{4} > \stratprof(p_4) + \frac{3b}{2} - \frac{5b}{4} = \stratprof(p_4) + \frac{b}{4} > \stratprof(p_1) + \frac{b}{4} = \stratprof(p_2) + \frac{b}{4}$, so $\stratprof(p_3)$ is significantly greater than $\stratprof(p_2)$.
                    In other words, even when voter $v_9$ pays a greater share to buy $p_4$ and is left with less money than $\frac{\stratprof(p_4)}{5}$, there is still enough money to purchase $p_3$. 
                    Therefore, $p_2$ can increase its cost from $\stratprof(p_2) = \stratprof(p_1)$ to $\stratprof(p_3)-\epsilon$ and lose in the second iteration with $p_4$, but still be funded before $p_3$ and obtain more money. 
                    For this reason, $\stratprof$ is not \MesAprNE{}.
            \end{itemize}

    \end{itemize}

    We showed that in no case \MesAprNE{} exists for this instance. The remaining tie-breaking order are completely analogous due to the instance's symmetry.
    
    It means that for the given instance there is no tie-breaking for which \MesAprNE{} exists.
\end{proof}%
}

\MesApr{} and  \Phragmen{} might appear similar, as we have guarantees of NE existence
only for very restricted domains, examples with nonexisting NE for
given delivery costs, %
and small examples with no NE regardless of the
tie-breaking order. However, despite optimizing a similar function
(i.e., minimizing the maximum amount of money one voter pays to
buy a project), \MesApr{} and \Phragmen{} differ in terms of
both the NE existence and the obtained equilibria. For example,
there are instances in
which symmetric projects (i.e.,\,projects that have exactly the same sets of
supporters) ask for different amounts in a \MesAprNE{}, but report identical
costs under $\Phragmen$. The proof is in the appendix.

\begin{theorem}\label{thm:mesAprUnsymmetricCostsForSymmetricProjects}
  In some instances symmetric projects (w.r.t. their supporters) submit different costs in a~\MesAprNE{}.
\end{theorem}

\appendixproof{Theorem}{thm:mesAprUnsymmetricCostsForSymmetricProjects}{%
\begin{proof}
     Our example is in fact the example from Proposition \ref{pro:phrag-laminar-no-ne} narrowed to prim voters and projects.

    We have three voters $v_1, v_2$, and $v_3$ as well as three projects $p_1, p_2$, and $p_3$.
    Projects $p_1$ and $p_2$ are approved by, respectively, voters $v_1$ and $v_2$, while project $p_3$ is approved by all three voters.
    We set the budget to be some positive integer denoted as $\budget$ and the delivery costs of these projects to be $\baseline \equiv 0$.
    We set the tie-breaking between projects to be $p_1 \succ p_2 \succ p_3$.
    For the sake of brevity, for yet-unselected project $p_i$ in iteration $j$, we denote $\alpha_j(p_i)$ as minimum value such that project $p_i$ is $\alpha_j(p_i)$-affordable according to \MesApr{} rule (or infinity, if such a value does not exist). 

    Let $(\stratprof(p_1) = \frac{7\budget}{36}, \stratprof(p_2) = \frac{8\budget}{36}, \stratprof(p_3) = \frac{21\budget}{36})$ be the costs the projects say. We will show that it is \MesAprNE{}.

    At the beginning, each of three voters receives $\frac{\budget}{3} = \frac{12\budget}{36}$ money.

    In the first iteration, $\alpha_1(p_1) = \frac{\nicefrac{7\budget}{36}}{1} = \frac{7\budget}{36}, \alpha_1(p_2) = \frac{\nicefrac{8\budget}{36}}{1} = \frac{8\budget}{36}, \alpha_1(p_3) = \frac{\nicefrac{21\budget}{36}}{3} = \frac{7\budget}{36}$, so $p_1$ and $p_3$ are tied, while $p_2$ has greater $\alpha$-value and is skipped for now. Thus, due to tie-breaking, we select $p_1$ and voter $v_1$ pays $\frac{7\budget}{36}$ for $p_1$.

    In the second iteration, $v_1$ has $\frac{5\budget}{36}$ money while both $v_2$ and $v_3$ have $\frac{12\budget}{36}$. Thus, since the budget of $v_2$ did not change, $\alpha_2(p_2) = \alpha_1(p_2) = \frac{8\budget}{36}$. In order to purchase $p_3$, voter $v_1$ would need to spend all left money whereas $v_2$ and $v_3$ would need to pay more to make up $v_1$'s insufficient money. More precisely, $\alpha_2(p_3) = \frac{\nicefrac{21\budget}{36} - \nicefrac{5\budget}{36}}{2} = \frac{\nicefrac{16\budget}{36}}{2} = \frac{8\budget}{36}$. So projects $p_2$ and $p_3$ are tied and we select $p_2$ due to tie-breaking (by taking money from $v_2$).

    Finally, in the third iteration, we will select $p_3$ as its supporters have $\frac{5\budget}{36} + \frac{4\budget}{36} + \frac{12\budget}{36} = \frac{21\budget}{36} = \stratprof(p_3)$, all voters will use use all their money for it.

    We showed that each project is selected with the proposed costs so no project benefits from lowering its costs. Next, $p_3$ will not increase its cost as its supporters would not have enough money to buy it in the last iteration. Further, $p_2$ will not increase its cost as it would lose with $p_3$ in the second iteration and its only supporter $v_2$ would pay $\frac{8\budget}{36}$ for $p_3$, leaving insufficient $\frac{4\budget}{36}$ for $p_2$. Analogously, $p_1$ will not increase its cost as it would lose with $p_3$ in the first iteration and its only supporter $v_1$ would pay $\frac{7\budget}{36}$ for $p_3$, leaving insufficient $\frac{5\budget}{36}$ for $p_1$.

    For this reason, $\stratprof = (\stratprof(p_1) = \frac{7\budget}{36}, \stratprof(p_2) = \frac{8\budget}{36}, \stratprof(p_3) = \frac{21\budget}{36})$ is a \MesAprNE{} for the given instance.

    One can show in an analogous way that $\stratprof' = (\stratprof'(p_1) = \frac{8\budget}{36}, \stratprof'(p_2) = \frac{7\budget}{36}, \stratprof'(p_3) = \frac{21\budget}{36})$ is also \MesAprNE{} for the given instance.
\end{proof}%
}

\section{Experiments}

\begin{table}[t]\centering
 \begin{tabular}{r | c c c c c c }
          Instance & $|P|$ & $|V|$ & Budget & \hspace{-0.14cm} Vote Len. \hspace{-0.1cm} & Rule\\
          \midrule
         Wesola & 29 & 1182 & \multicolumn{1}{r}{~$1 011$k PLN} & 7.87 &\textbasicAV \\
         Kleine Wereld & 52 & 426 & \multicolumn{1}{r}{250k EUR} & 11.93  & \textbasicAV\\
          \midrule%
    \end{tabular}
  \caption{PB instances analyzed in our experiments. ``Vote
    Len.'' means the average number of approvals in a ballot. ``Rule''
    shows the rule that was used in the original election.}
     \label{tab:pb_details}
 \end{table}

We move on to the experimental analysis of our games for real-life PB
instances from
Pabulib~\citep{fal-fli-pet-pie-sko-sto-szu-tal:c:pabulib}. We focus on two instances, one held in 2022 in Wesola (a district of
Warsaw, Poland) and one held in 2019 in Kleine Wereld (a part of %
the Noord District of Amsterdam, The Netherlands). A few details
regarding these instances are available in \Cref{tab:pb_details}.  We
have also analyzed many other instances from Pabulib. We show some of
them in \Cref{app:experiments}. While all of them lead to similar
conclusions, the instances we chose are 
particularly illustrative.

We perform two experiments. First, we take a given PB
instance and the original costs reported for the projects.
Then, for each project we compute its best response.
Second, we seek equilibiria in the games defined by our
instances (with zero delivery costs).

\newcommand\marginplotwidth{4.1cm}
\newcommand\gameplotwidth{4.1cm}

\begin{figure}[t!]
     \centering
     \includegraphics[width=\marginplotwidth]{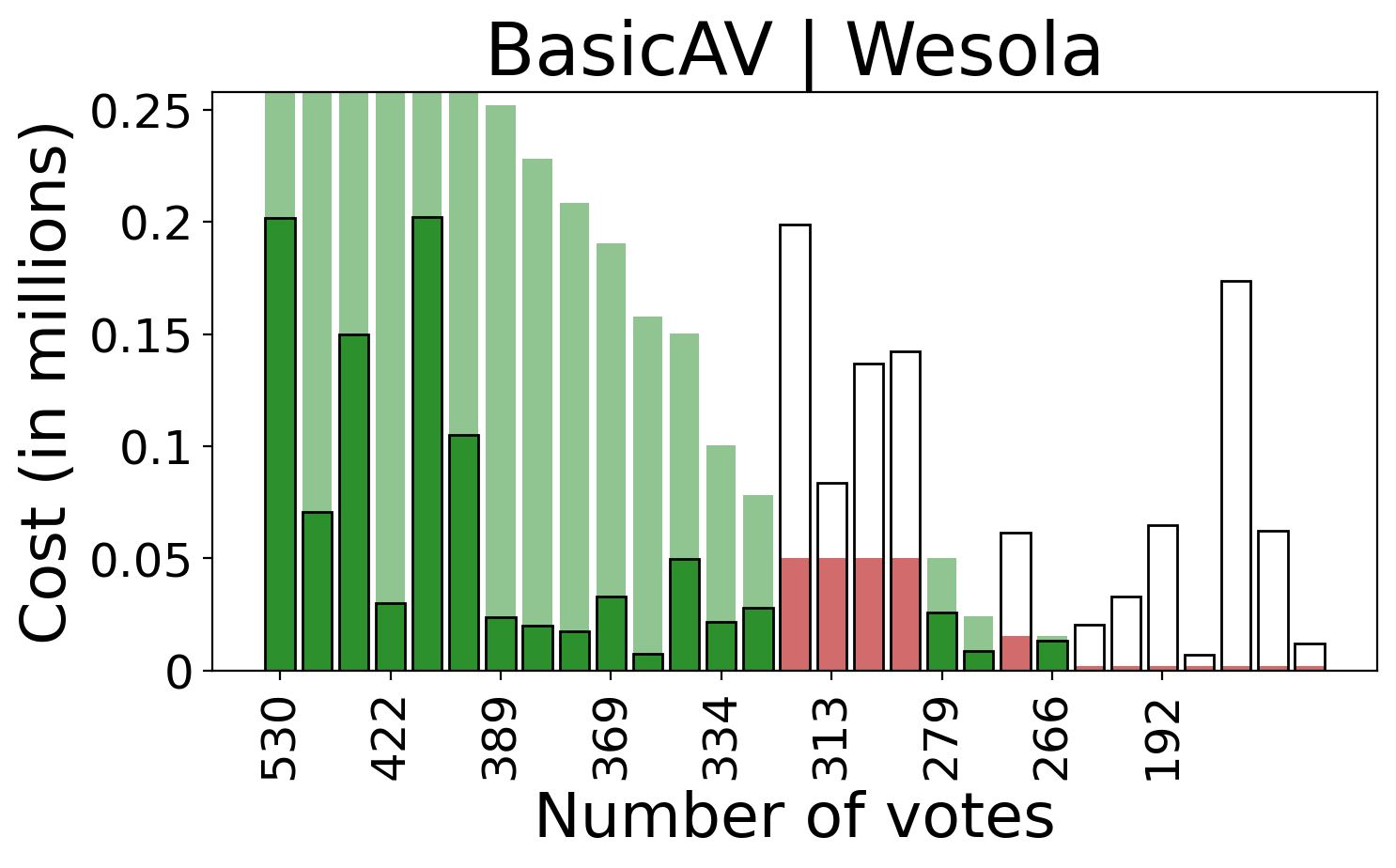}%
         \hspace{0mm}
    \includegraphics[width=\marginplotwidth]{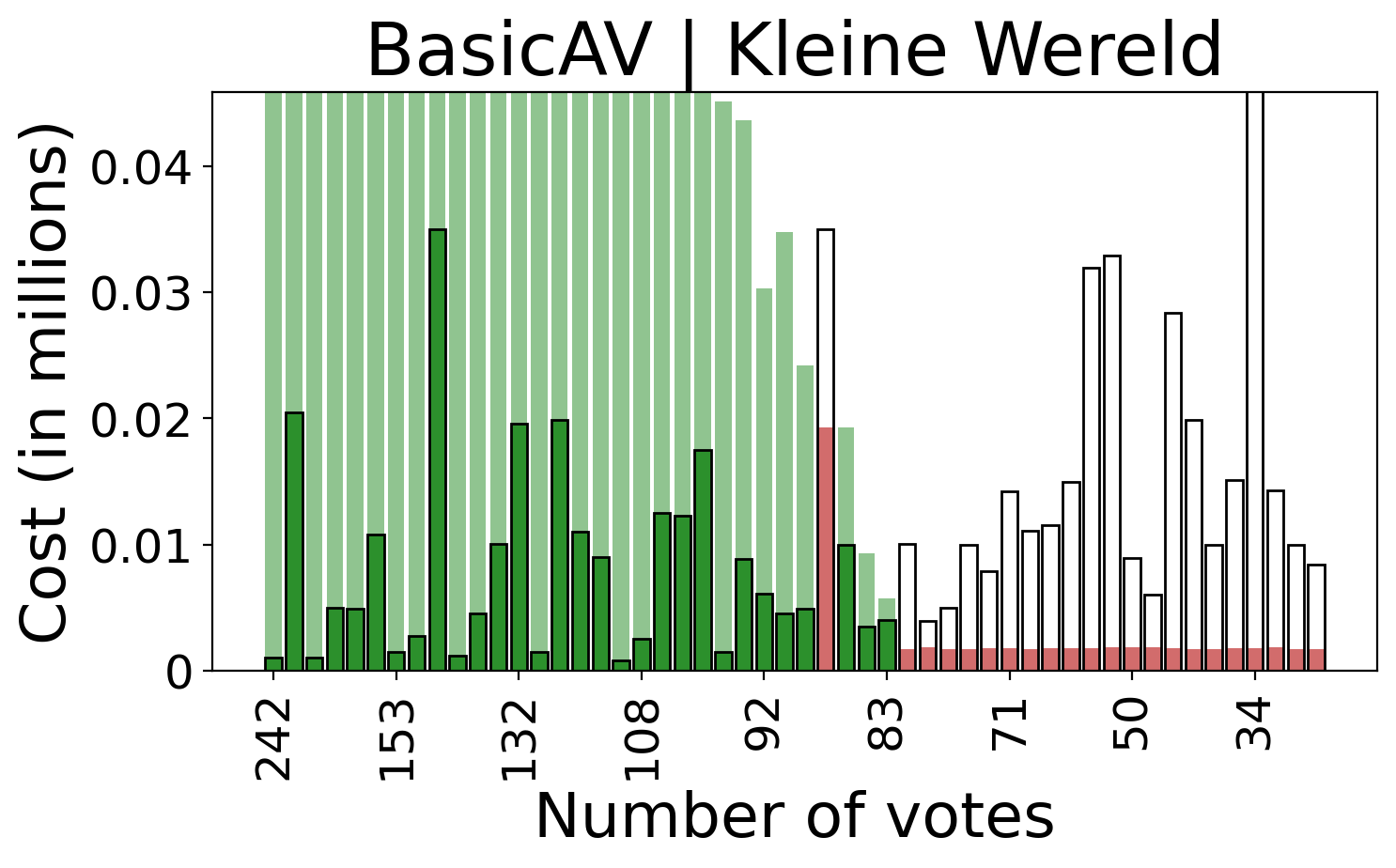}
   
   \includegraphics[width=\marginplotwidth]{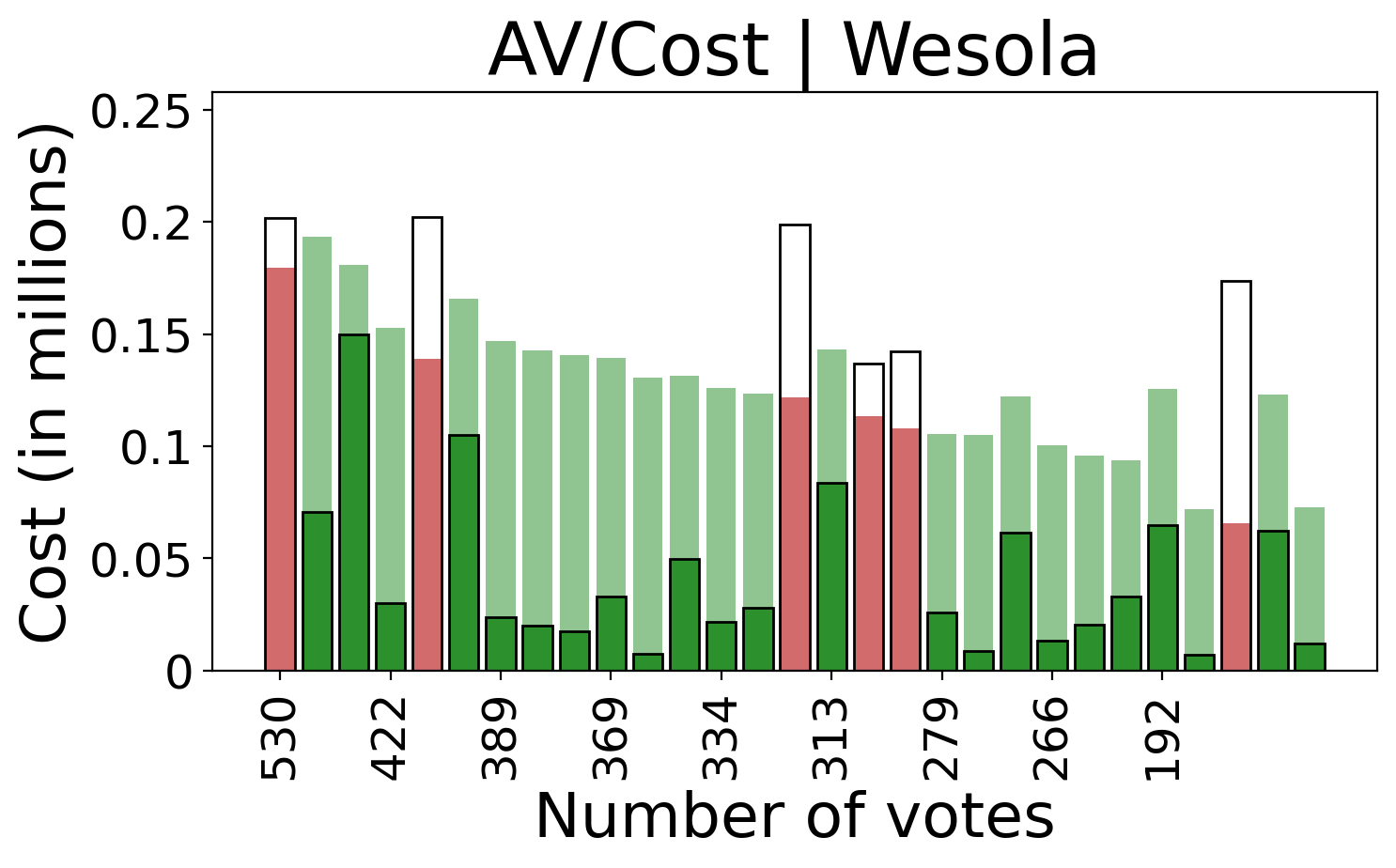}%
         \hspace{0mm}
    \includegraphics[width=\marginplotwidth]{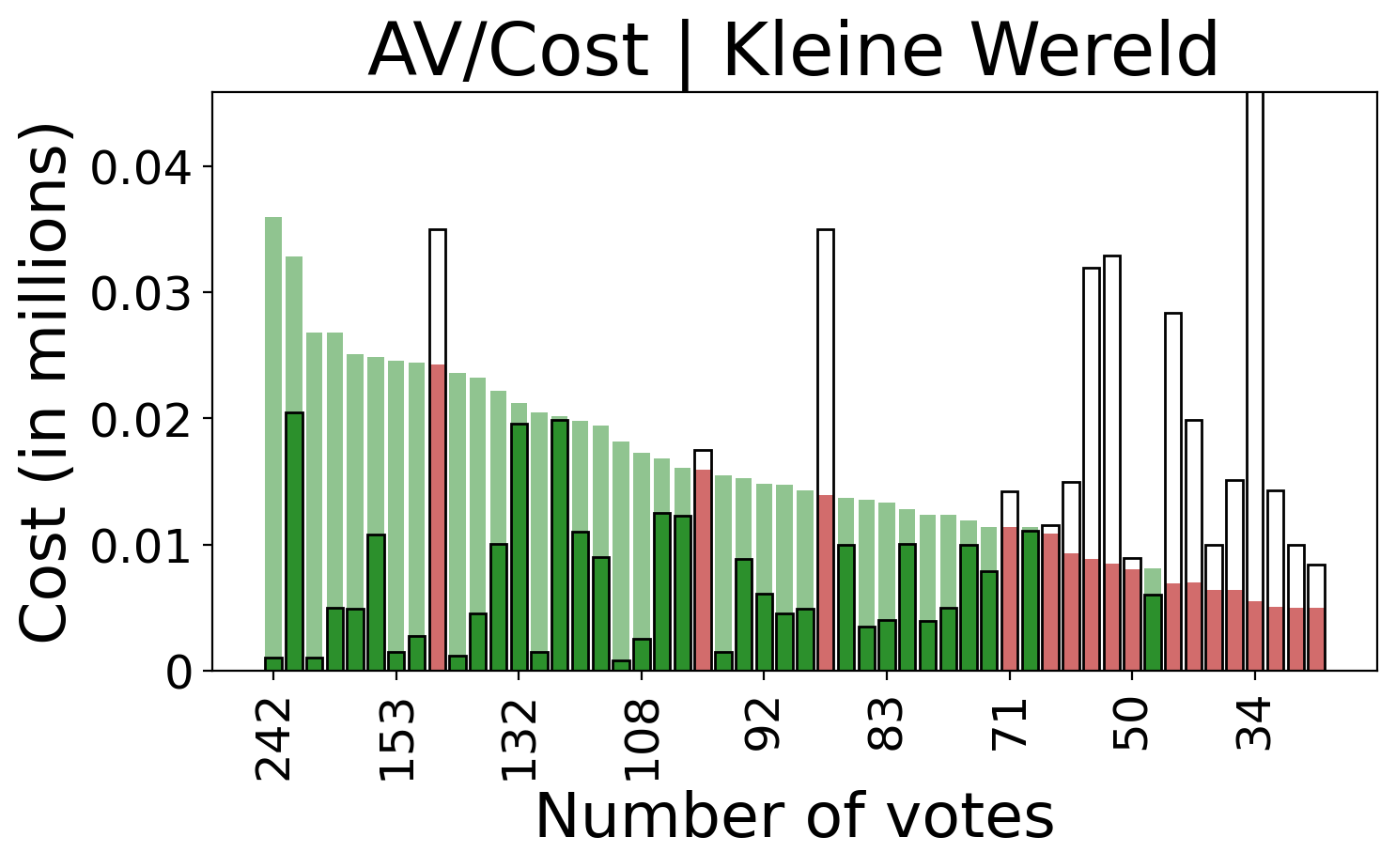}
    
    \includegraphics[width=\marginplotwidth]{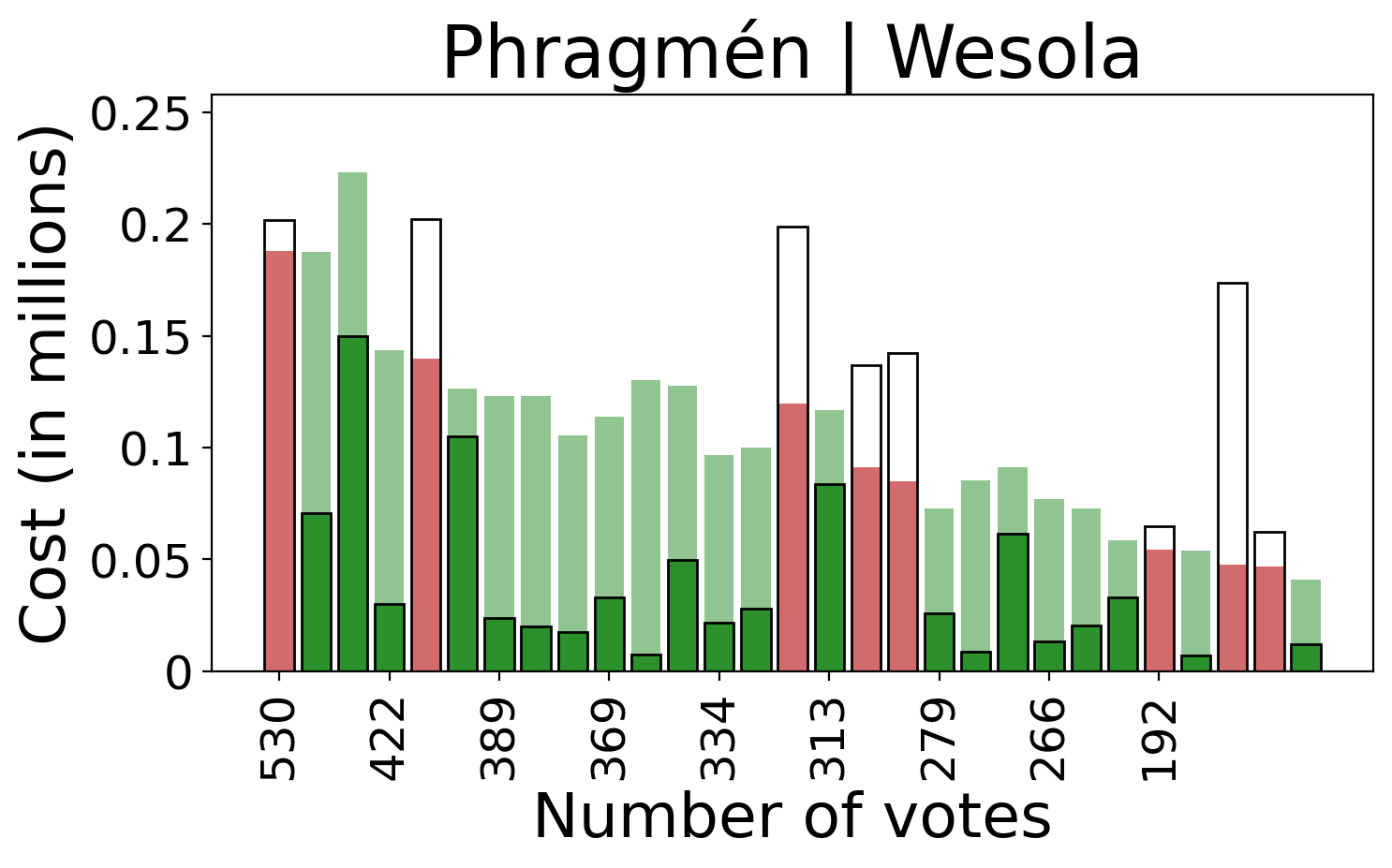}%
         \hspace{0mm}
    \includegraphics[width=\marginplotwidth]{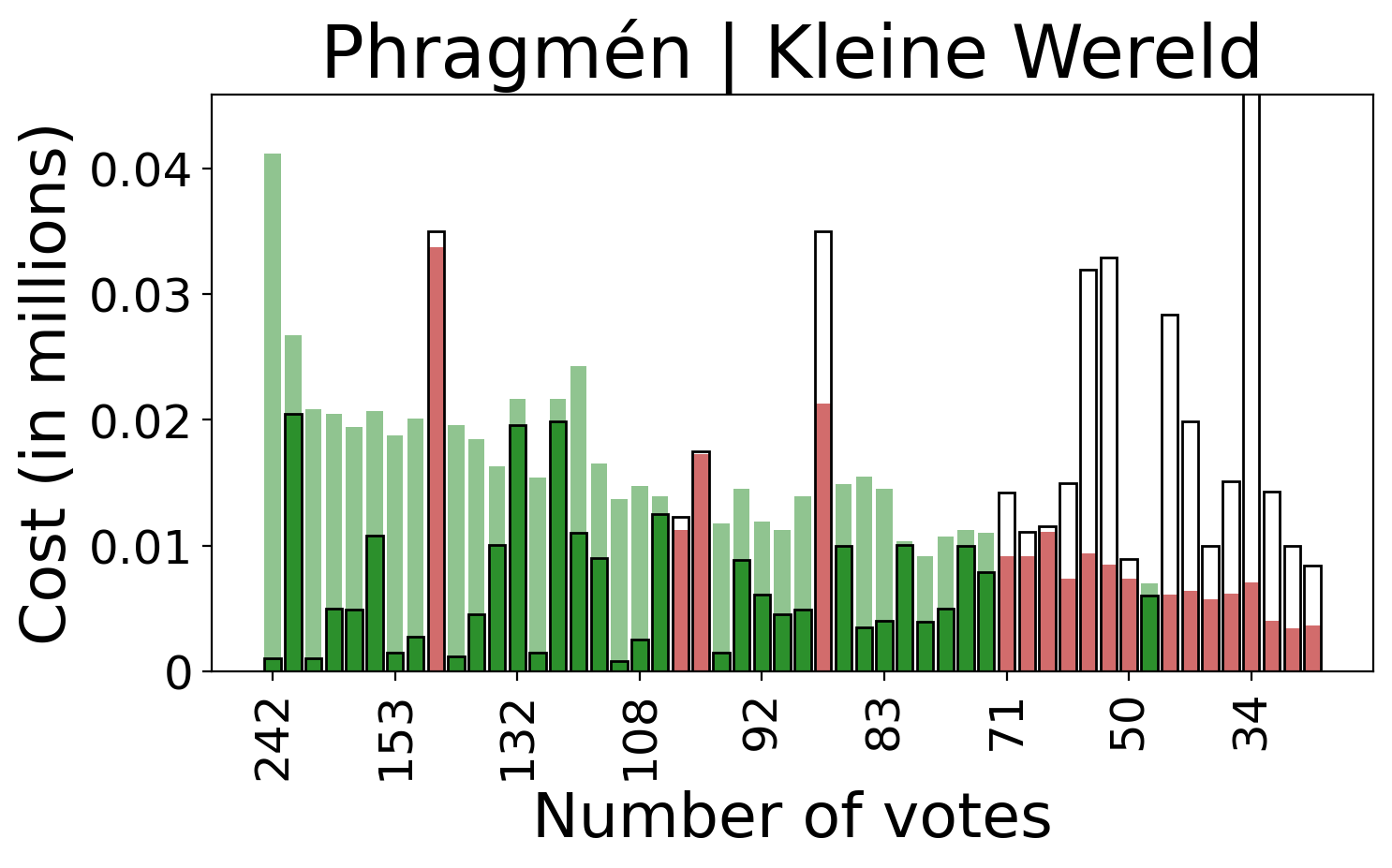}

    \includegraphics[width=\marginplotwidth]{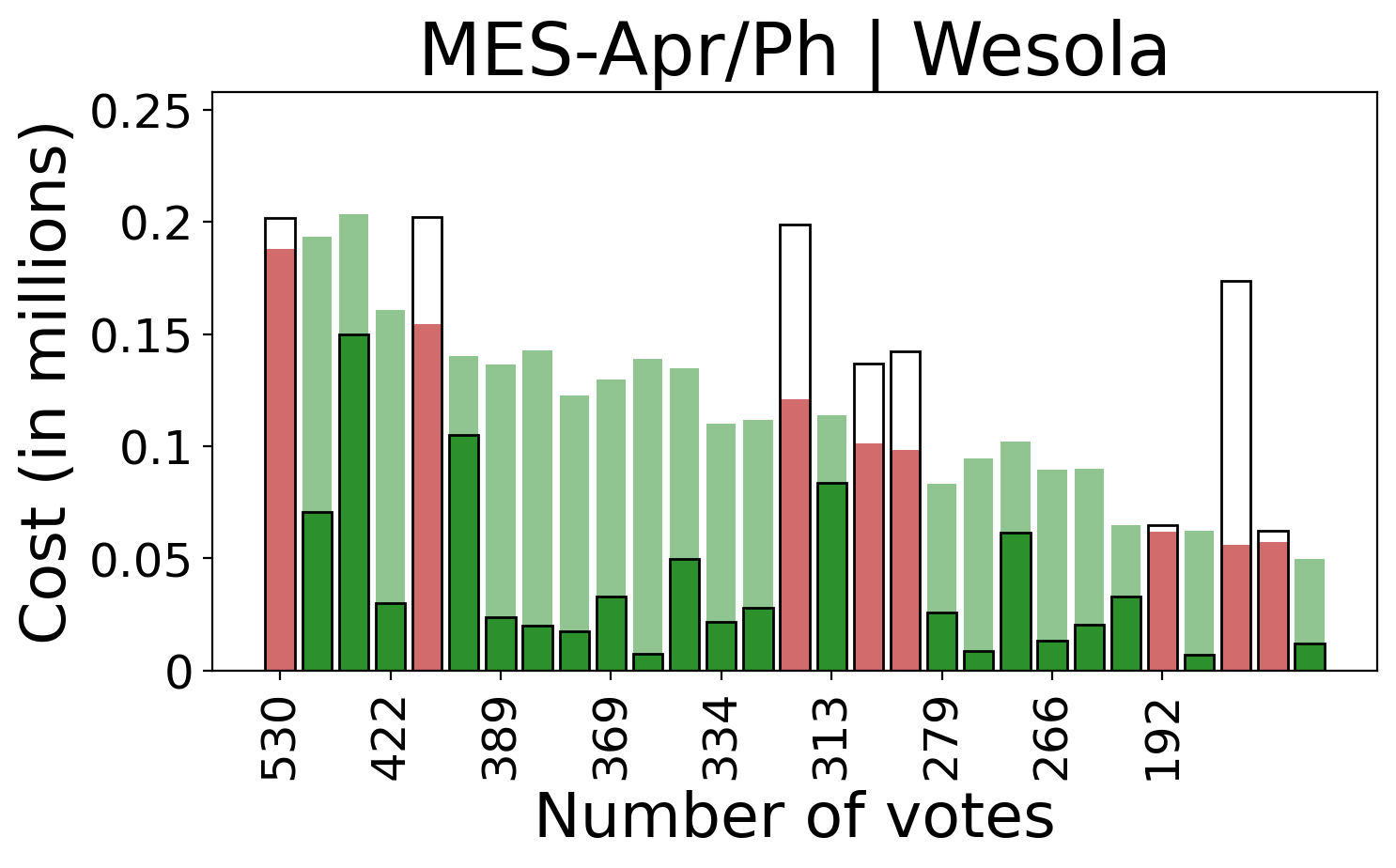}%
         \hspace{0mm}
    \includegraphics[width=\marginplotwidth]{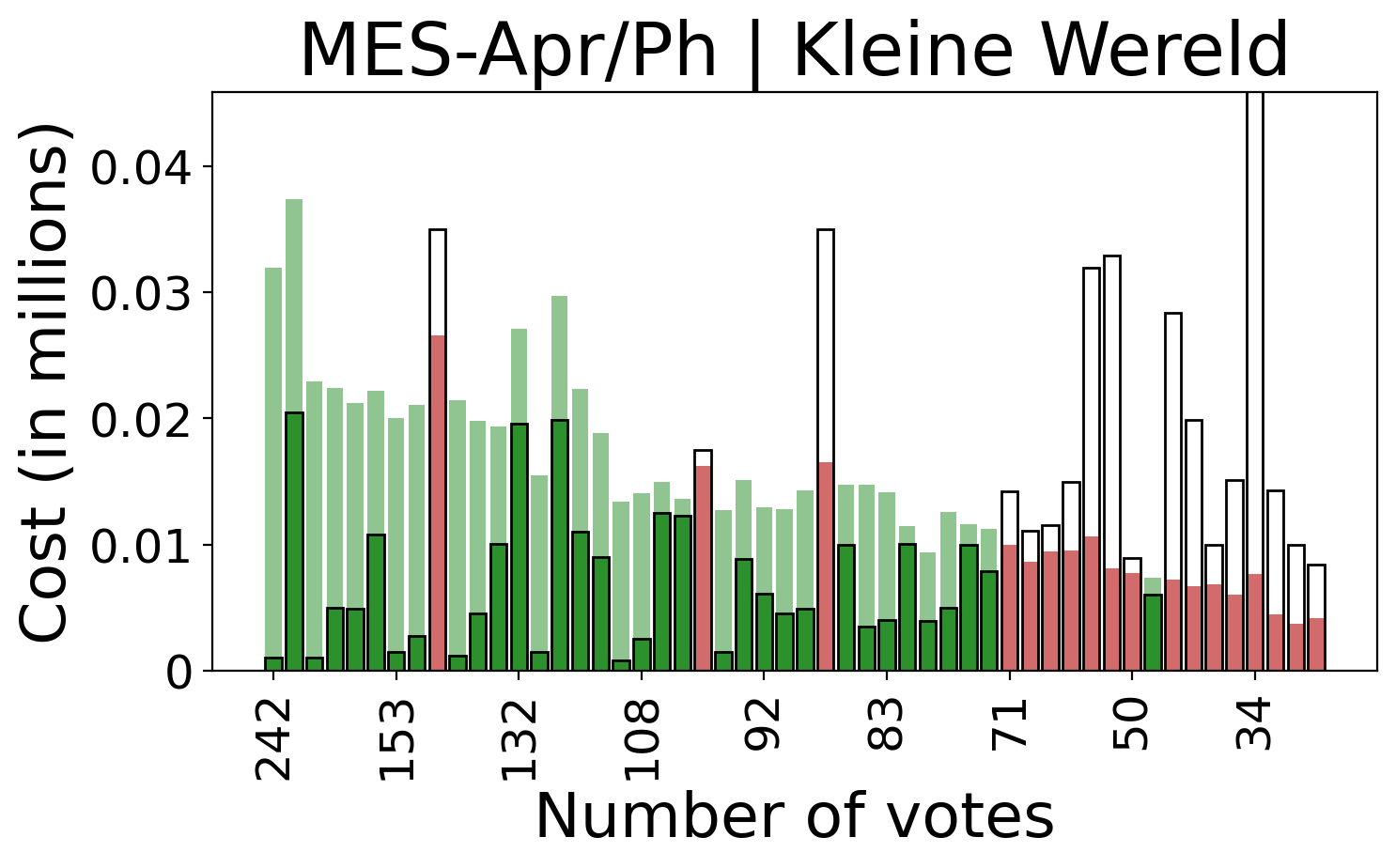}
    
    \includegraphics[width=\marginplotwidth]{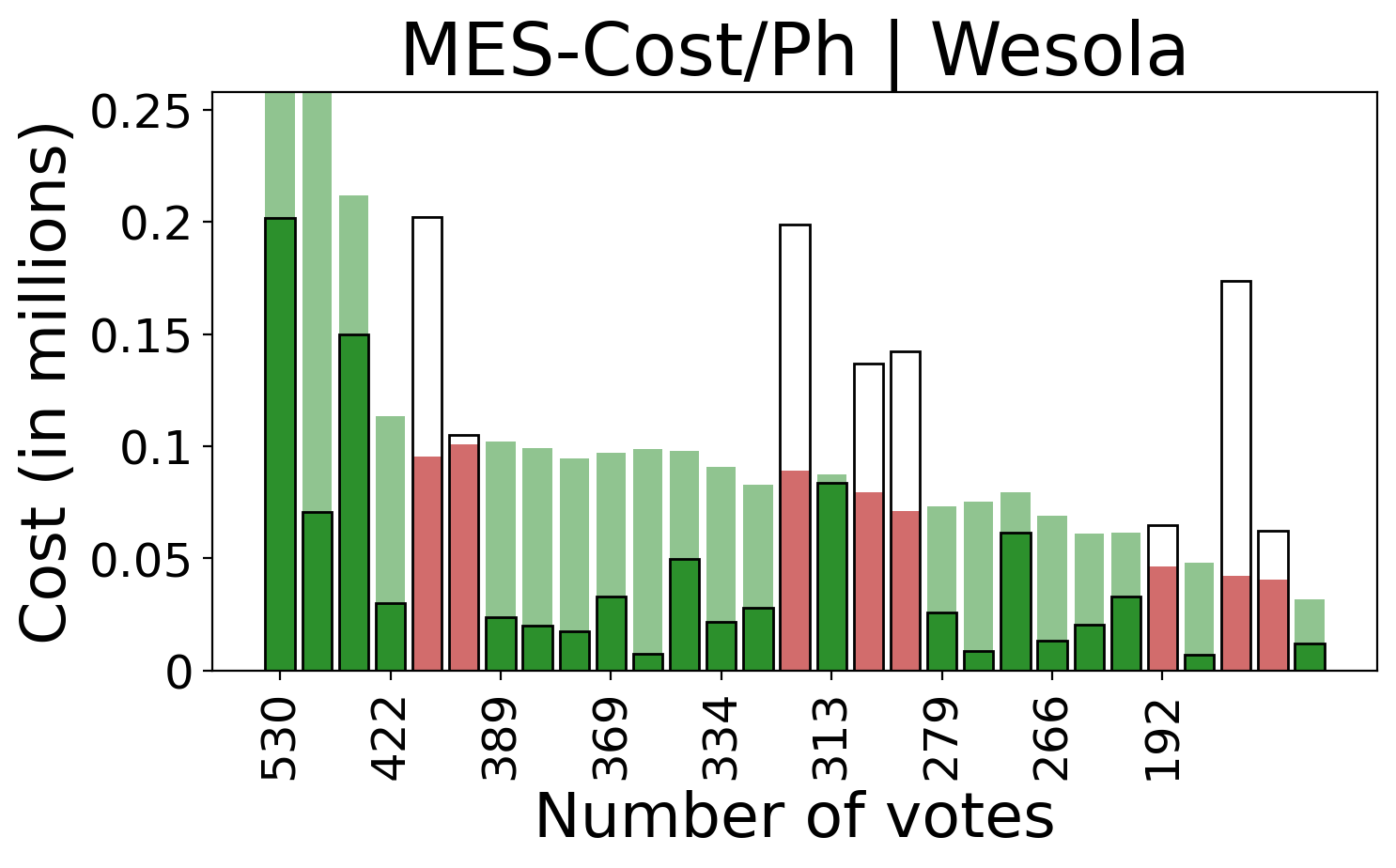}%
          \hspace{0mm}
    \includegraphics[width=\marginplotwidth]{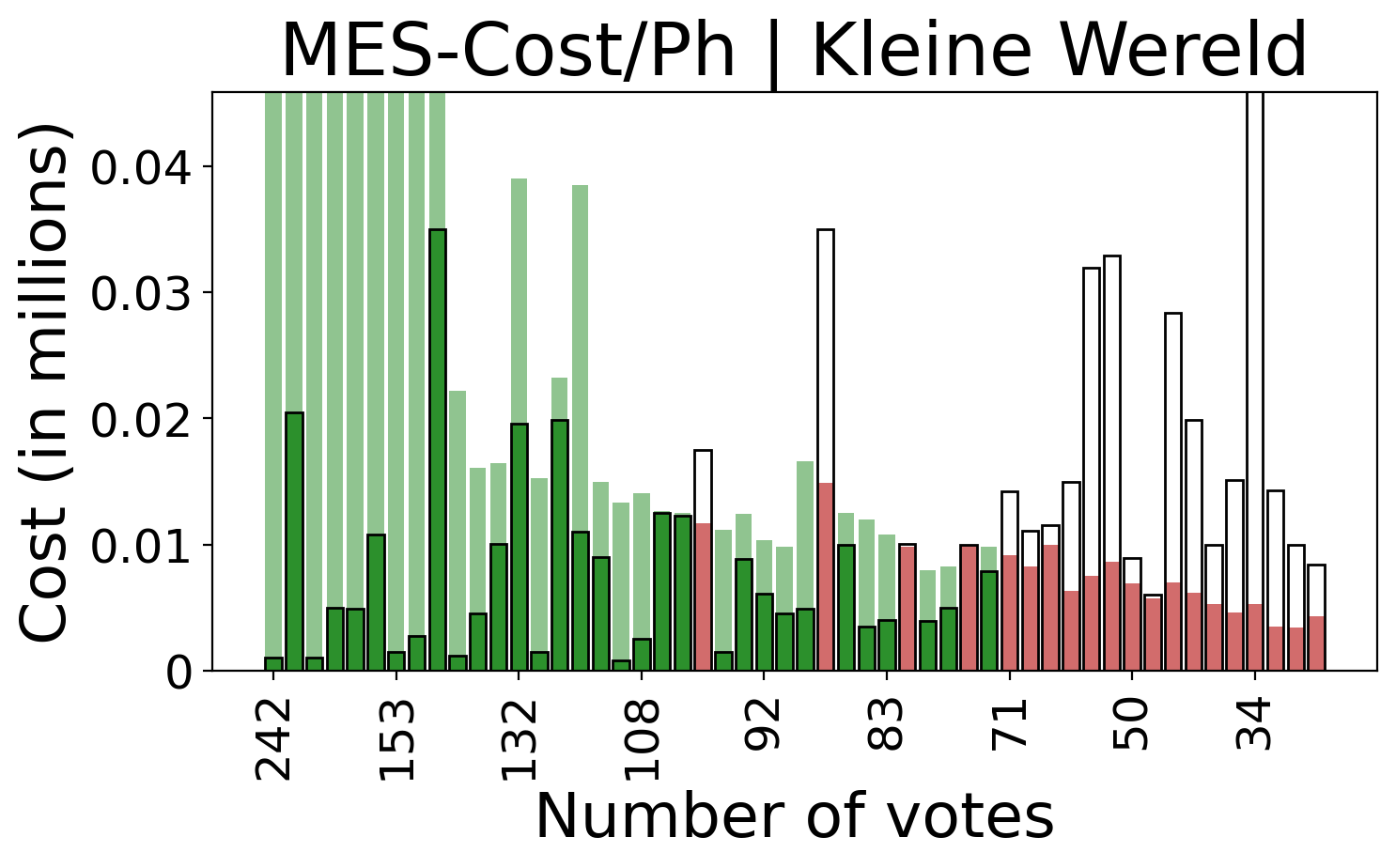}
    
    \caption{Winning and losing margins in real-life PB. Each bar
      shows a single project (the projects are ordered by their
      approval score, shown on the ${x}$~axis). Green and red ones
      give best responses (where green means a winning margin and red
      means a losing one).
      Black
	    rectangles show the original costs, with budget $1 011$k PLN for Wesola and 250k EUR for Kleine Wereld.}\label{fig:margins} 
\end{figure}

\subsection{Winning and Losing Margins}

Take a PB instance $E$ %
and a PB rule~$f$. Intuitively, for each project $p$ its \emph{best response}
$\response(p)$ is the highest cost such that if we use it instead of
$\cost(p)$, then~$p$ is (still) funded under~$f$. We define it as the supremum of the set of costs under which
$p$ is funded. Then, if~$p$ is winning in~$E$ under~$f$, then we let
its \emph{winning margin} be $\response(p) - \cost(p)$, and if it is
losing then we let its losing margin be $\cost(p) - \response(p)$.

In~\Cref{fig:margins} we show the winning and losing margins (and,
hence, the best responses) for all the projects from our two PB
instances.  For $\textbasicAV$, we observe very high winning margins
for the projects with the largest numbers of approvals, and close-to-zero
best responses for the remaining ones (this phenomenon is particularly
visible in Kleine Wereld).  In principle, under $\textbasicAV$ if a
project receives few votes even making it cheap will not make it
funded, unless it is so cheap that it is selected simply because
nothing else could fit within the budget left.  For $\textAVover$ the
best responses roughly correspond to the numbers of approvals each
project got; the some, roughly, holds for $\Phragmen$ and $\MESAP$.
The margins %
for $\MESP$ are significantly different. In
particular, for $\MESP$, several most-approved projects usually have
much higher winning margins than the rest, unlike for
$\textAVover$, $\Phragmen$ and $\MESAP$.

The main general observation is that unless a project received
relatively few approvals, then typically its proposers could have
reported a much higher cost. This suggests that they should choose
reasonable project costs that get their projects completed, rather
than minimize the costs as much as possible, with the hope that this
would improve their chances.
Yet, of course, there are projects where lowering their costs could
get them funded (see, e.g, \Cref{fig:margins} and the first project
for Wesola under $\textAVover$, $\Phragmen$, or $\MESAP$).

\subsection{Finding Equilibria Using Dynamics}
In the second experiment, our goal was to compute (approximate)
equilibria for our two instances under our rules. While for
$\textbasicAV$, $\textAVover$, and $\MESP$ we could compute the
equilibria using results from the previous section,\footnote{For
  $\MESP$ this follows from the fact that, for our instances, the
  equilibrium for $\MesCost$ is exhaustive.} this would not be
possible for $\Phragmen$ and $\MESAP$. Instead,  we simulate
certain dynamics (for the rules where we can compute equilibria
directly, our dynamics give nearly identical results, so we expect
that the results are also meaningful for the other rules, albeit with
no guarantees\footnote{For $\Phragmen$ and $\MESAP$ there may be no
  NE. We do not have an immediate way of checking this, but the
  results suggest that the profiles that we find are, at least, close
  to being equilibria.}).

Given a PB instance, our dynamics go as follows.  First, each proposer
reports the same cost as was originally chosen for his or her
project. %
Then, in each iteration, one of them, selected uniformly at random,
either slightly increases or decreases his or her project's
cost. Specifically, the proposer chooses a number $x$ between $0$ and
$\cost/10$ uniformly at random (where $\cost$ is the current project's
cost) and if their project was losing in the previous iteration, then
the proposer decreases its cost by $x$, and if it was winning, then he
or she increases its cost by $x$ (but if this action would have caused
the project to lose, then the proposer does not change the project's
cost).
We expect to converge to an NE, if one exists, after
running sufficiently many iterations.

In~\Cref{fig:games} we present the results of the dynamics, that is, the
obtained strategy profiles, after $10 \ 000$ iterations. Under
$\textbasicAV$, as expected, all the budget goes to the project
with the most votes. Under $\textAVover$, every project ends up with a
cost proportional to its support. Like in the previous
experiment, $\Phragmen$ and $\MESAP$ produce similar results,
while those of $\MESP$ remain different. For $\textbasicAV$,
$\textAVover$, and $\MESP$ most of the projects ``reached'' the costs
predicted by the equilibrium.\footnote{%
	In the equilibrium,
  all
	projects should be winning,
	so we should observe no red bars. Their existence
	means that the equilibrium was not reached within 10 000
  iterations.}

The main conclusion from this experiment (supported also by
simulations for other PB instances; see \Cref{app:experiments}) is
that under $\textAVover$, $\Phragmen$, and $\MESAP$ the proposers are
incentivized to use costs that mostly reflect the number of approvals
they receive. Under $\MESP$, most strongly supported projects, as well
as those that are supported by many voters who do not approve other,
more popular projects, can request notably higher costs. Consequently,
in an equilibrium the $\MESP$ rule funds fewer, more expensive projects than
$\textAVover$, $\Phragmen$, and $\MESAP$. We view this observation as
one of the more important take-home messages from our work, which can
be used to favor $\MESP$ in practice (while funding some cheap
projects is important, funding only such projects is logistically
problematic for cities, as it requires significant effort in
coordinating their implementation).

\begin{figure}[t!]
     \centering
     
     \includegraphics[width=\gameplotwidth]{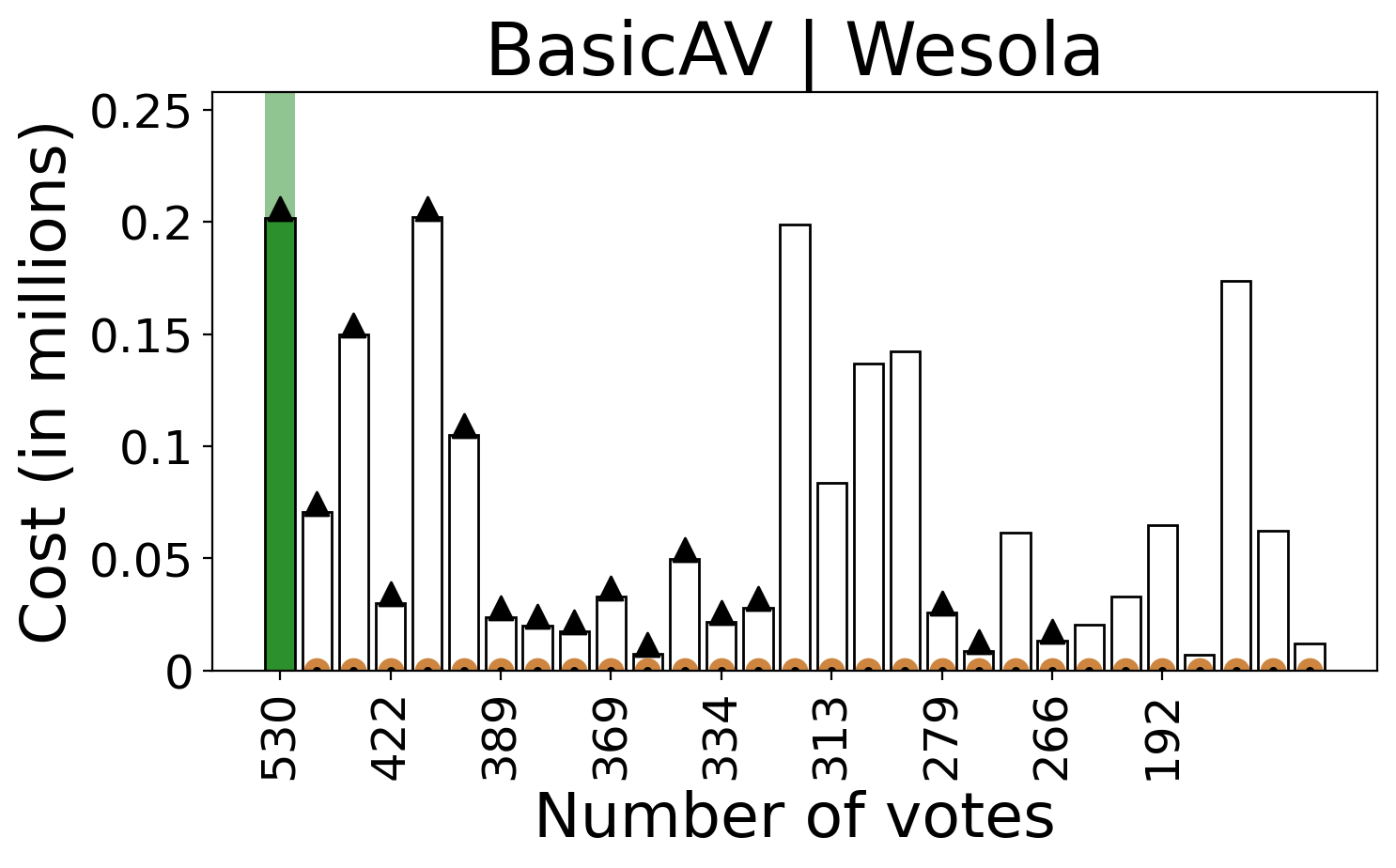}%
          \hspace{0mm}
     \includegraphics[width=\gameplotwidth]{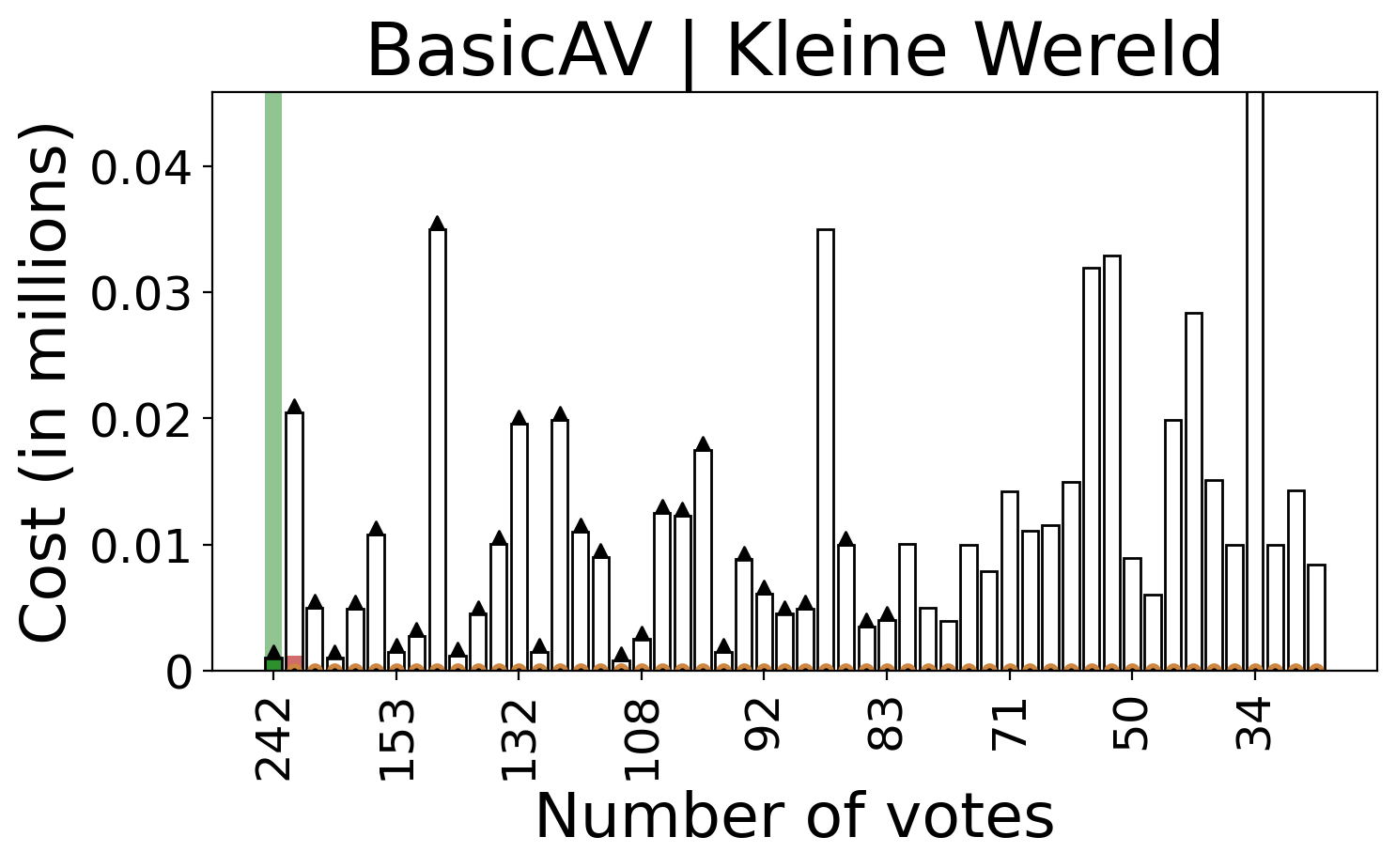}%
         
    \includegraphics[width=\gameplotwidth]{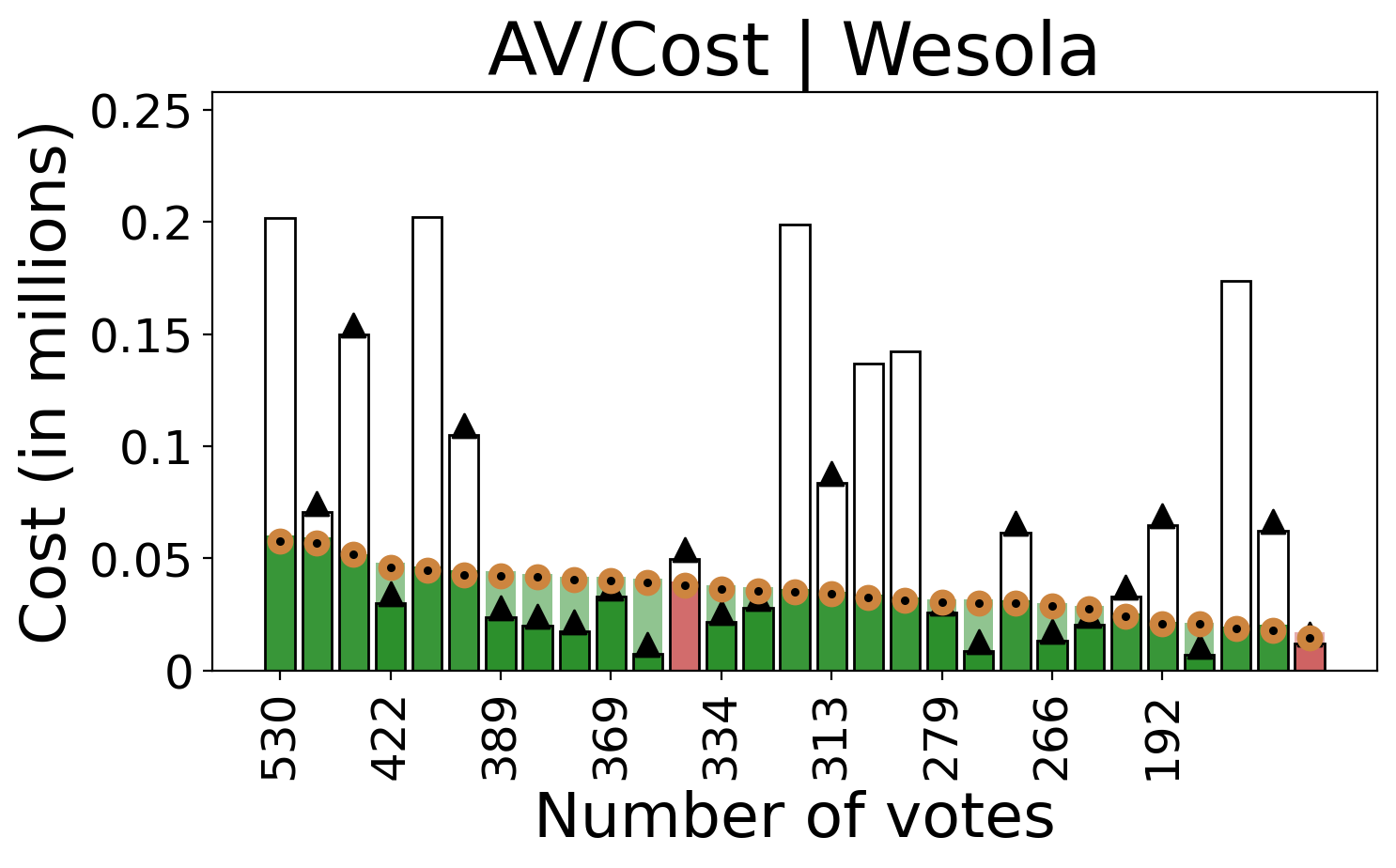}%
              \hspace{0mm}
     \includegraphics[width=\gameplotwidth]{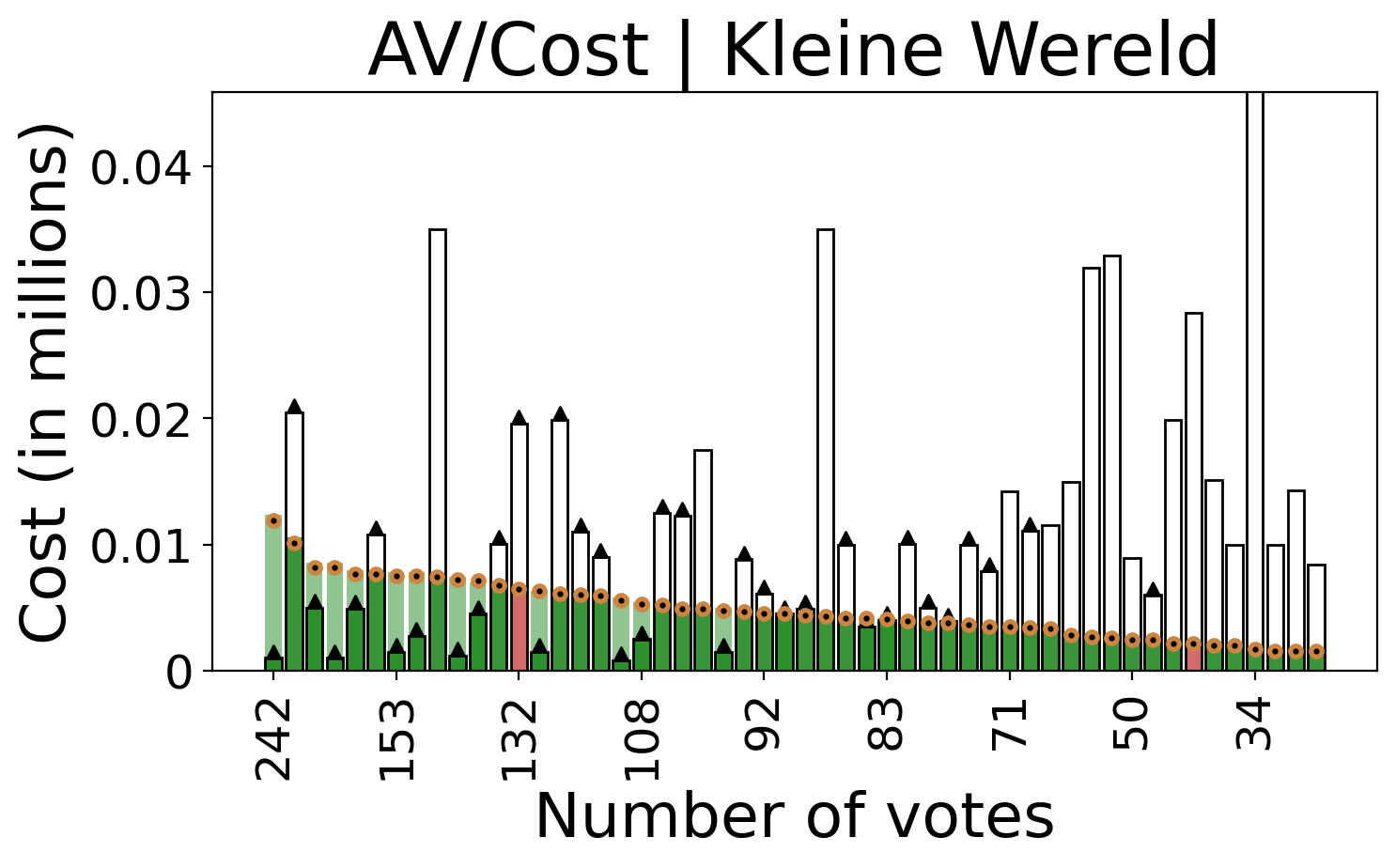}%
     
     \includegraphics[width=\gameplotwidth]{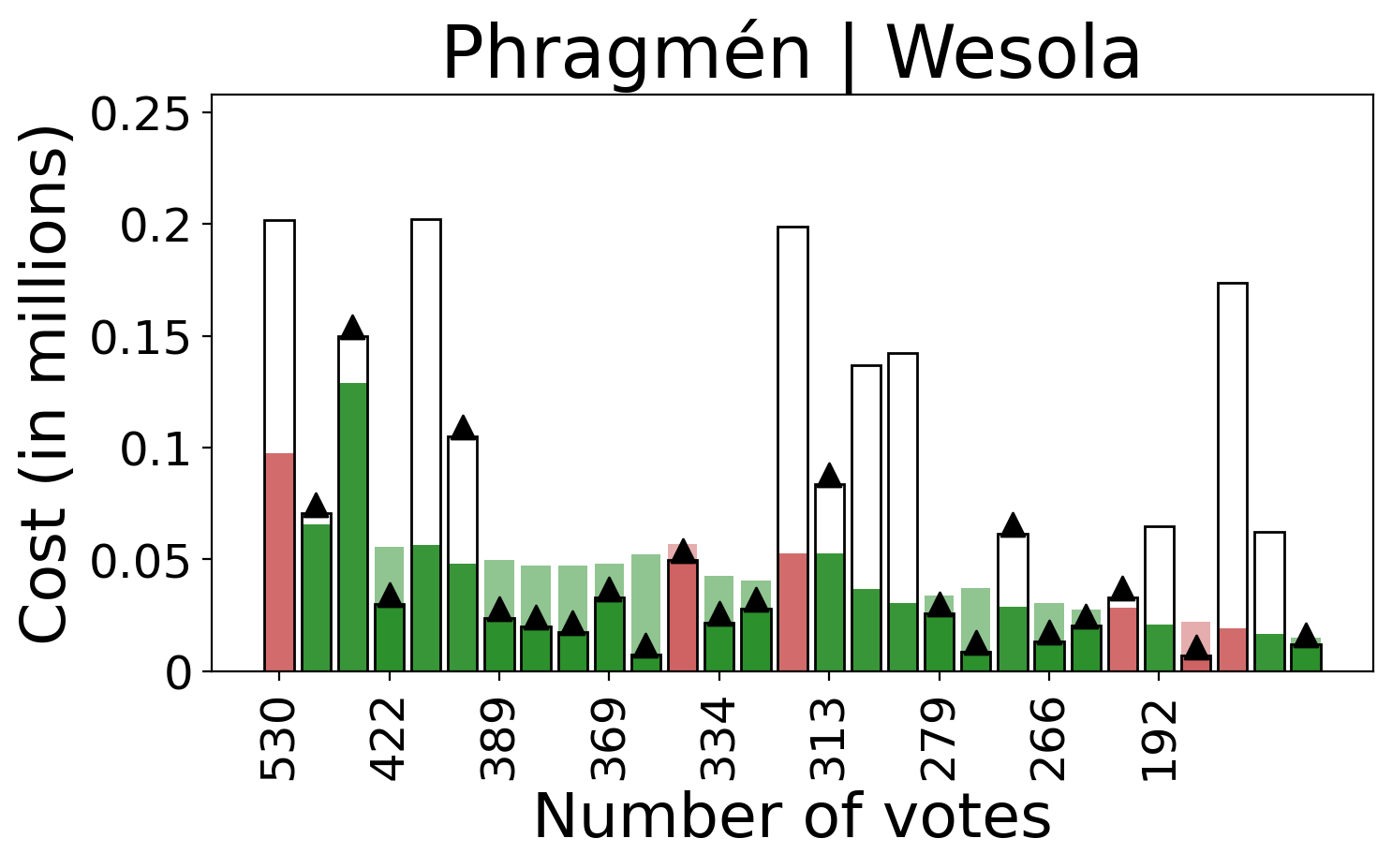}%
              \hspace{0mm}
     \includegraphics[width=\gameplotwidth]{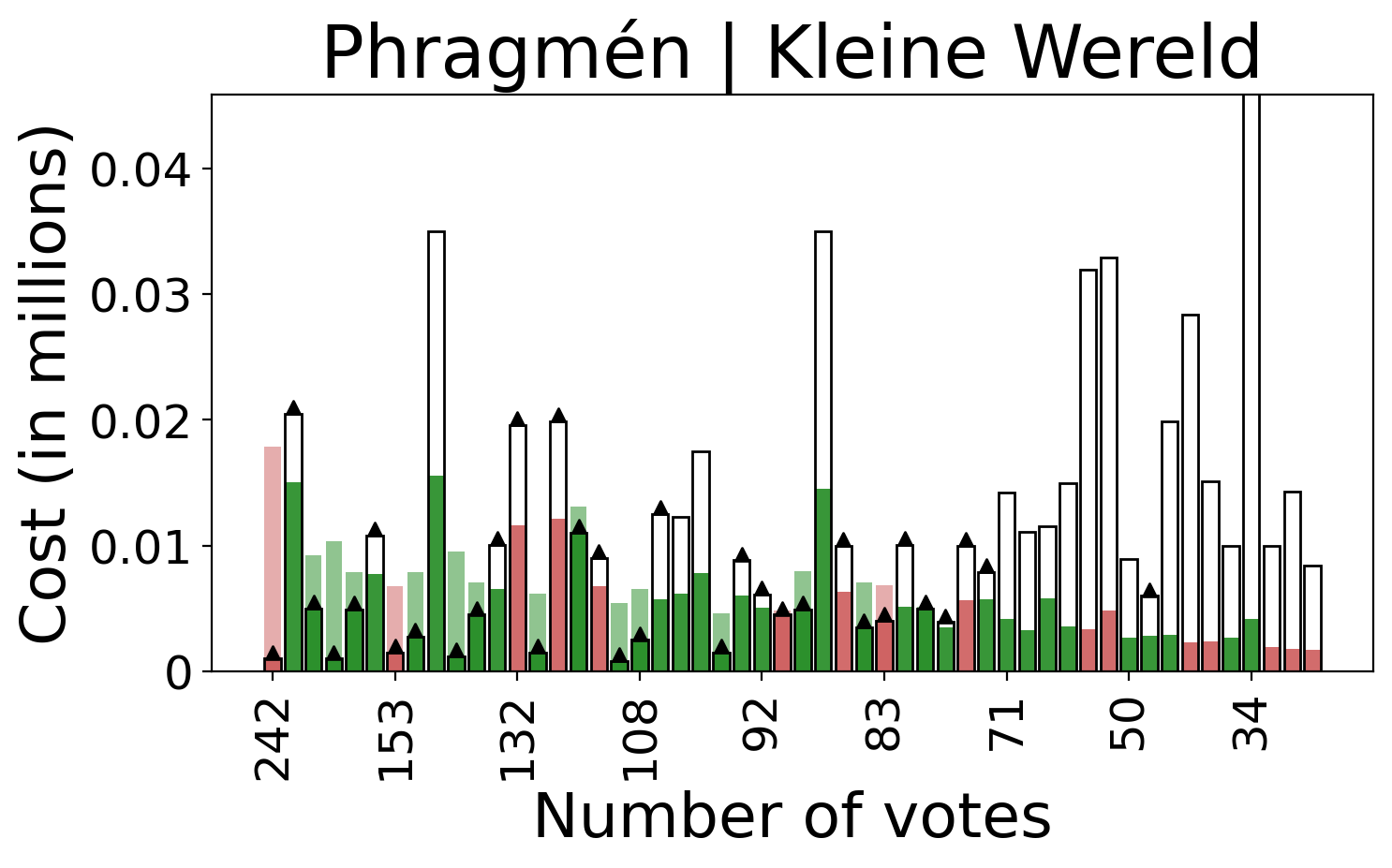}
    
    \includegraphics[width=\gameplotwidth]{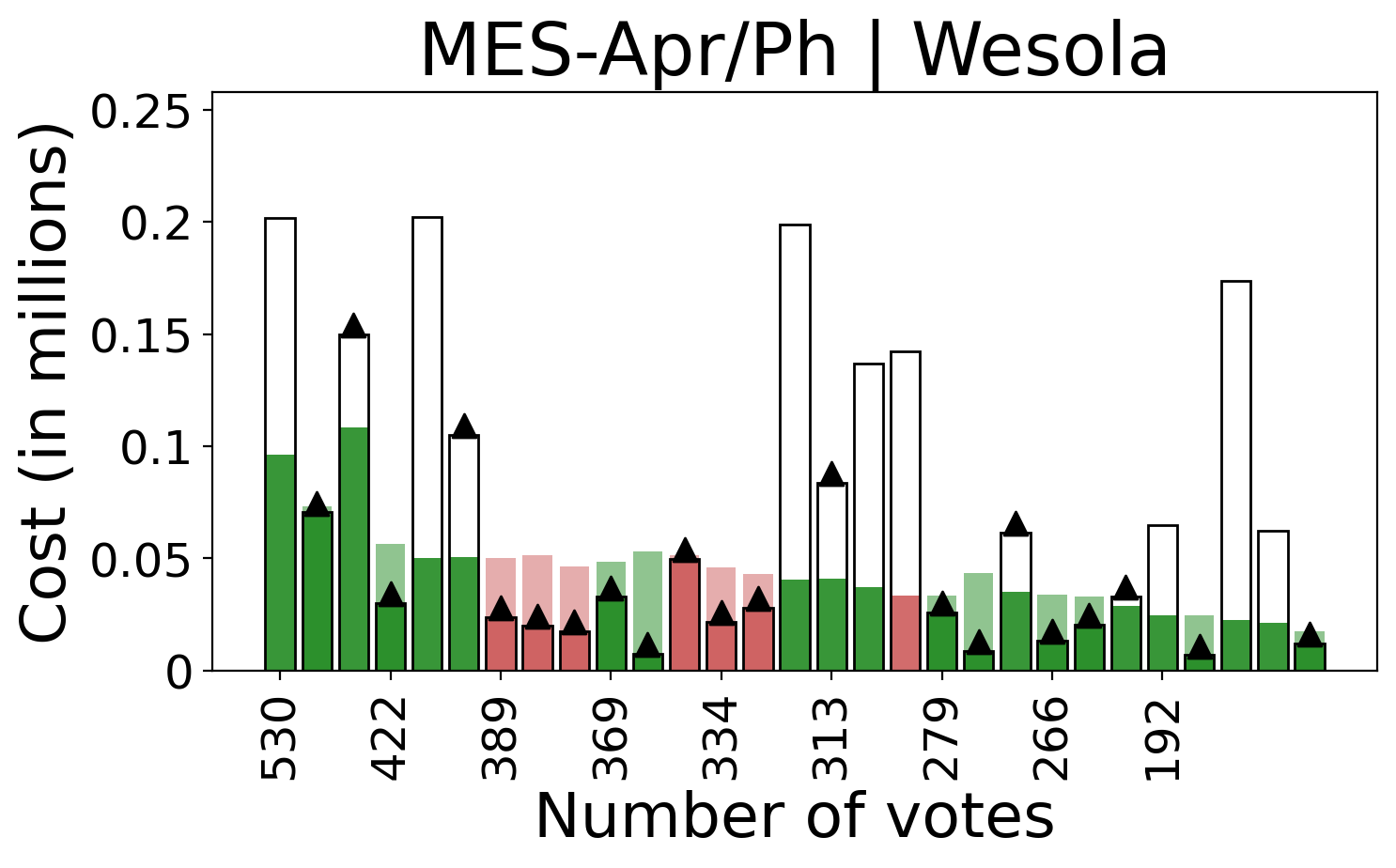}%
              \hspace{0mm}
    \includegraphics[width=\gameplotwidth]{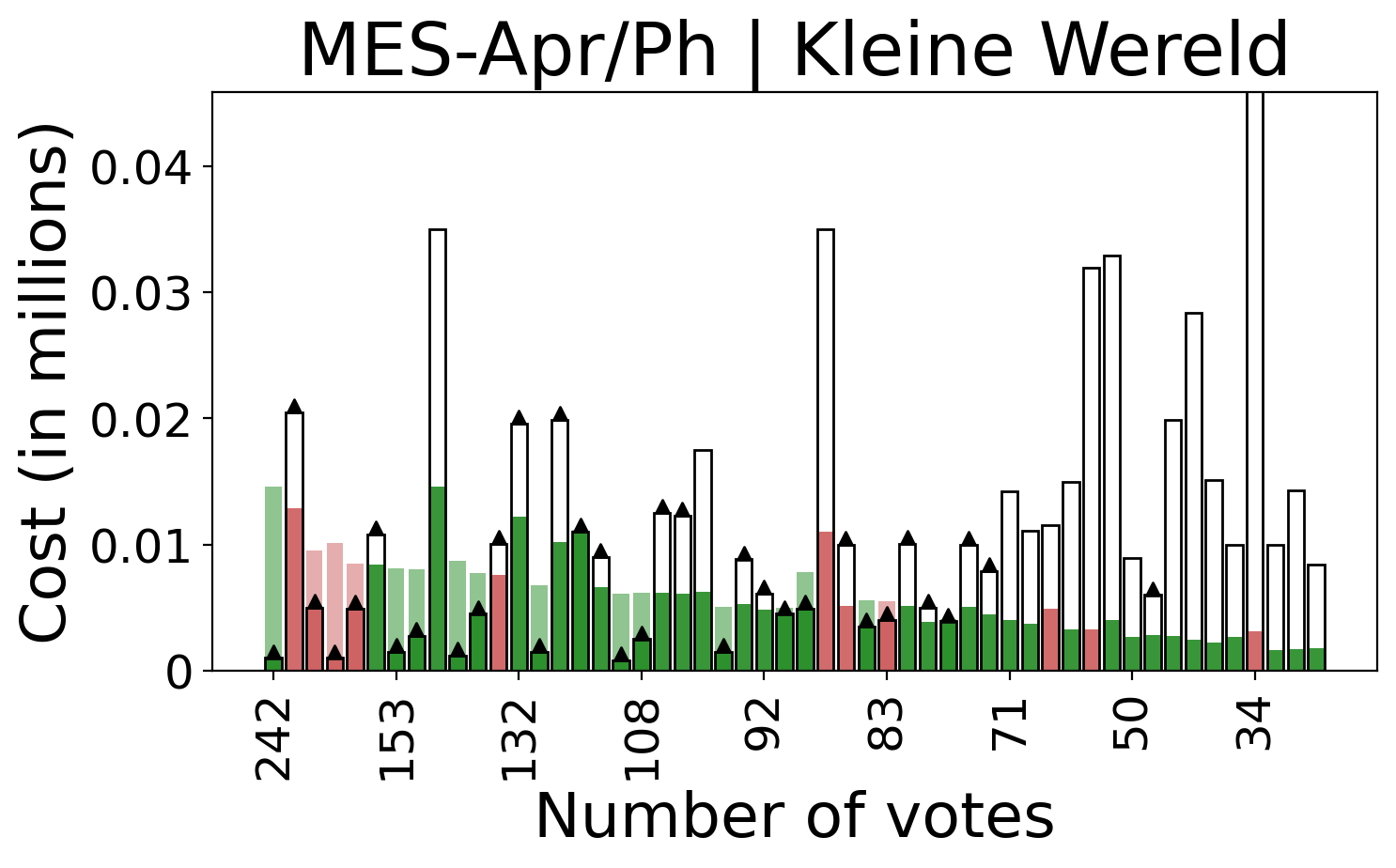}
    
    \includegraphics[width=\gameplotwidth]{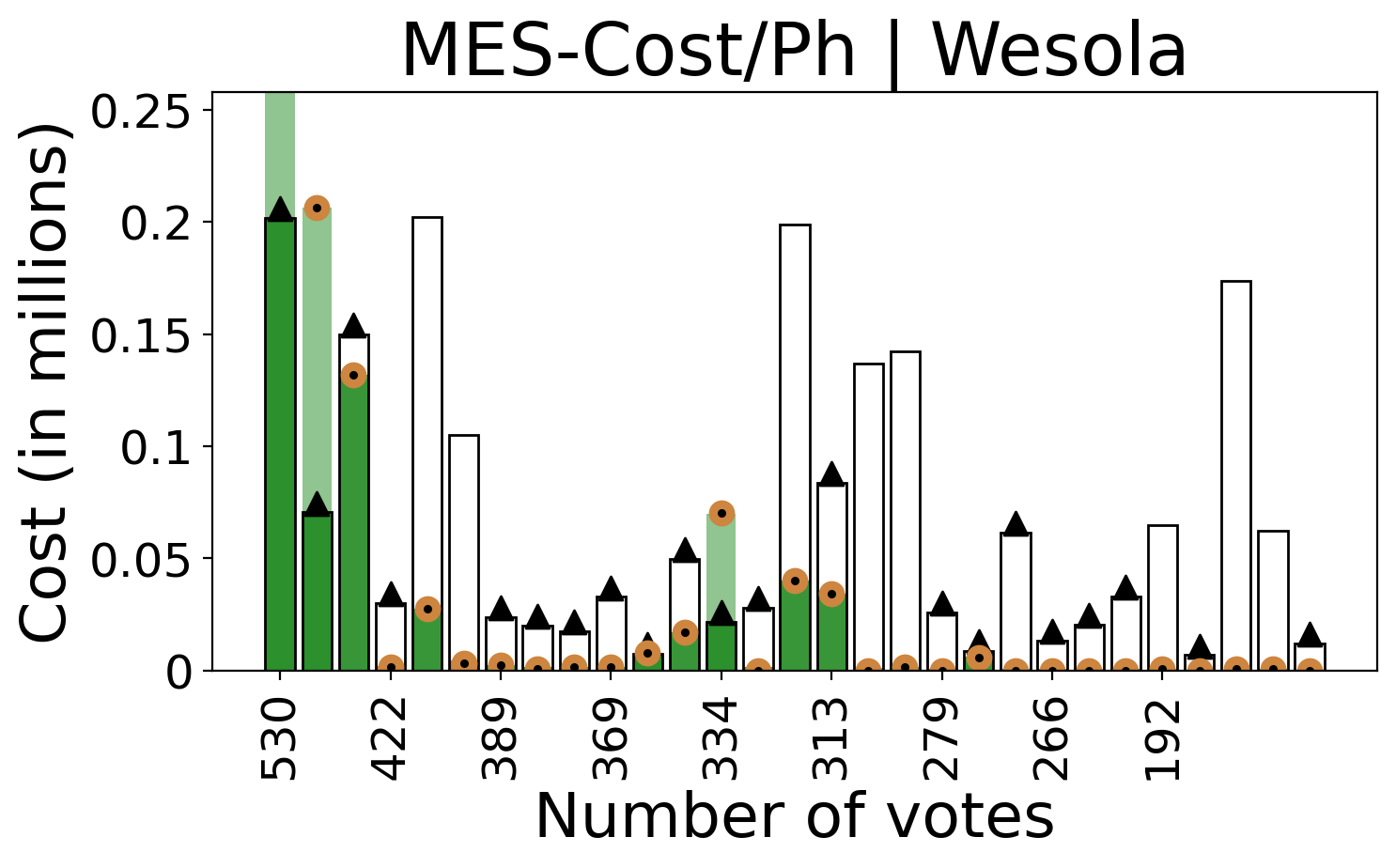}%
              \hspace{0mm}
    \includegraphics[width=\gameplotwidth]{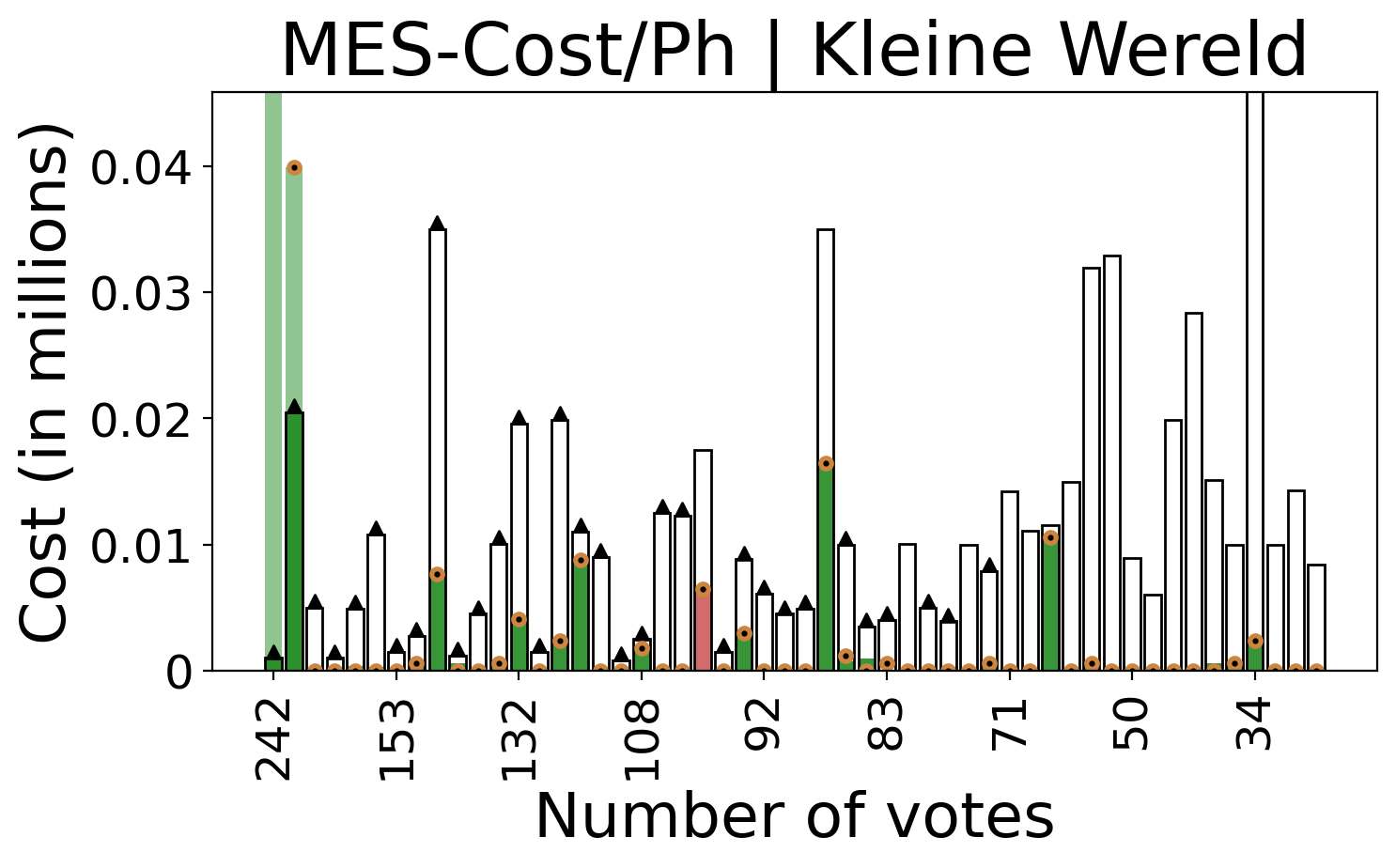}%
   
    \caption{ Strategy profiles after
      ${10 \ 000}$ iterations of our dynamics. Bars
      represent the projects (%
      in the order of their approval score,
			depicted
			on the ${x}$
      axis). The green bars show final costs of the winning projects
			(the brighter part emphasizes the
      increase, as compared to the original cost). The red bars show
      final costs of losing projects.
			Black outlines denote the original costs of the
      projects. Black triangles mark the originally winning projects. %
      Brown circles denote the equilibrium costs (for
      rules where we can compute it).}\label{fig:games}
\end{figure}

\section{Conclusion}\label{sed:conc}

We have introduced a game-theoretic framework capturing strategic cost
selection of projects in participatory budgeting scenarios. Focusing
on $\textbasicAV$, $\textAVover$, $\Phragmen$, and $\Mes$ rules, we
explored the conditions under which NE exist in such
games. Then, using real-life election data, we explored the effects of
iterative, profitable changes in projects' costs.
We believe to have set the groundwork for further theoretical and experimental
study based on our framework. We envision a number of directions for future
work. One is a natural computational question of establishing the complexity
of deciding if equilibria exist in our games. Further, studying a more
involved class of games emerging from relaxing our basic assumptions seems a
very interesting and promising research avenue. In particular, it is
compelling to verify whether the increased model complexity still leads to the
observed phenomena, or rather introduces new dynamics.

\section*{Acknowledgments}
This project has received funding from the European Research Council (ERC) under
the European Union’s Horizon 2020 research and innovation programme (grant
agreement No 101002854). Piotr Skowron was supported by the European Union (ERC,
PRO-DEMOCRATIC, 101076570).
\begin{center}
  \includegraphics[width=2cm]{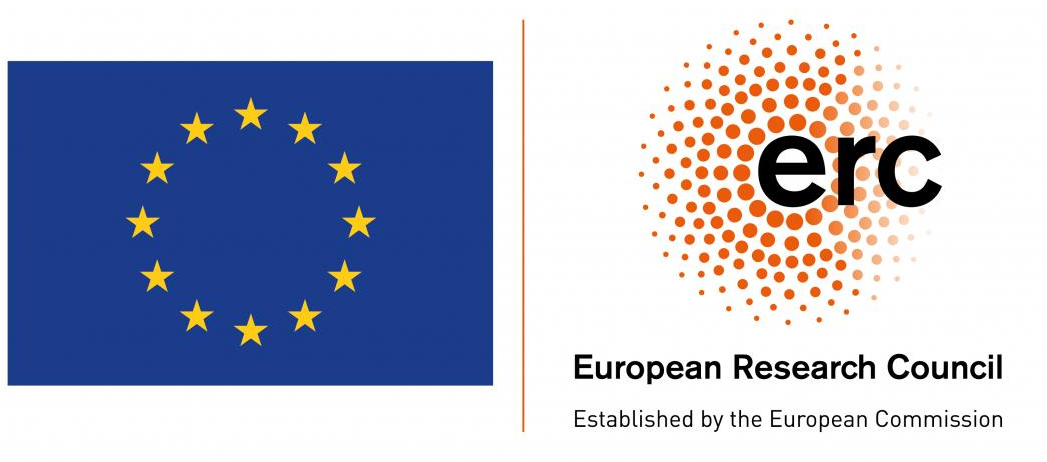}
\end{center}
\bibliographystyle{abbrvnat}
\bibliography{pb-cost}

\clearpage
\appendix

\section{Missing Proofs}
\appendixProofs

\section{Additional Experimental Results}\label{app:experiments}

Here, we present experimental results for four additional PB instances held in 2022 in different districts of Warsaw, Poland, i.e., in Bemowo, Bielany, Wilanow, and Wlochy. Details regarding these datasets are presented in~\Cref{tab:pb_details_apdx}.
    In~\Cref{fig:margins_apdx} we present the winning and losing margins in additional real-life PB instances, and in~\Cref{fig:games_apdx} we present the strategy profiles after
      $10'000$ iterations of our dynamics.

Moreover, in~\Cref{tab:margins_values} we present comparison of the original winning and losing margins and those after 10'000 iterations. We observe that in most cases the size of the margins drastically decreased.

\begin{table}[t]
  \caption{Additional PB instances that we analyze in the experiments. Vote
    Len. means the average number of approvals in a ballot. Rule
    indicates the PB rule that was used in the original election.}
     \label{tab:pb_details_apdx}%
    \centering%
  \begin{tabular}{c | c c c c c }
          Instance & $|P|$ & $|V|$ & Budget & \hspace{-0.14cm} Vote Len. \hspace{-0.1cm} & Rule\\
          \midrule
         Bemowo & 83 & 5181 & 4854279 & 10.8 &\textbasicAV \\
         Bielany & 98 & 4957 & 5258802 & 9.8 &\textbasicAV \\
         Wilanow & 35 & 2359 & 1516962 & 9.59 &\textbasicAV \\
         Wlochy & 43 & 2221 & 1719224 & 9.51 &\textbasicAV        \\  
          \midrule
    \end{tabular}
 \end{table}

\begin{table*}[t]
    \centering
    \begingroup
		\renewcommand*{\arraystretch}{1.3}
    \begin{tabular}{c | c |cc|cc}
    \toprule
          & \multirow{2}{*}{Rule} & \multicolumn{2}{c}{Original Margins}  & \multicolumn{2}{c}{10000 it. Margins} \\
          &  & Winning & Losing & Winning & Losing \\
          \midrule
\multirow{5}{*}{\rotatebox{90}{Bemowo}} & \textbasicAV & $1830 \pm 1283$ & $122 \pm 182$ & 0 & $24 \pm 37$ \\
& \textAVover & $150 \pm 104$ & $225 \pm 257$ & $1 \pm 3$ & $0.2 \pm 0.1$ \\
& \Phragmen & $141 \pm 91$ & $220 \pm 261$ & $3 \pm 3$ & $0.9 \pm 0.7$ \\
& \MESAP & $150 \pm 95$ & $223 \pm 264$ & $3 \pm 4$ & $2 \pm 2$ \\
& \MESP & $250 \pm 323$ & $205 \pm 244$ & $1 \pm 2$ & $1 \pm 2$ \\
 \midrule
\multirow{5}{*}{\rotatebox{90}{Bielany}} & \textbasicAV & $1447 \pm 1355$ & $218 \pm 205$ & 0 & $2 \pm 2$ \\
& \textAVover & $157 \pm 128$ & $171 \pm 150$ & $0.9 \pm 2$ & $0.2 \pm 0.3$ \\
& \Phragmen & $146 \pm 121$ & $176 \pm 154$ & $4 \pm 3$ & $0.9 \pm 0.9$ \\
& \MESAP & $148 \pm 121$ & $176 \pm 151$ & $3 \pm 4$ & $1 \pm 2$ \\
& \MESP & $172 \pm 193$ & $180 \pm 154$ & $3 \pm 6$ & $3 \pm 4$ \\
\midrule
\multirow{5}{*}{\rotatebox{90}{Wesola}} & \textbasicAV & $265 \pm 249$ & $64 \pm 51$ & 0 & 0 \\
& \textAVover & $86 \pm 27$ & $55 \pm 31$ & $0.3 \pm 0.4$ & $0.2 \pm 0.1$ \\
& \Phragmen & $69 \pm 31$ & $51 \pm 37$ & $2 \pm 1$ & $0.9 \pm 1$ \\
& \MESAP & $79 \pm 33$ & $43 \pm 37$ & $2 \pm 1$ & $0.7 \pm 0.7$ \\
& \MESP & $70 \pm 56$ & $65 \pm 45$ & $0.1 \pm 0.1$ & $0.4 \pm 0.5$ \\
\midrule
\multirow{5}{*}{\rotatebox{90}{Wilanow}} & \textbasicAV & $536 \pm 399$ & $87 \pm 80$ & 0 & 0 \\
& \textAVover & $87 \pm 59$ & $51 \pm 34$ & $2 \pm 2$ & $0.1 \pm 0.0$ \\
& \Phragmen & $77 \pm 52$ & $66 \pm 41$ & $2 \pm 2$ & $0.3 \pm 0.2$ \\
& \MESAP & $89 \pm 56$ & $58 \pm 37$ & $2 \pm 1$ & $0.7 \pm 0.8$ \\
& \MESP & $79 \pm 120$ & $95 \pm 64$ & $1 \pm 2$ & $2 \pm 4$ \\
\midrule
\multirow{5}{*}{\rotatebox{90}{Wlochy}} & \textbasicAV & $532 \pm 524$ & $46 \pm 35$ & 0 & 0 \\
& \textAVover & $107 \pm 49$ & $56 \pm 48$ & $0.8 \pm 1.0$ & $0.1 \pm 0.1$ \\
& \Phragmen & $83 \pm 52$ & $62 \pm 42$ & $3 \pm 3$ & $2 \pm 1$ \\
& \MESAP & $85 \pm 53$ & $64 \pm 45$ & $2 \pm 2$ & $1 \pm 1$ \\
& \MESP & $138 \pm 186$ & $56 \pm 46$ & $0.2 \pm 0.2$ & $0.4 \pm 0.5$ \\
\midrule
\multirow{5}{*}{\rotatebox{90}{Kleine Wereld}} & \textbasicAV & $117 \pm 78$ & $14 \pm 11$ & 0 & $1.0 \pm 0.0$ \\
& \textAVover & $12 \pm 8$ & $12 \pm 12$ & $0.1 \pm 0.1$ & $0.1 \pm 0.0$ \\
& \Phragmen & $10 \pm 8$ & $10 \pm 11$ & $0.5 \pm 0.6$ & $0.2 \pm 0.2$ \\
& \MESAP & $11 \pm 7$ & $11 \pm 11$ & $0.4 \pm 0.6$ & $0.2 \pm 0.2$ \\
& \MESP & $26 \pm 34$ & $11 \pm 11$ & $2 \pm 1$ & $0.1 \pm 0.0$ \\
\bottomrule
    \end{tabular}
		\caption{Comparison of winning and losing margins in the original elections
		and after 10000 iterations of the game. (All the values are in 1000
		PLN\textbackslash{}EUR). In each entry the first value denotes the average and the second
value (i.e.,\,the one after $\pm$ sign) denotes the standard
deviation.}\label{tab:margins_values}
\endgroup
\end{table*}

\newgeometry{left=2cm,right=2cm}
\newcommand{\hgap}{2pt}
\begin{figure*}[t]\centering%
     \includegraphics[width=\marginplotwidth]{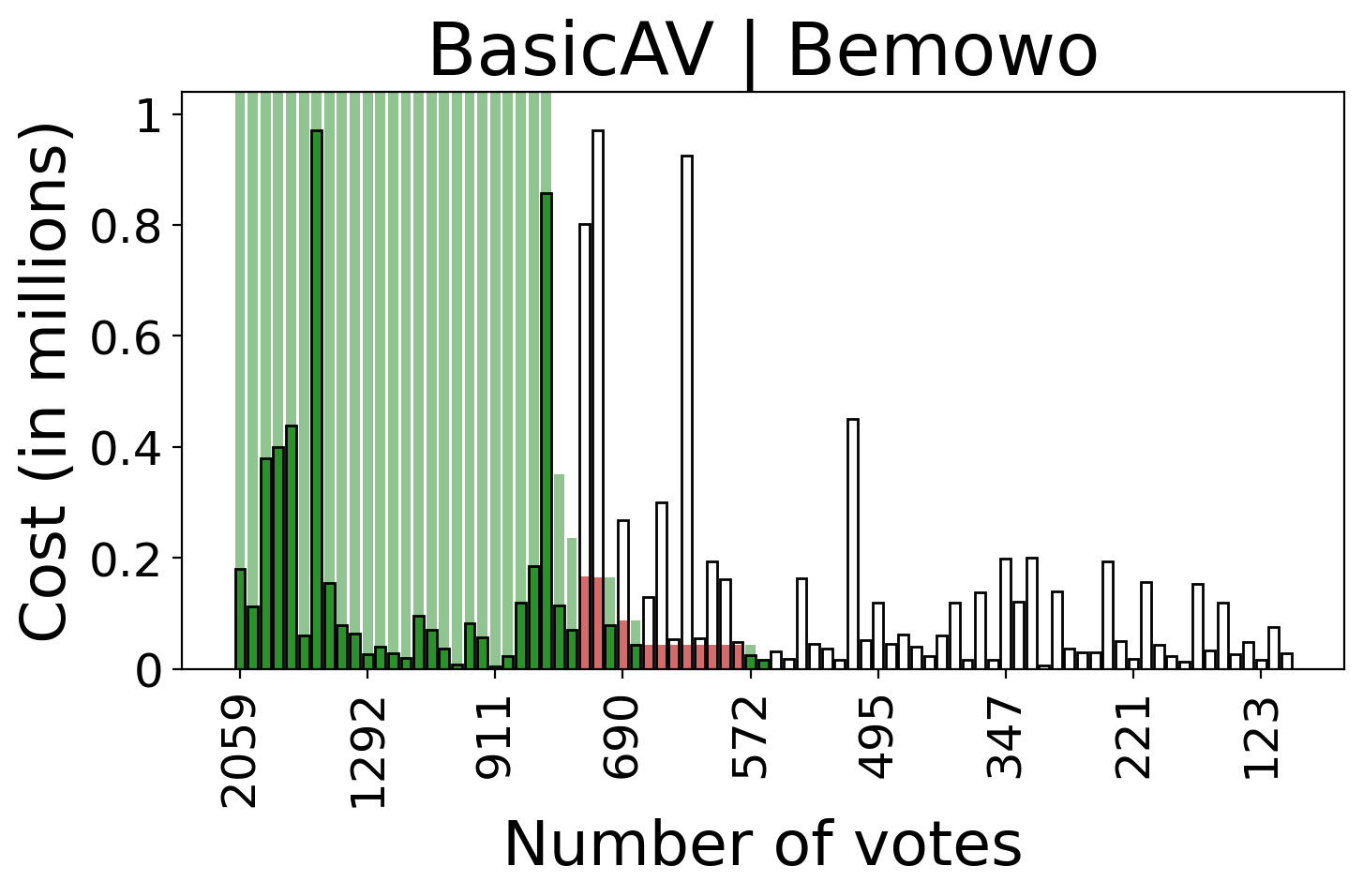}%
         \hspace{\hgap}%
     \includegraphics[width=\marginplotwidth]{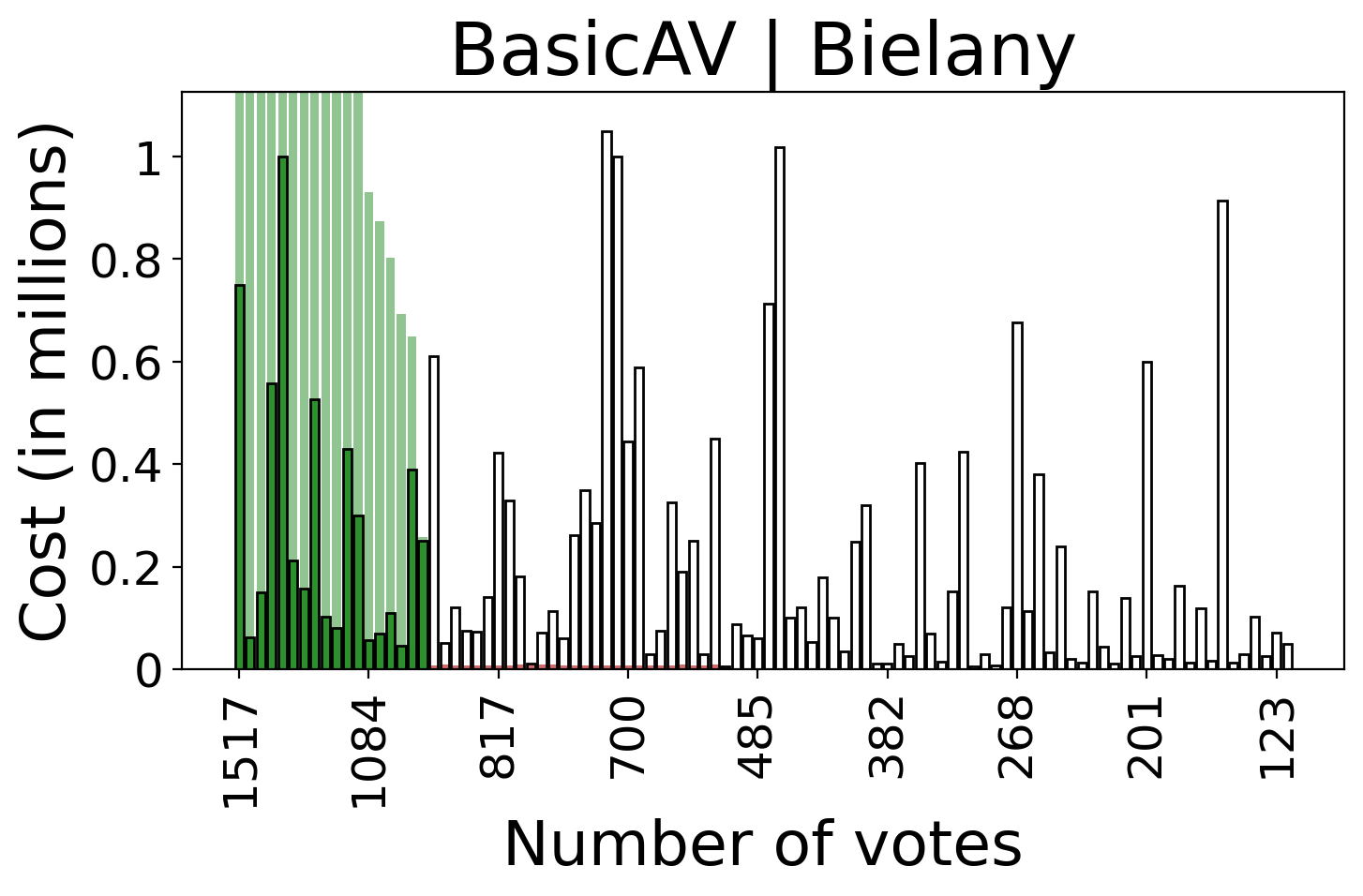}%
         \hspace{\hgap}%
     \includegraphics[width=\marginplotwidth]{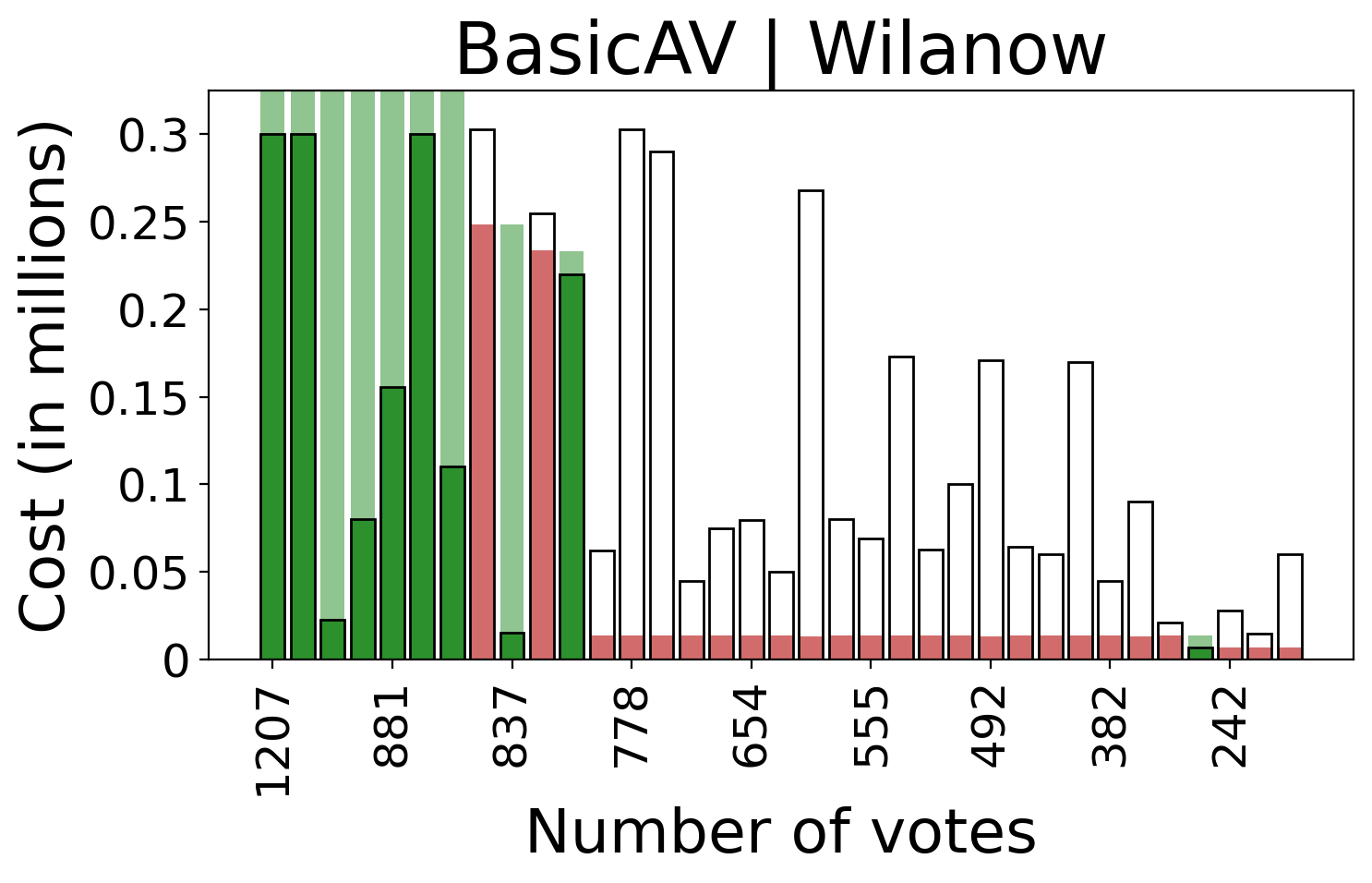}%
         \hspace{\hgap}%
     \includegraphics[width=\marginplotwidth]{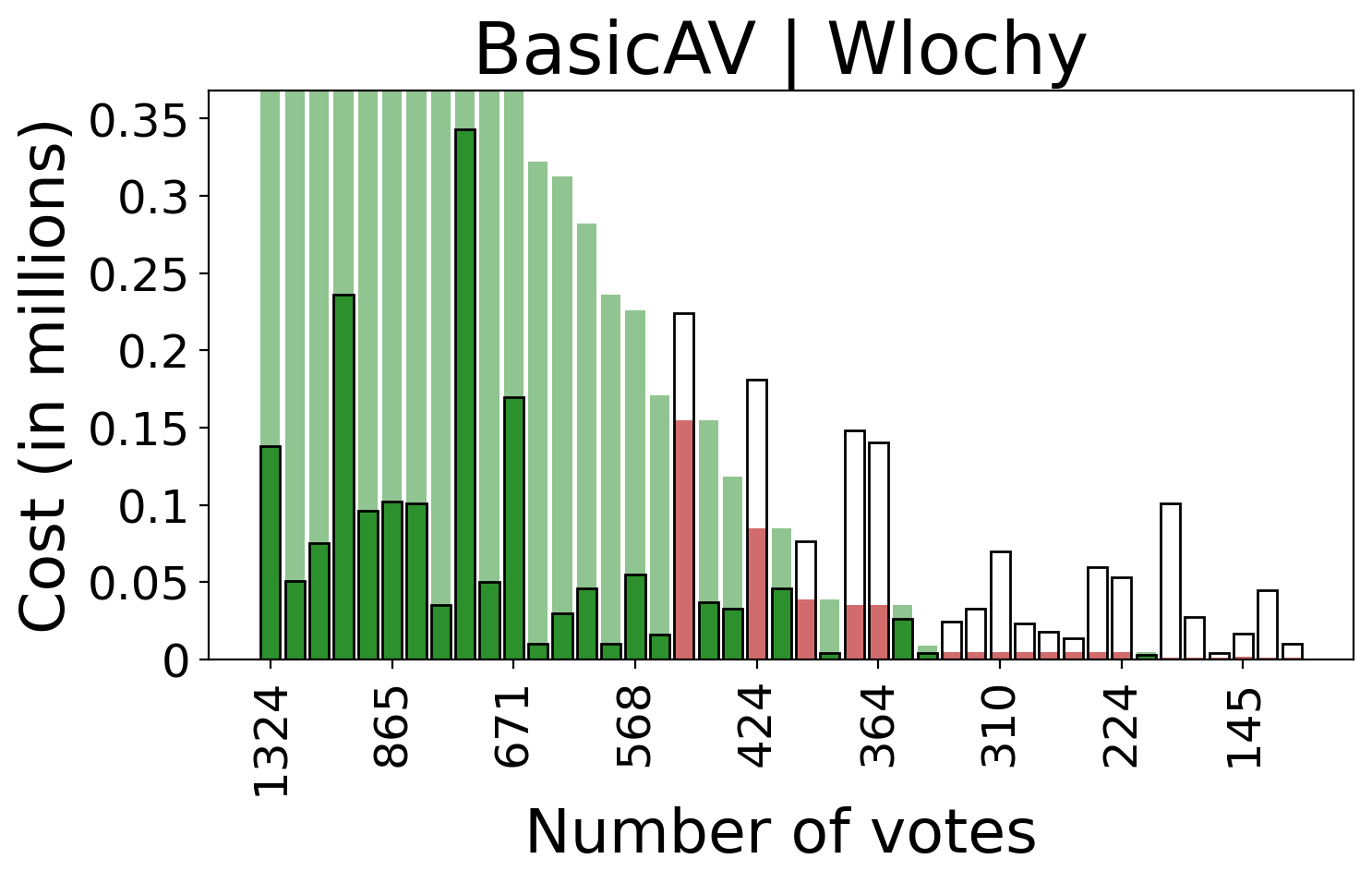}%

     \includegraphics[width=\marginplotwidth]{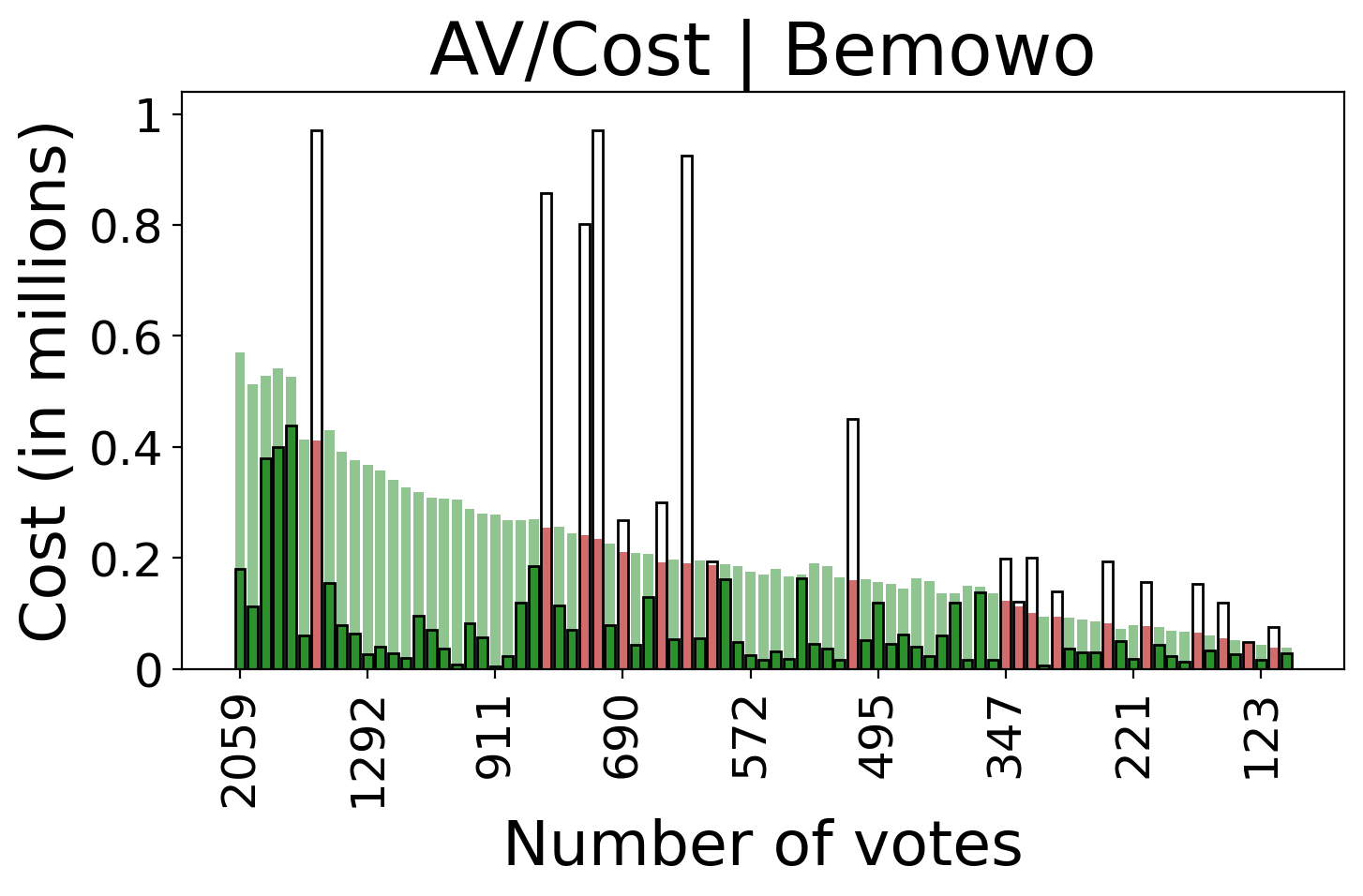}%
         \hspace{\hgap}%
     \includegraphics[width=\marginplotwidth]{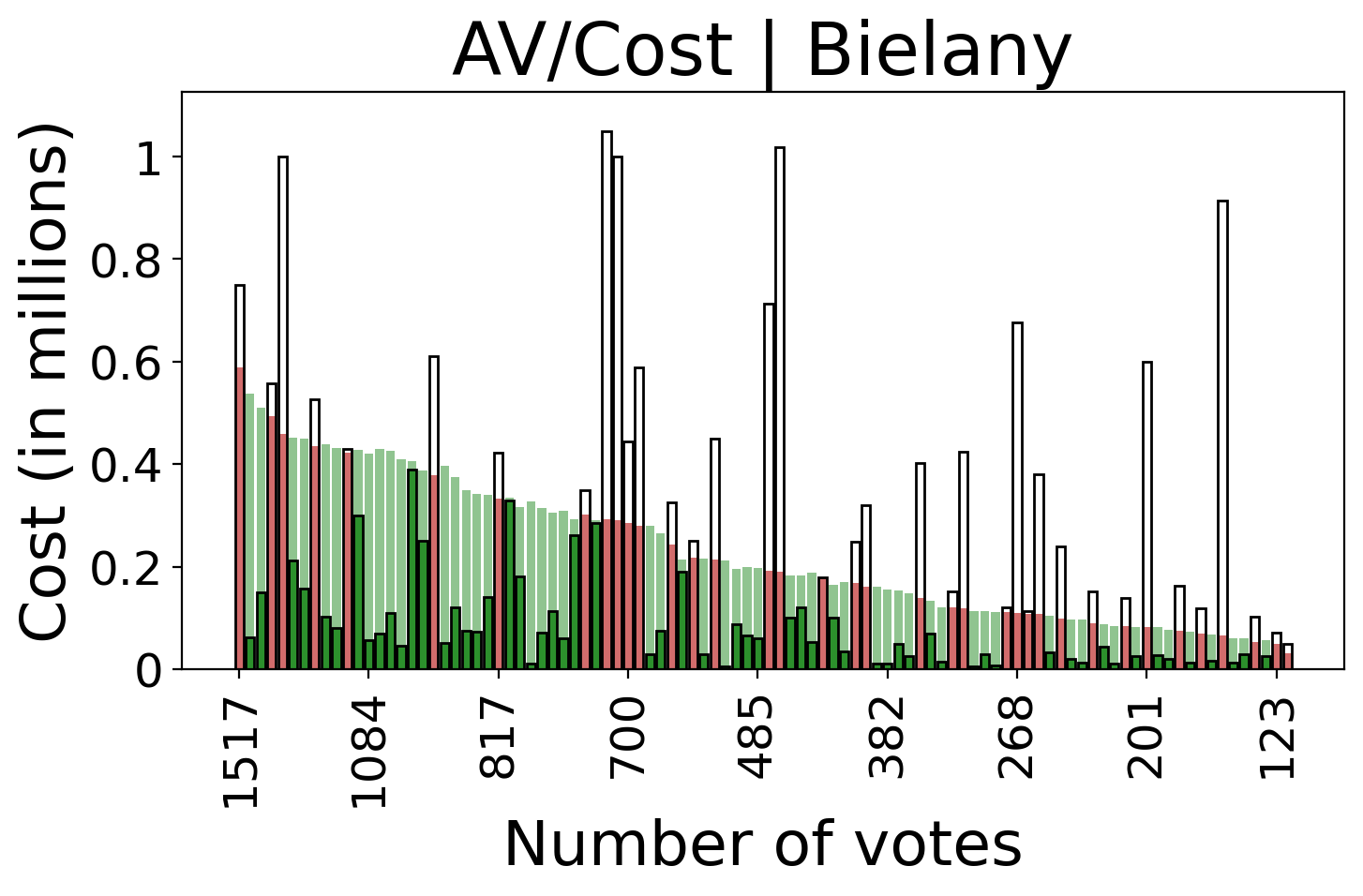}%
         \hspace{\hgap}%
     \includegraphics[width=\marginplotwidth]{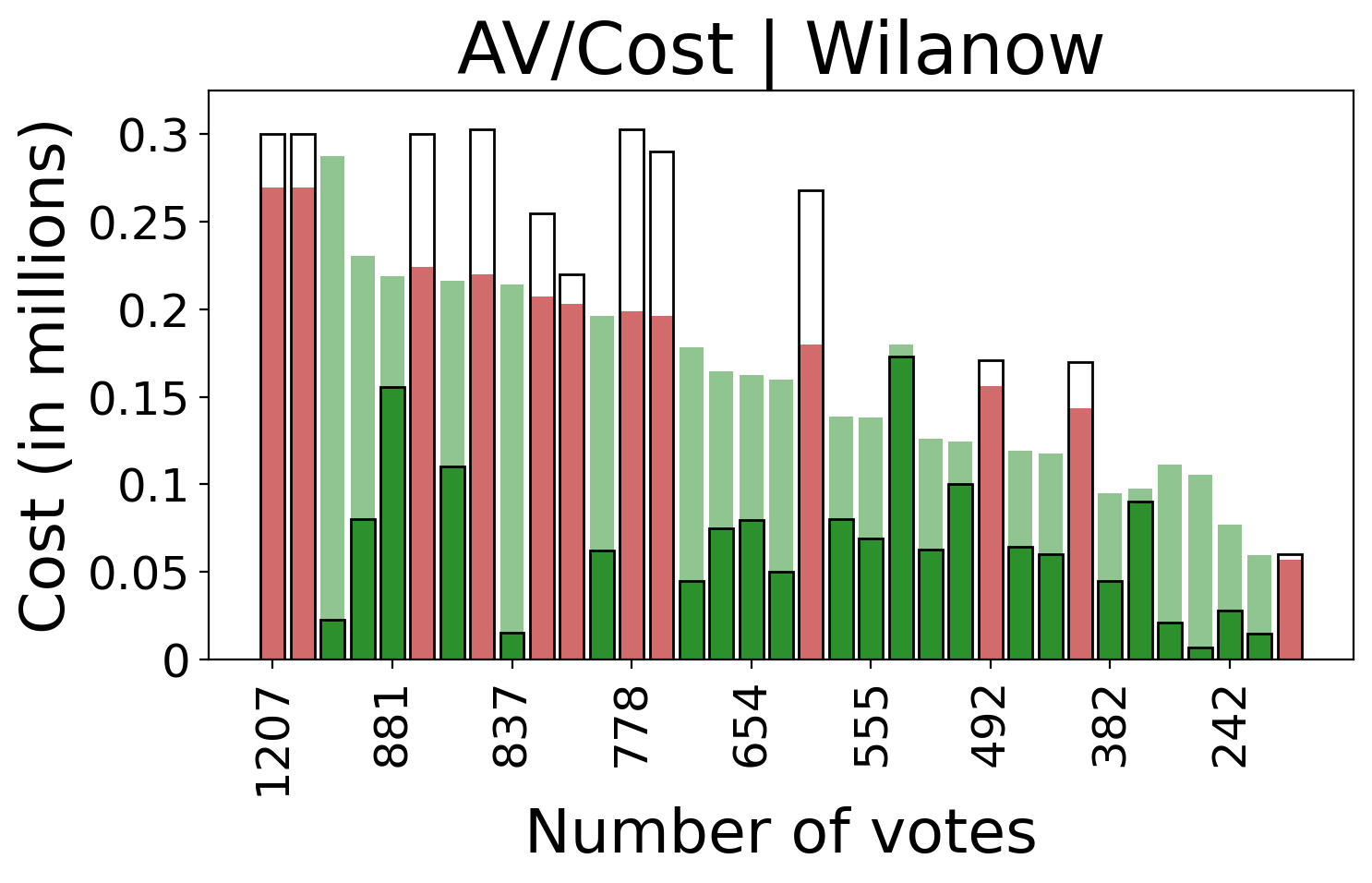}%
         \hspace{\hgap}%
     \includegraphics[width=\marginplotwidth]{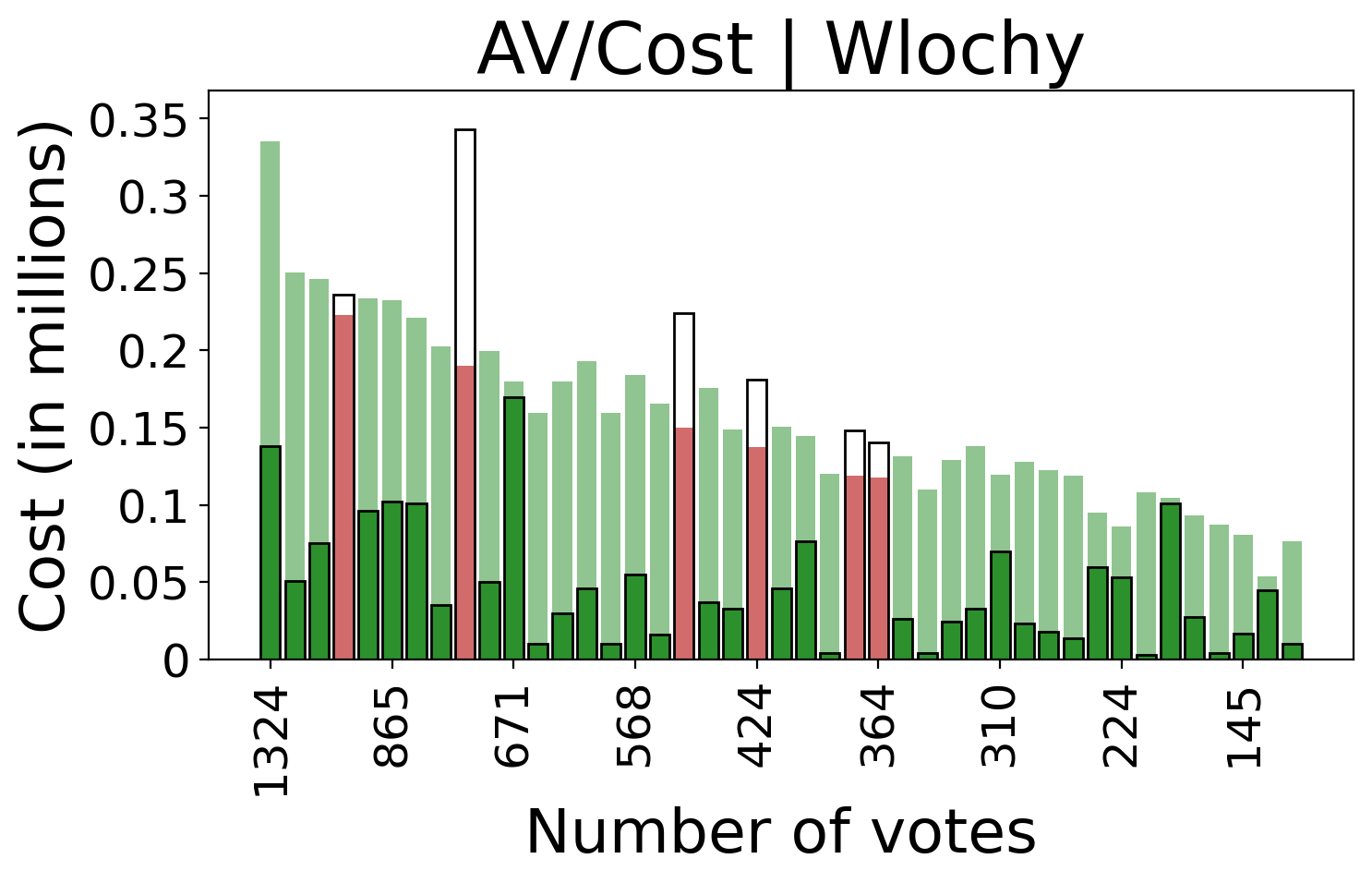}%

     \includegraphics[width=\marginplotwidth]{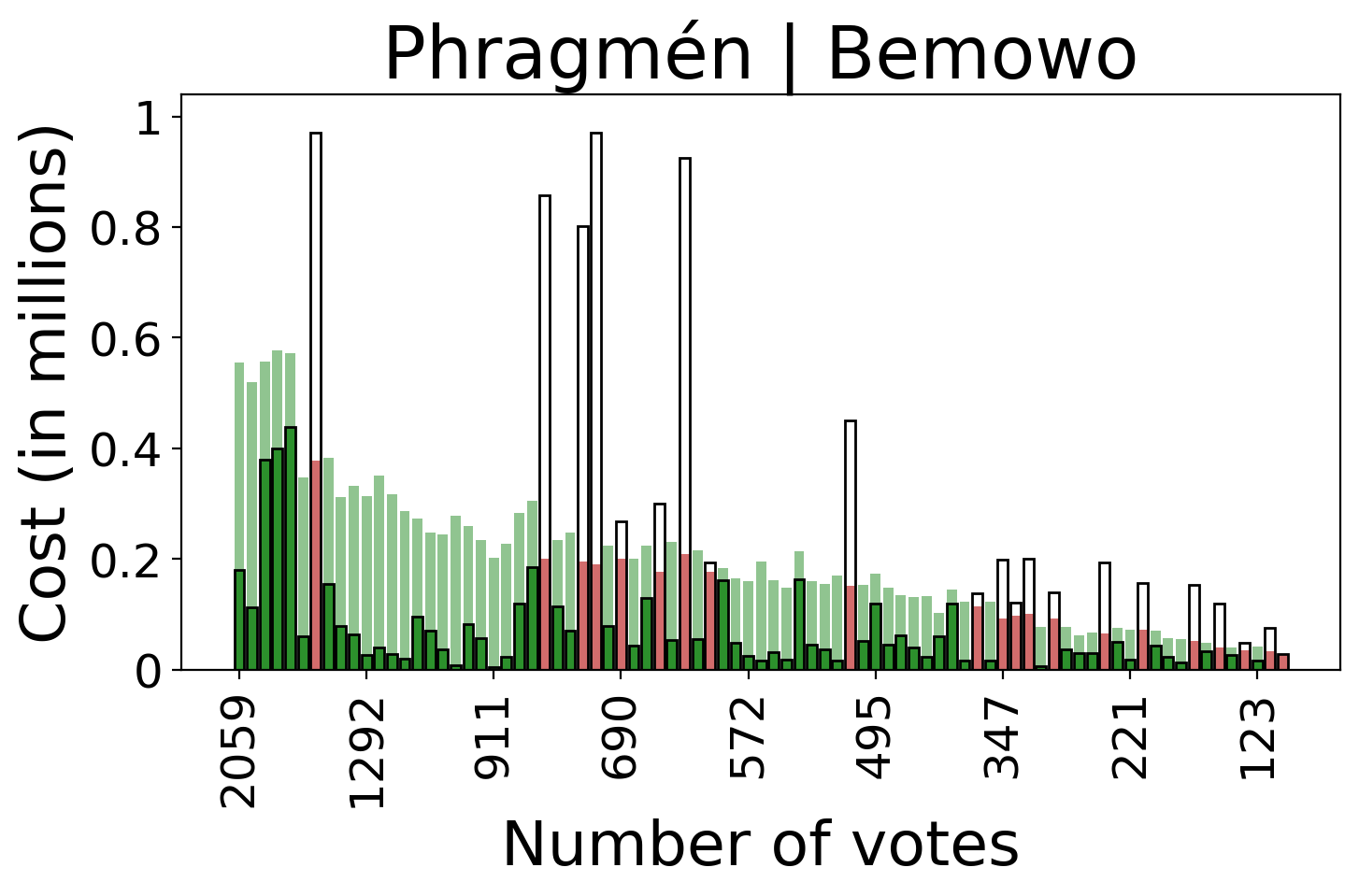}%
         \hspace{\hgap}%
     \includegraphics[width=\marginplotwidth]{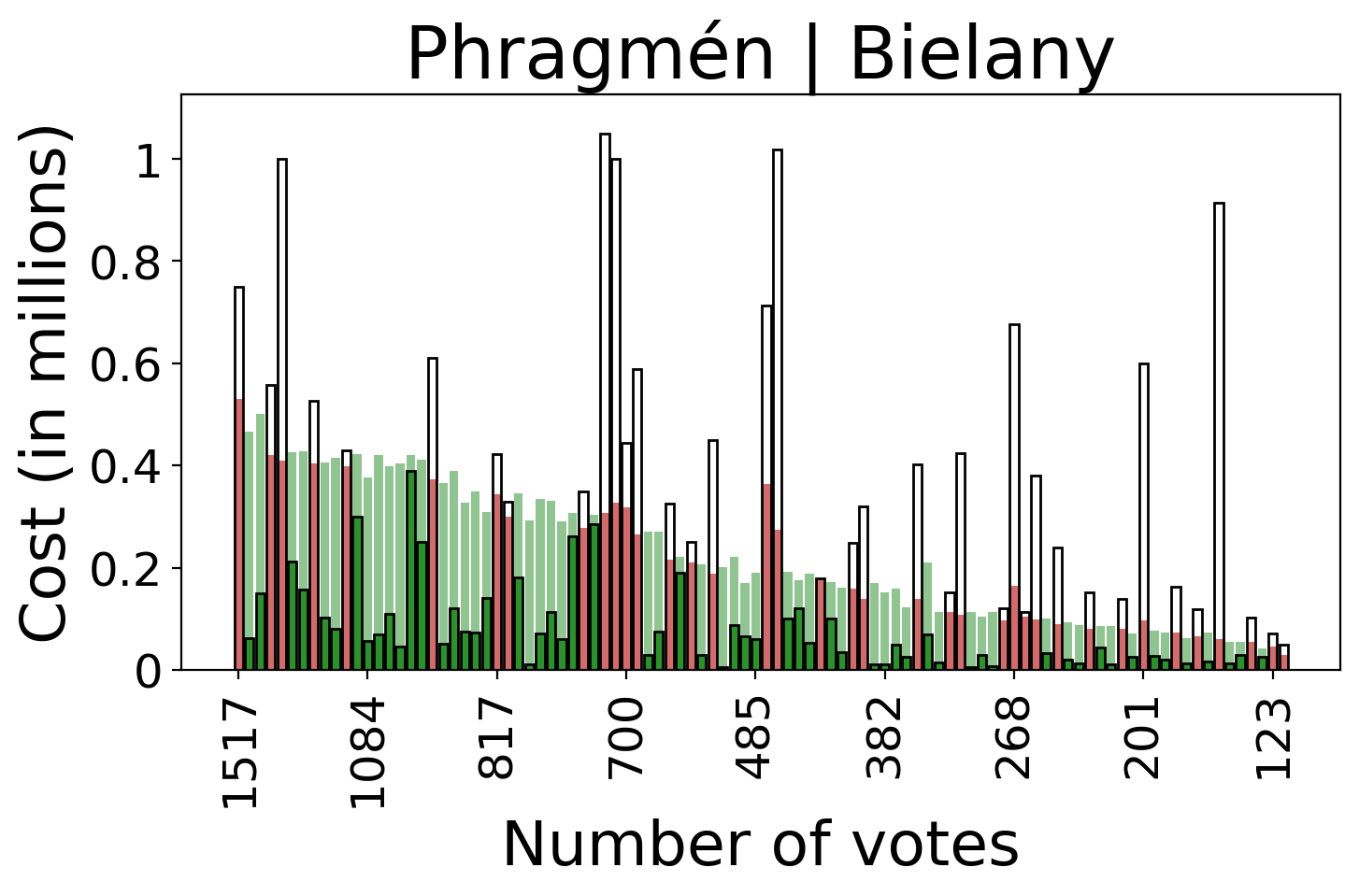}%
         \hspace{\hgap}%
     \includegraphics[width=\marginplotwidth]{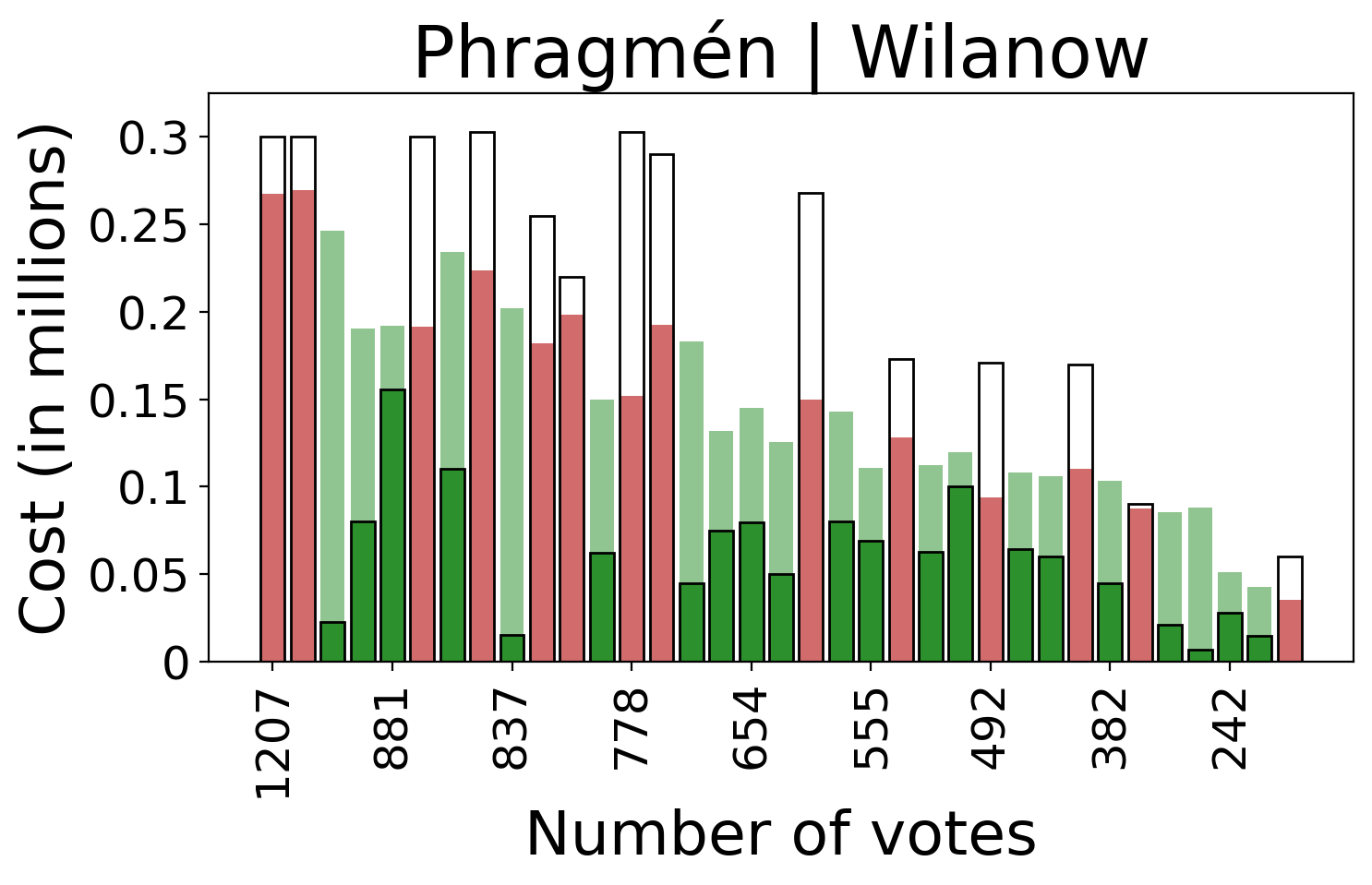}%
         \hspace{\hgap}%
     \includegraphics[width=\marginplotwidth]{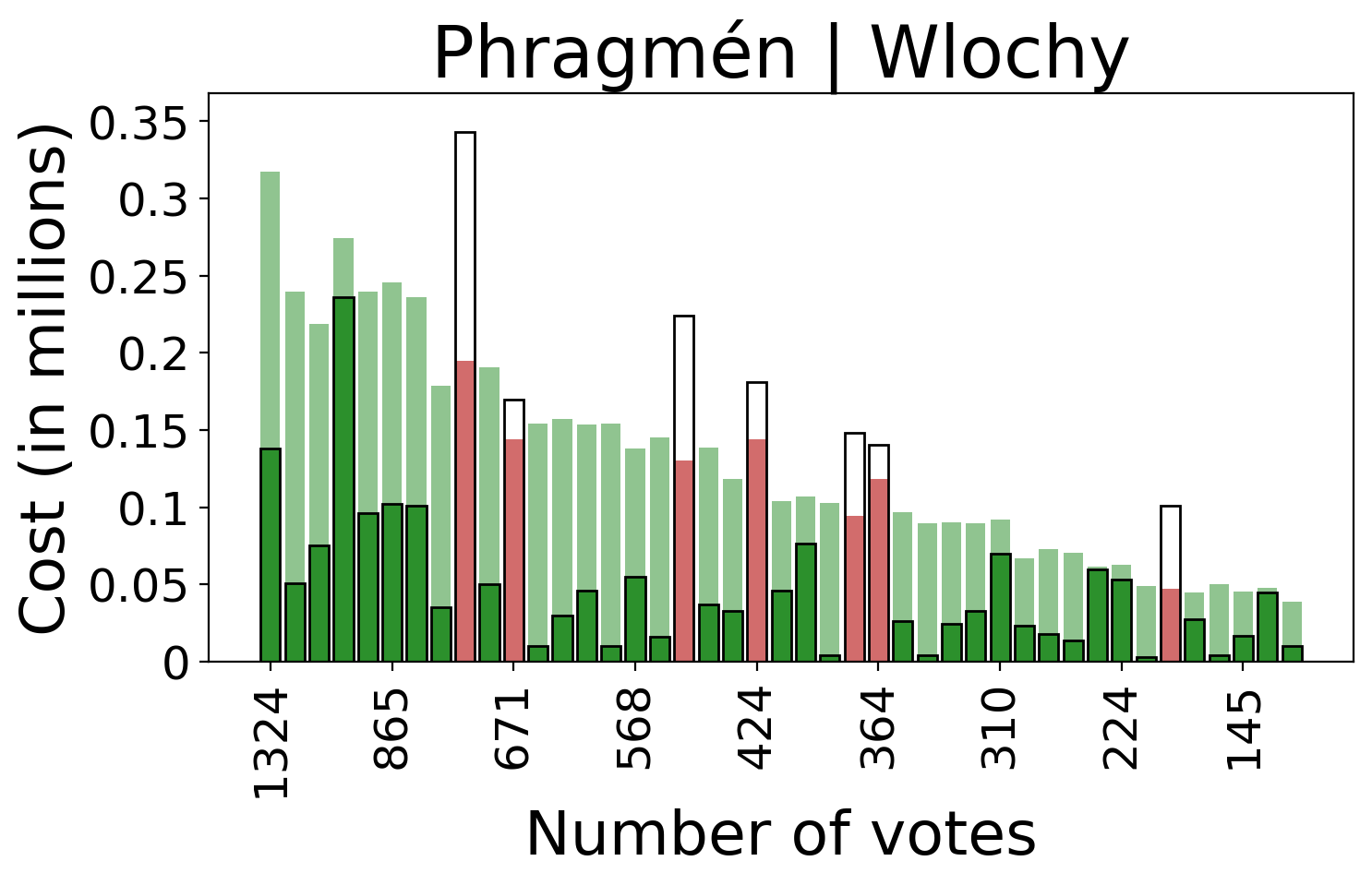}%

     \includegraphics[width=\marginplotwidth]{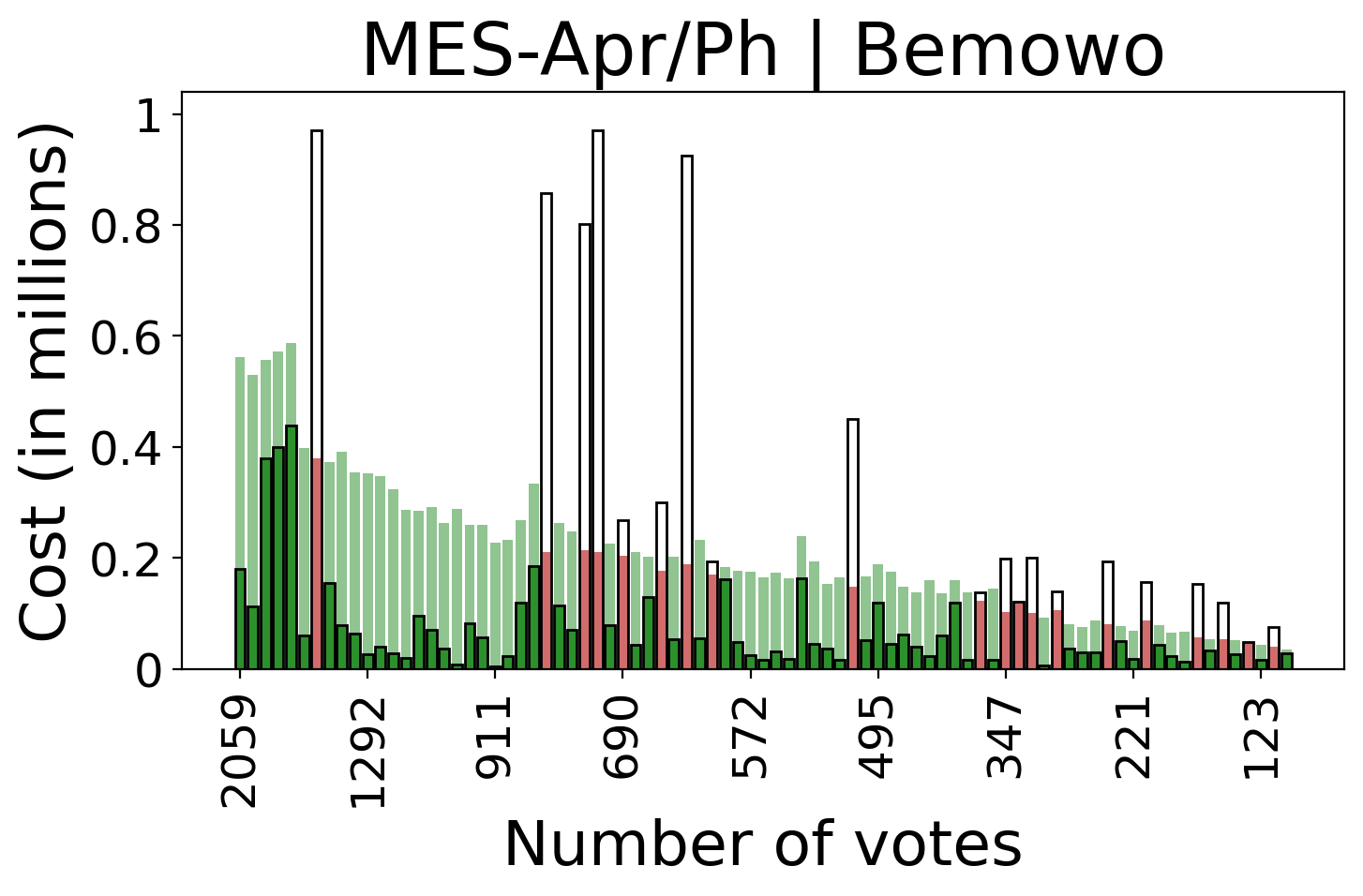}%
         \hspace{\hgap}%
     \includegraphics[width=\marginplotwidth]{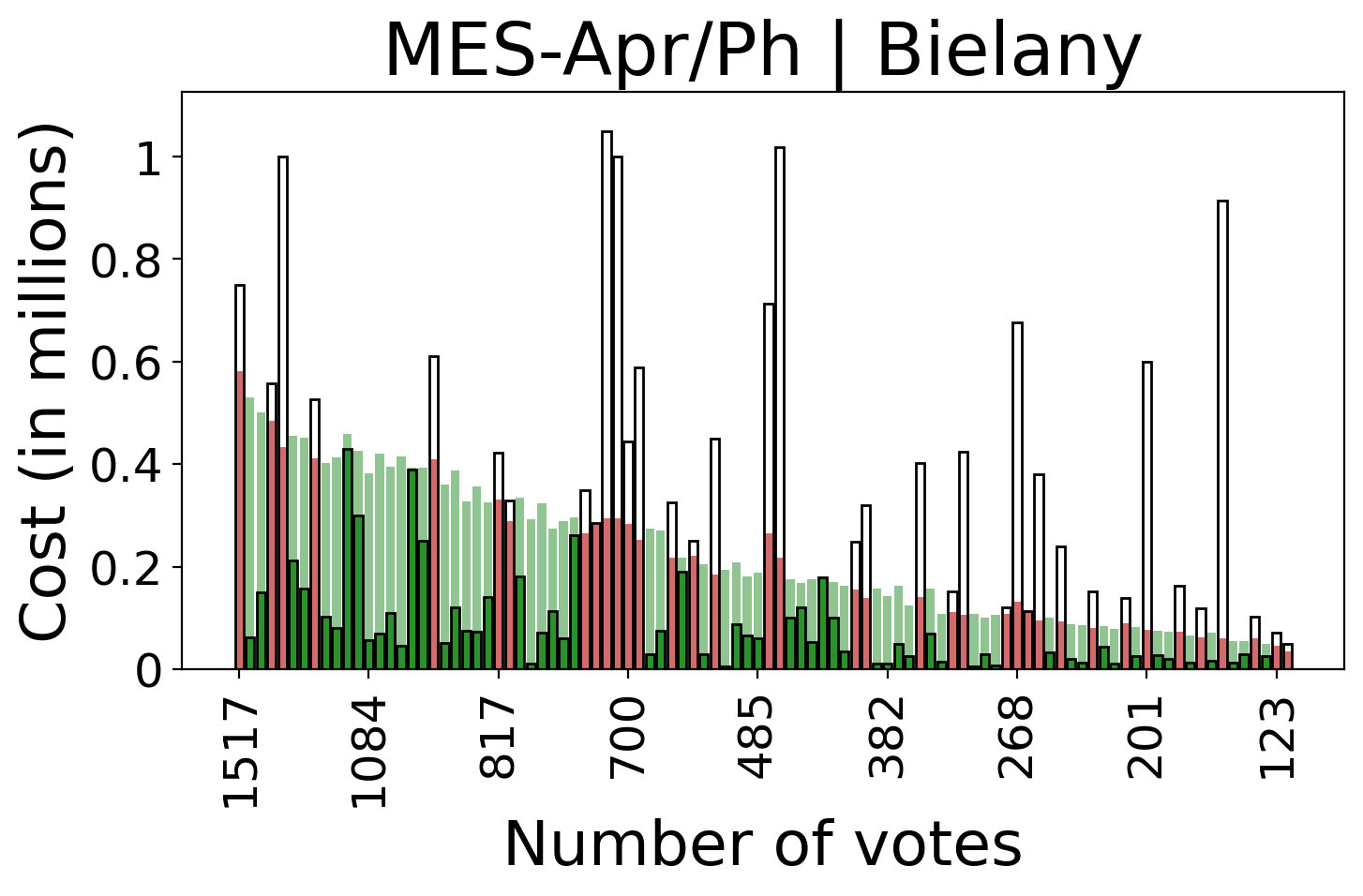}%
         \hspace{\hgap}%
     \includegraphics[width=\marginplotwidth]{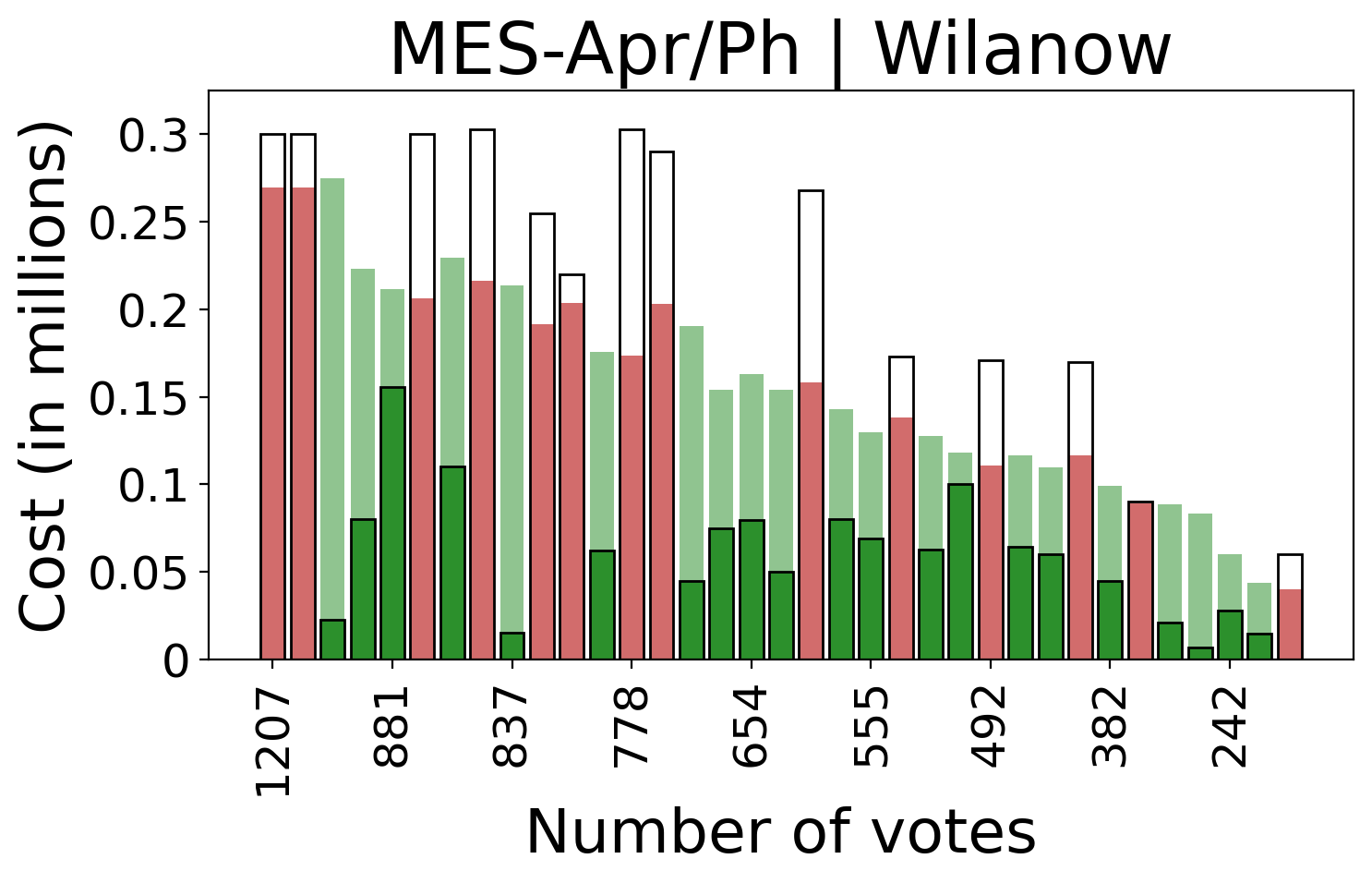}%
         \hspace{\hgap}%
     \includegraphics[width=\marginplotwidth]{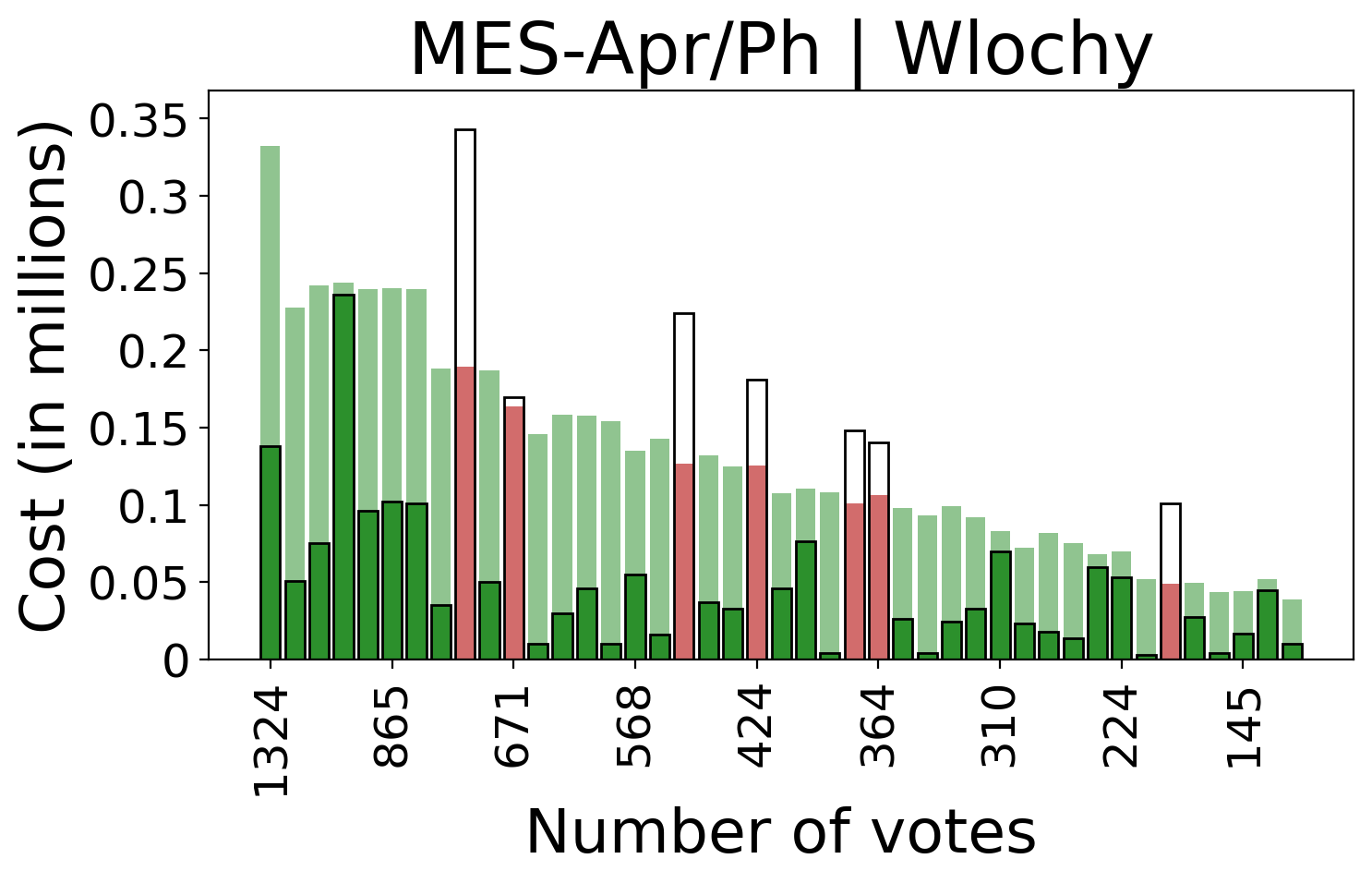}%

     \includegraphics[width=\marginplotwidth]{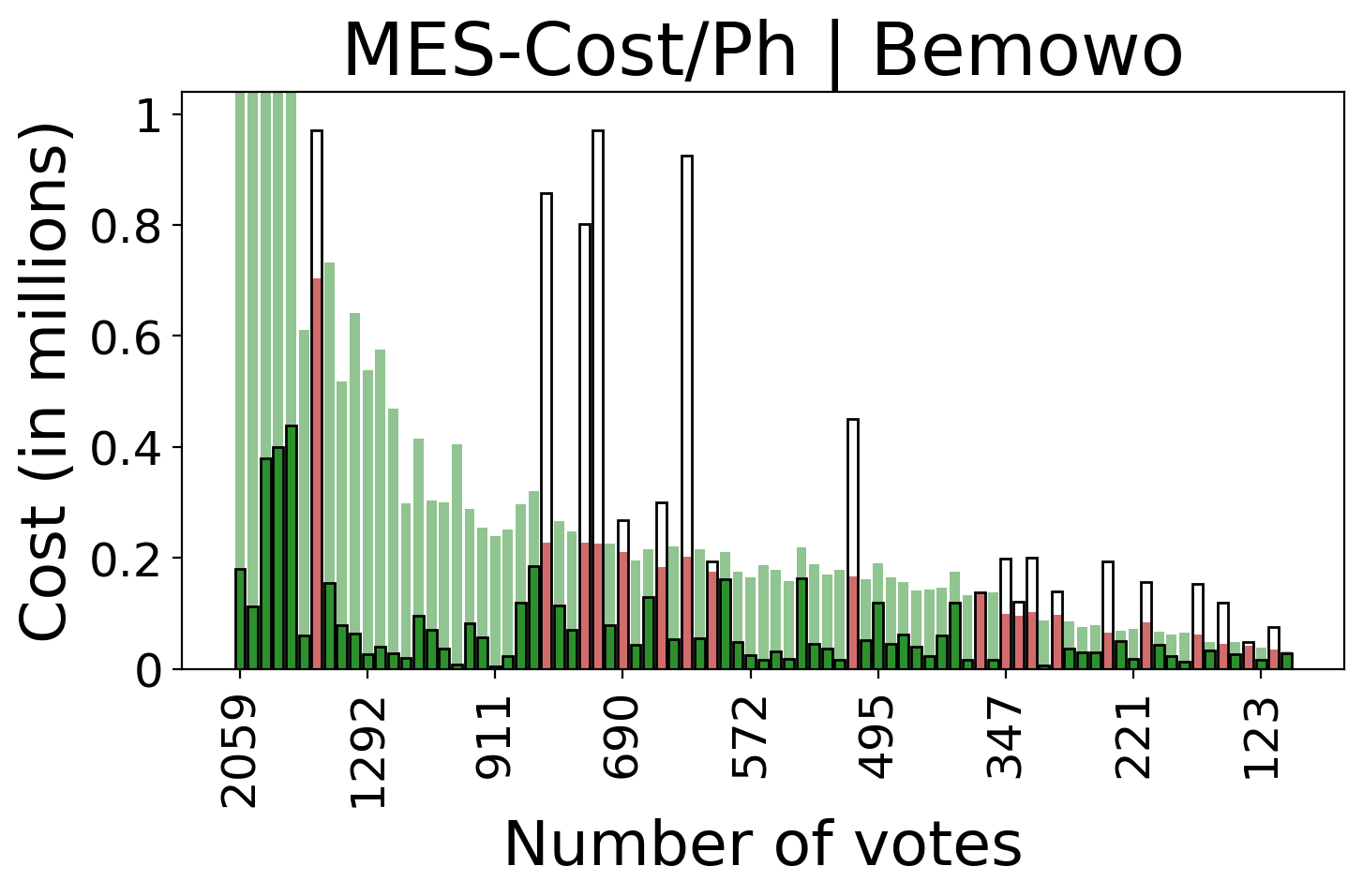}%
         \hspace{\hgap}%
     \includegraphics[width=\marginplotwidth]{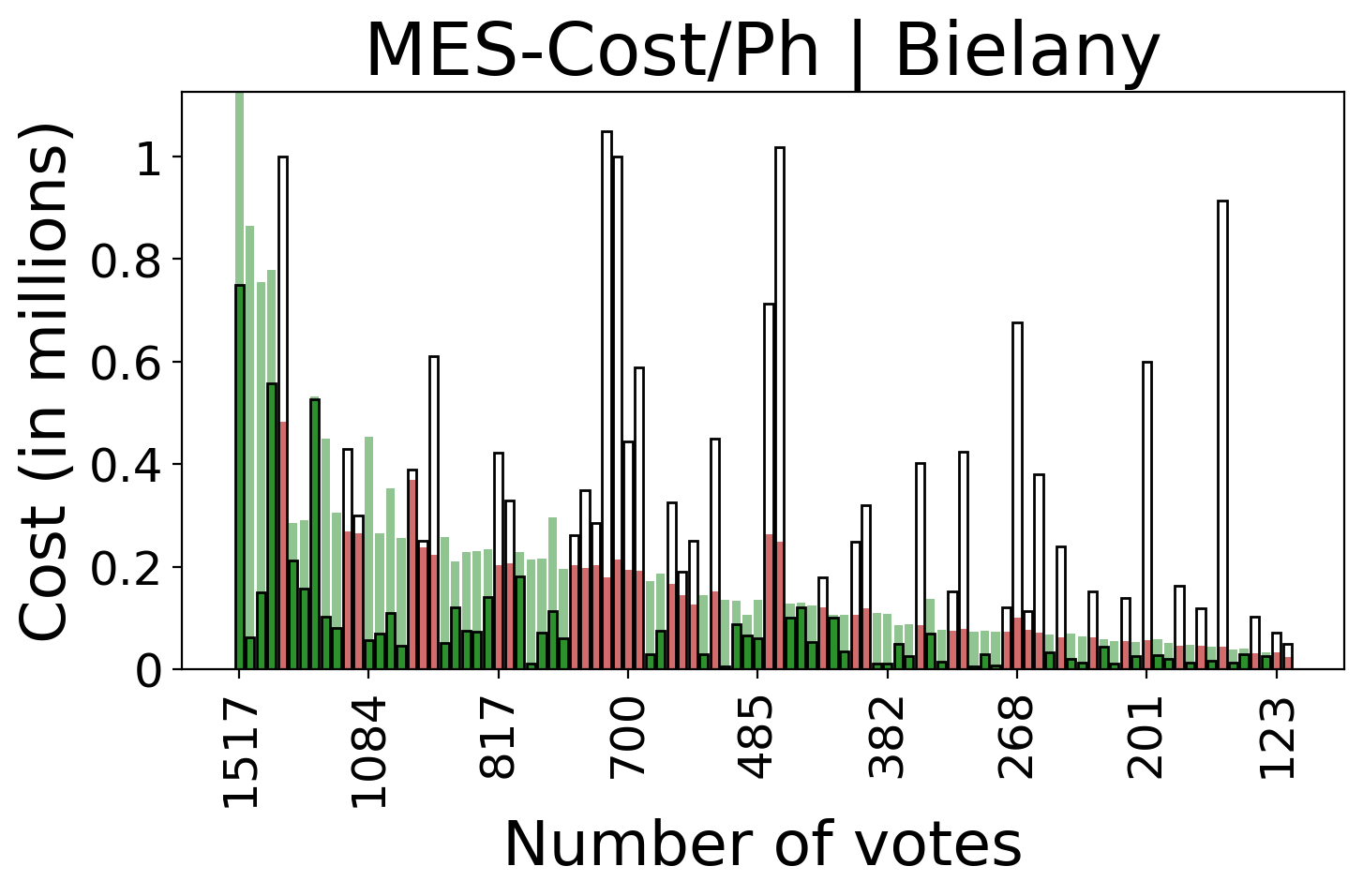}%
         \hspace{\hgap}%
     \includegraphics[width=\marginplotwidth]{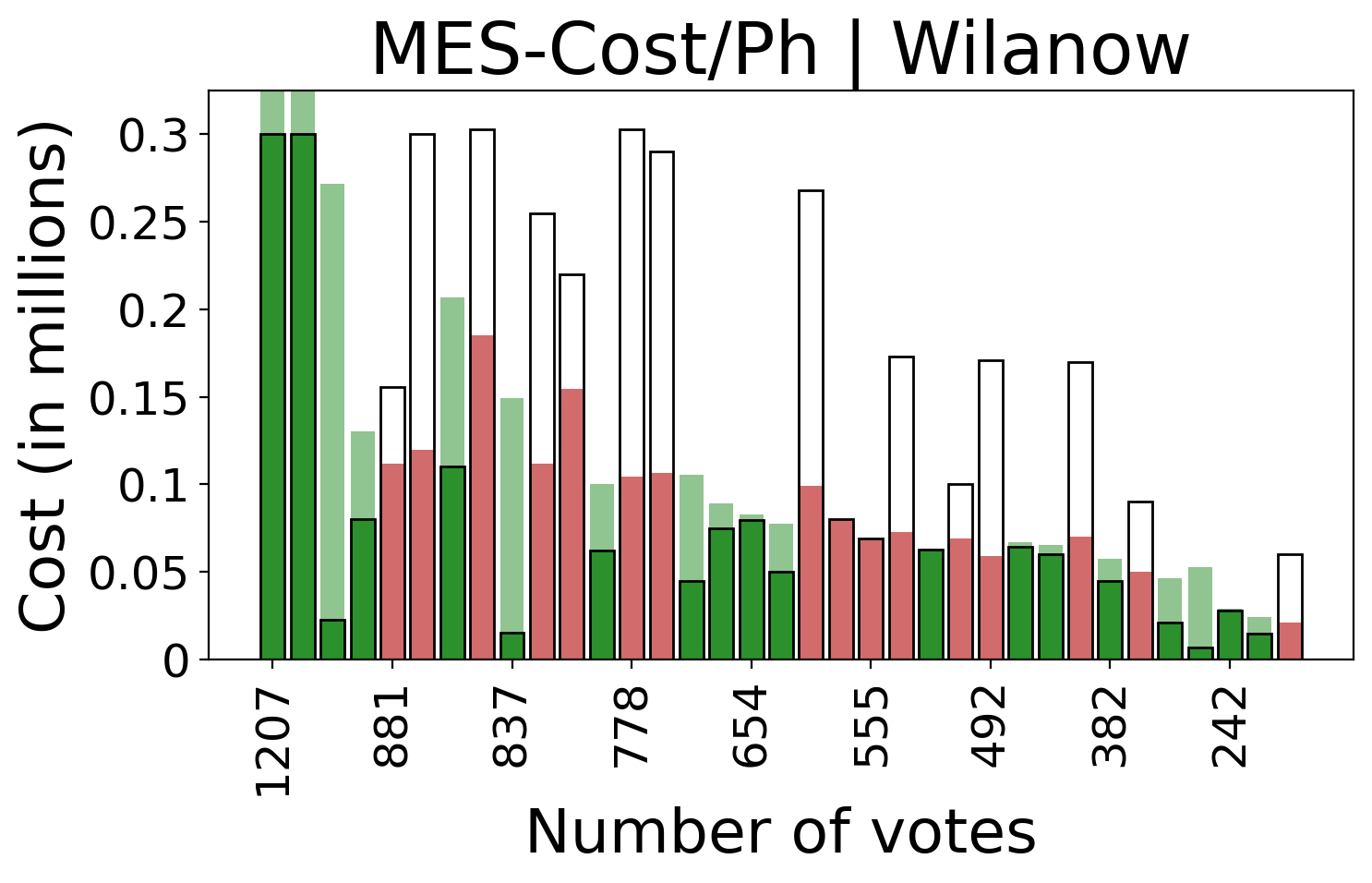}%
         \hspace{\hgap}%
     \includegraphics[width=\marginplotwidth]{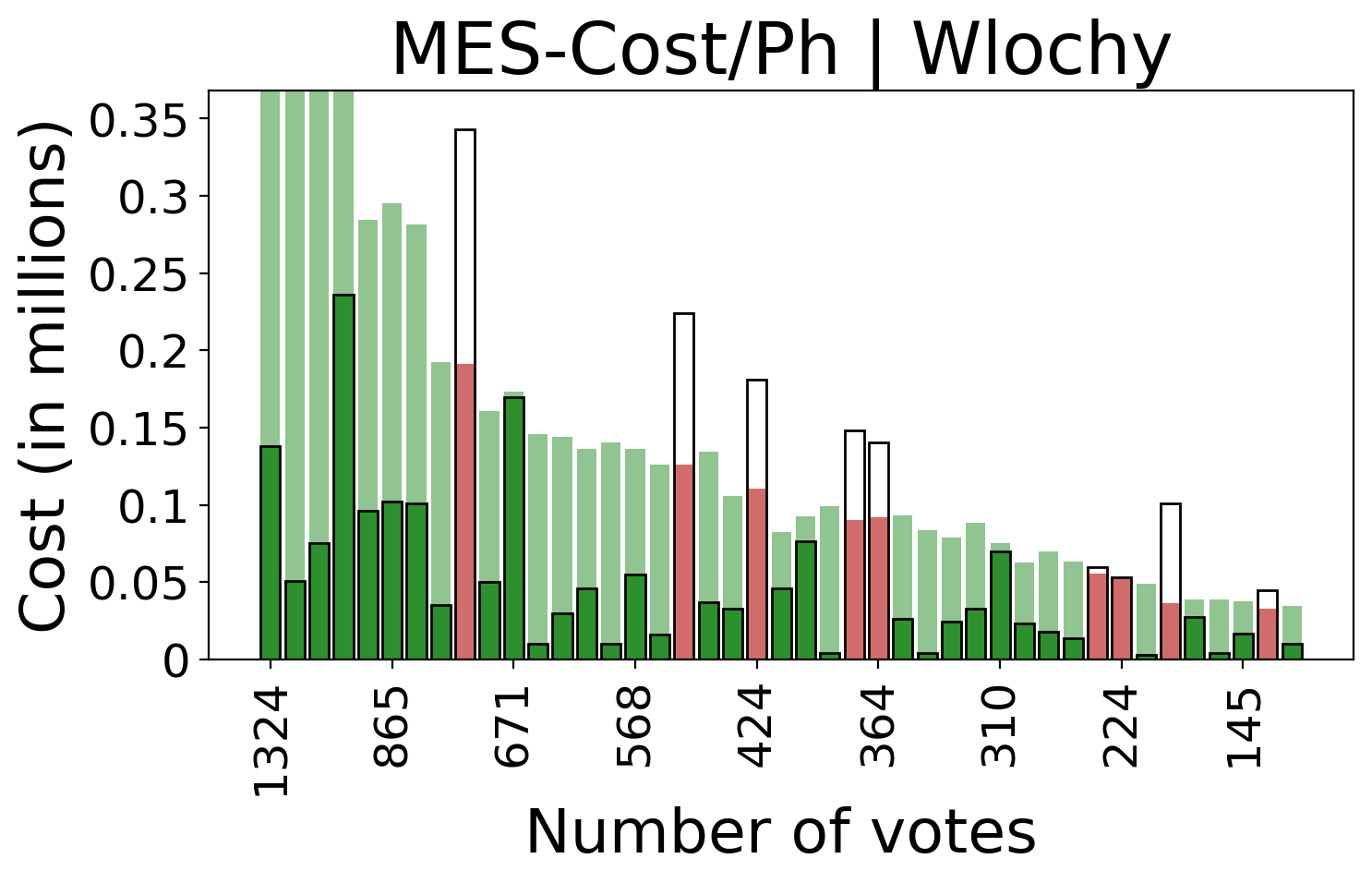}%

    \caption{Winning and losing margins in real-life PB. The green bars denote the max-winning cost (with the brighter part depicting the winning margin). The red bars denote the min-losing cost (with the white part depicting the losing margin). Black outlines denote the original costs of the projects.}\label{fig:margins_apdx}
\end{figure*}

\begin{figure*}[t]
     \centering
     \includegraphics[width=\marginplotwidth]{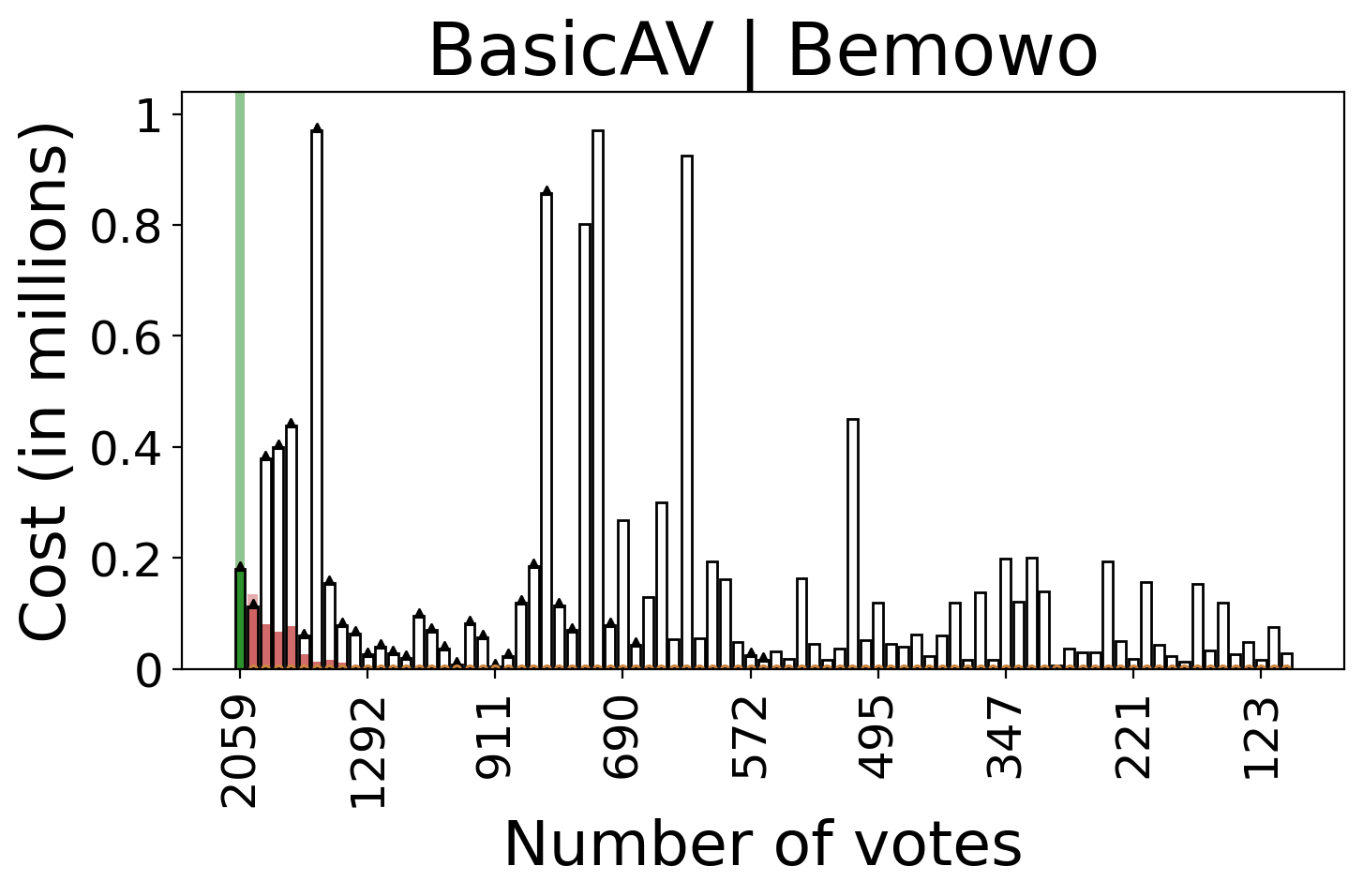}%
         \hspace{\hgap}%
     \includegraphics[width=\marginplotwidth]{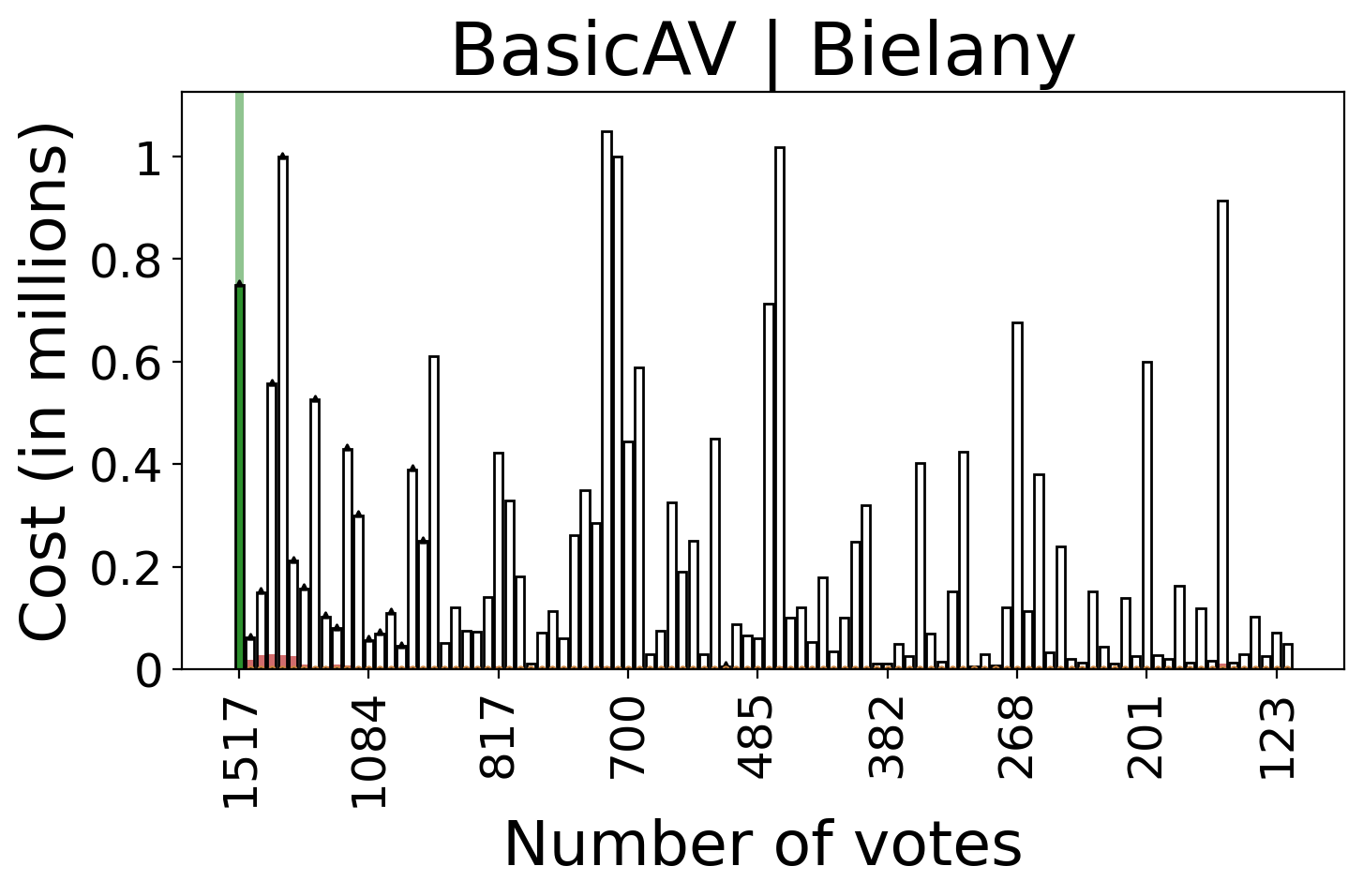}%
         \hspace{\hgap}%
     \includegraphics[width=\marginplotwidth]{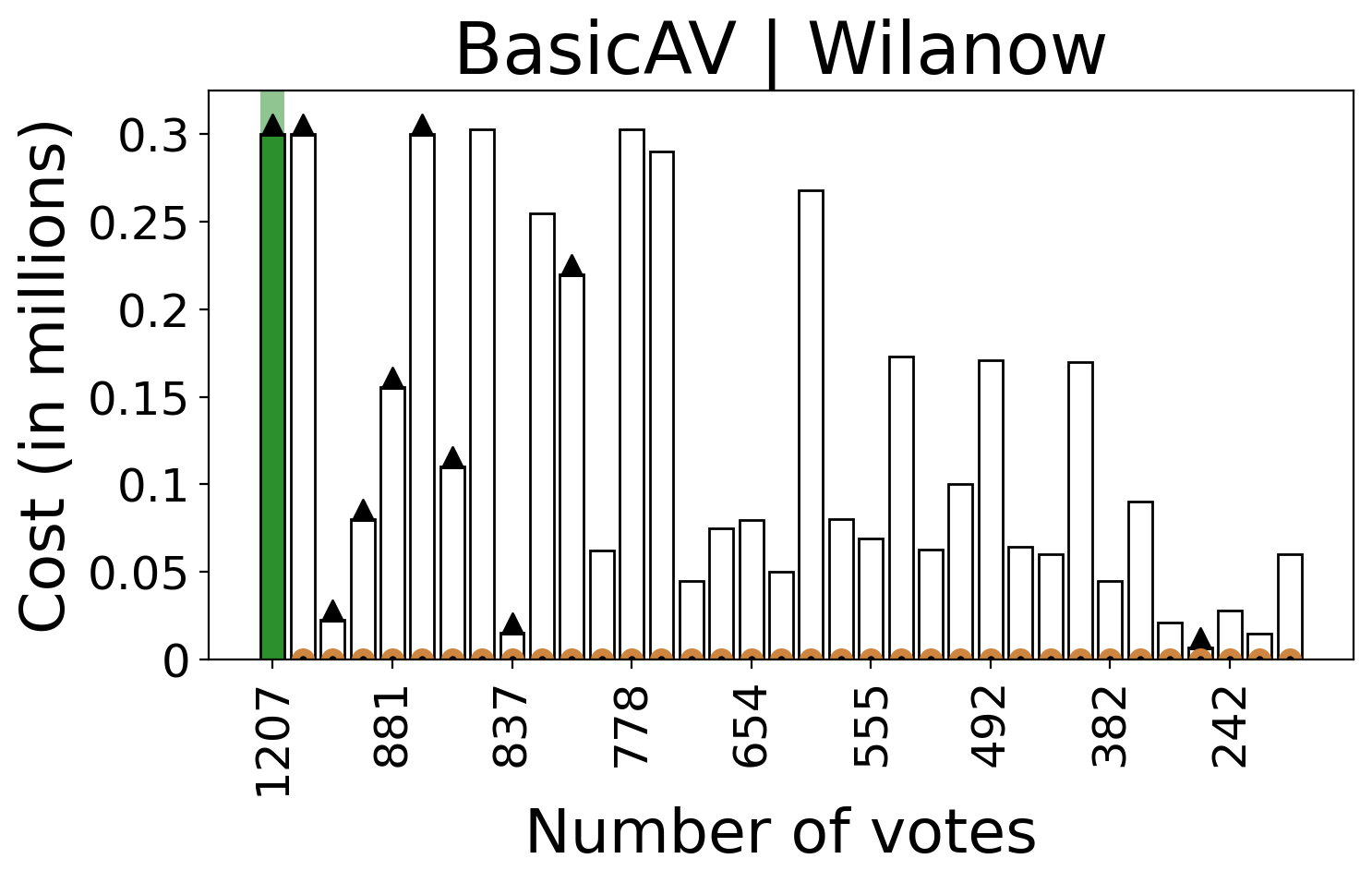}%
         \hspace{\hgap}%
     \includegraphics[width=\marginplotwidth]{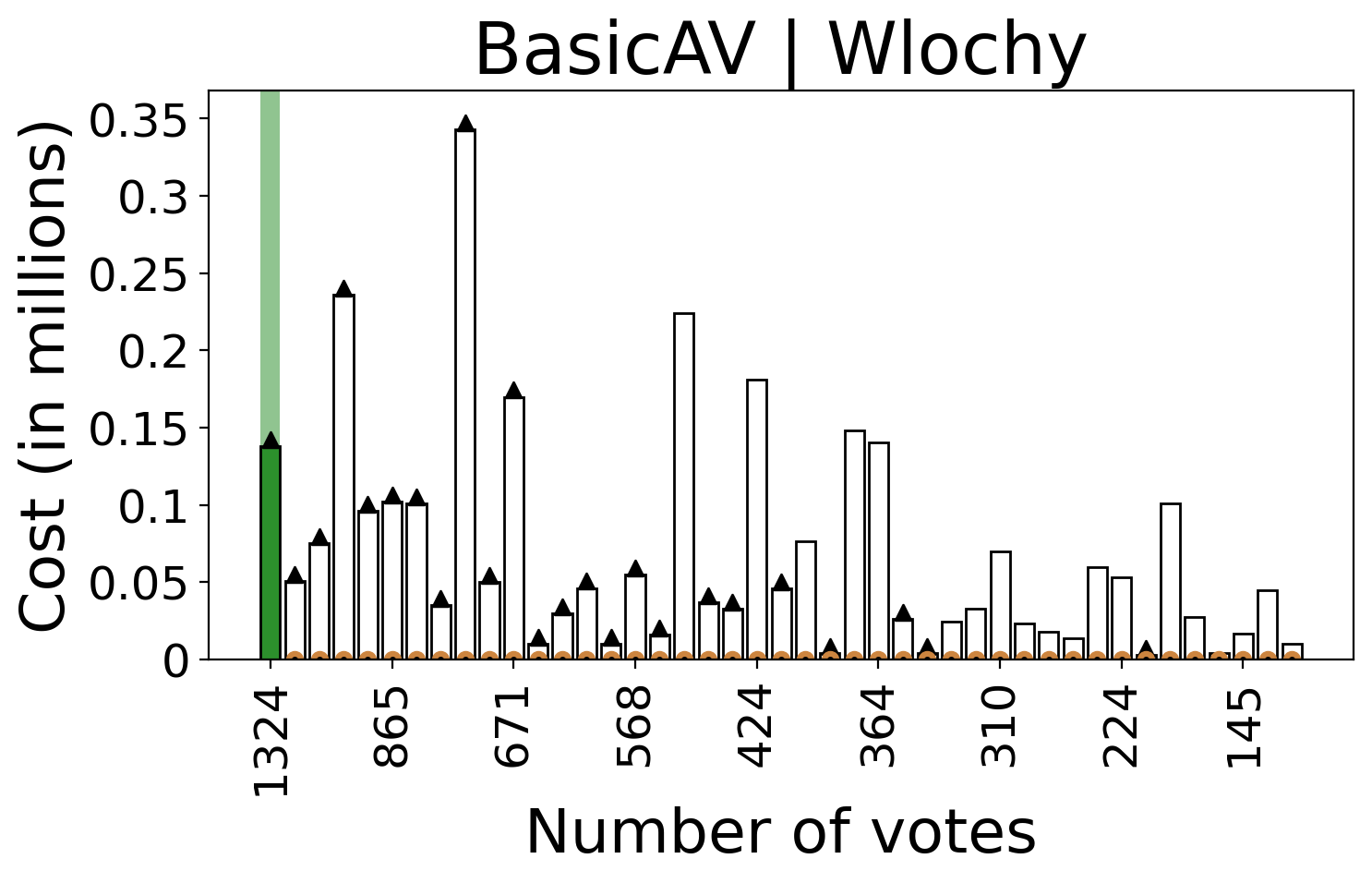}%

     \includegraphics[width=\marginplotwidth]{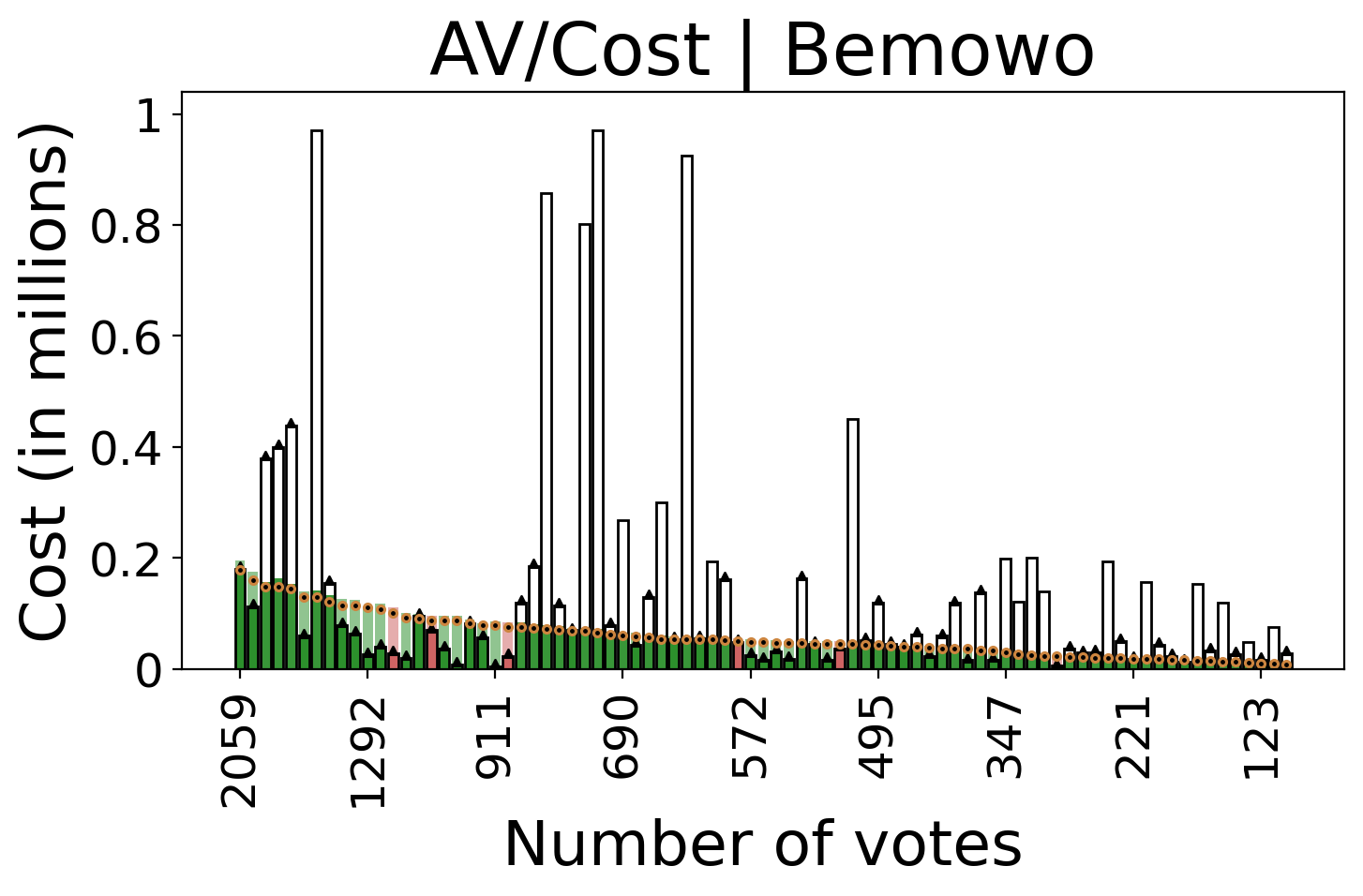}%
         \hspace{\hgap}%
     \includegraphics[width=\marginplotwidth]{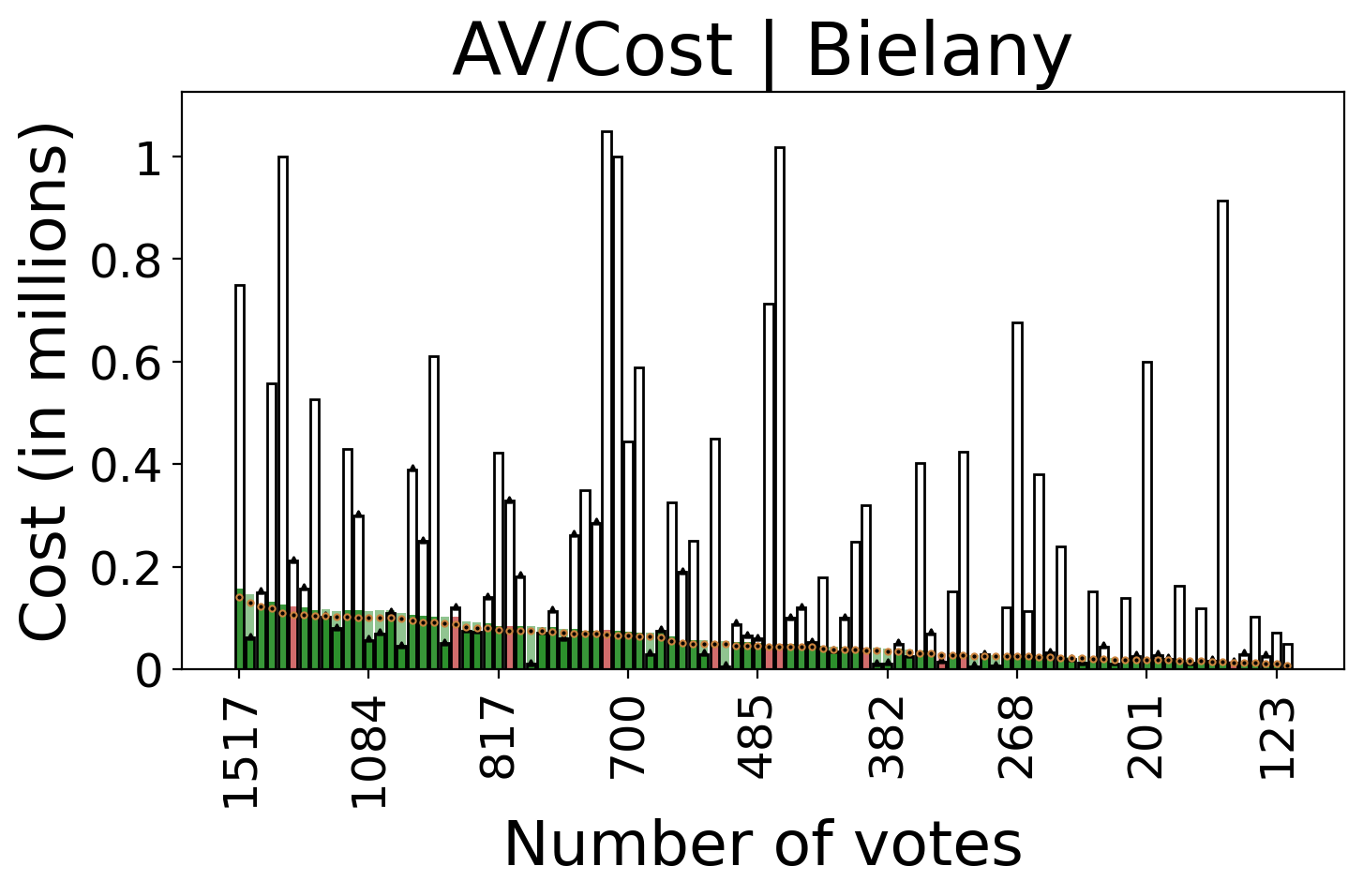}%
         \hspace{\hgap}%
     \includegraphics[width=\marginplotwidth]{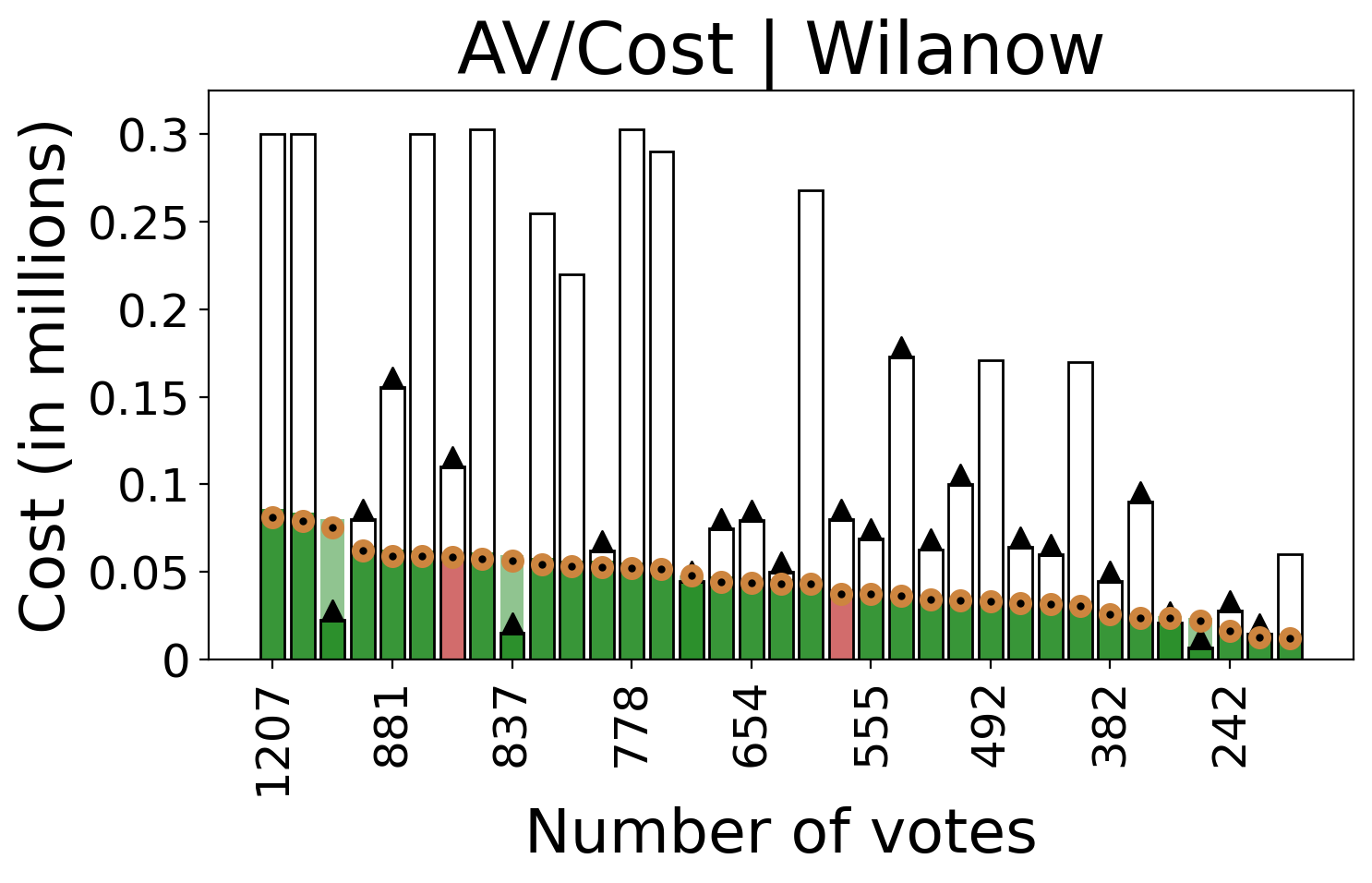}%
         \hspace{\hgap}%
     \includegraphics[width=\marginplotwidth]{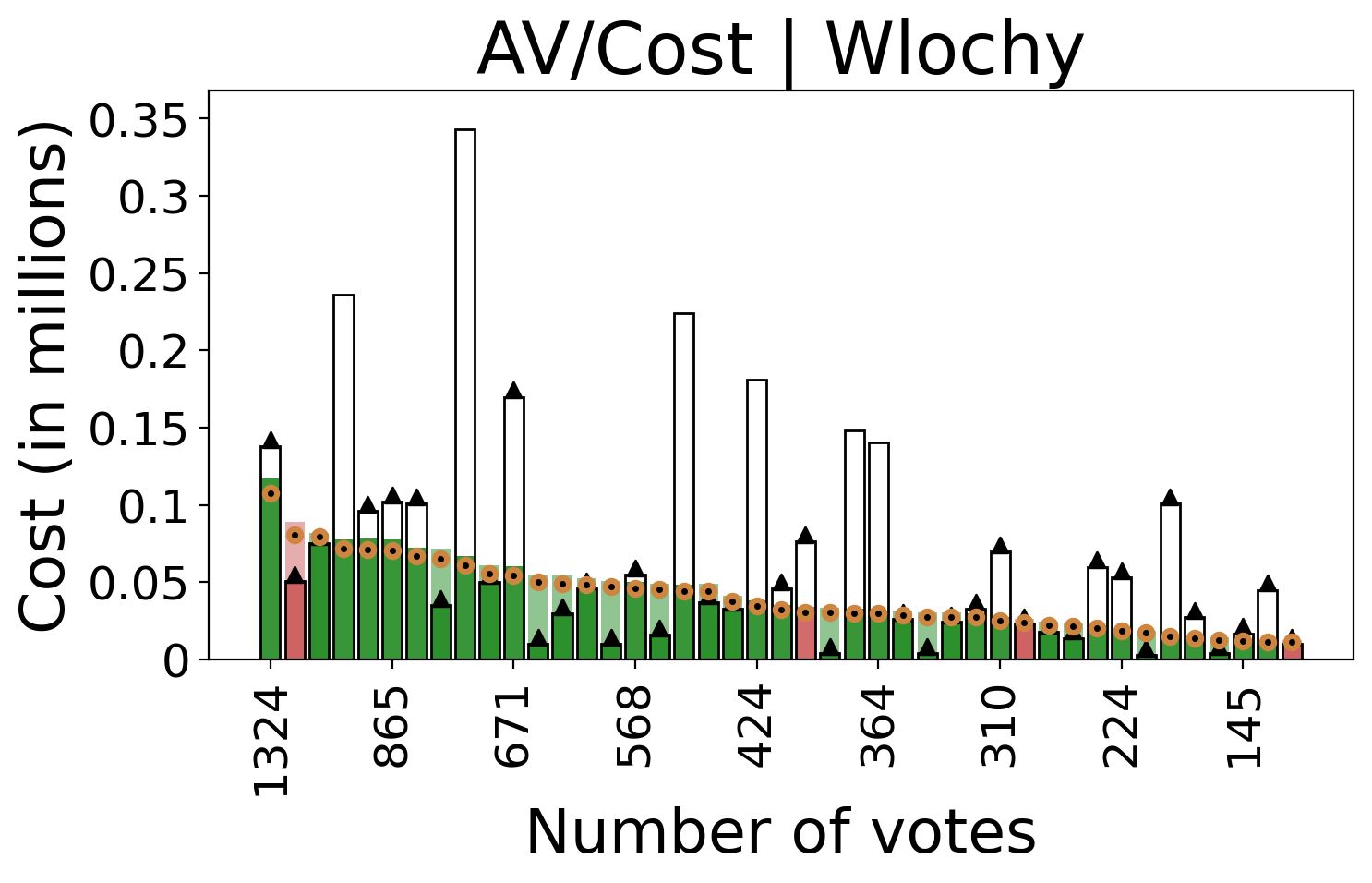}%

     \includegraphics[width=\marginplotwidth]{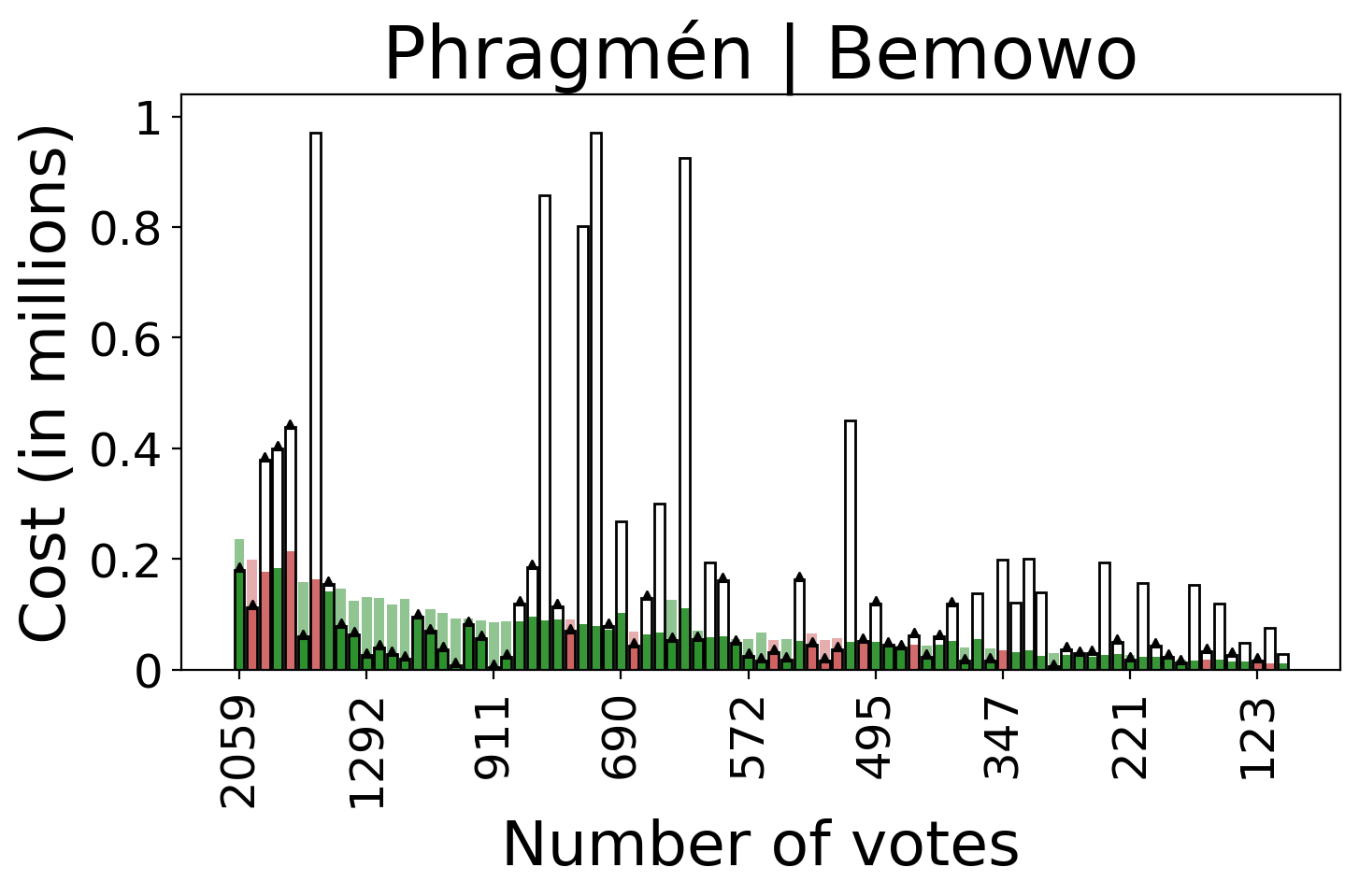}%
         \hspace{\hgap}%
     \includegraphics[width=\marginplotwidth]{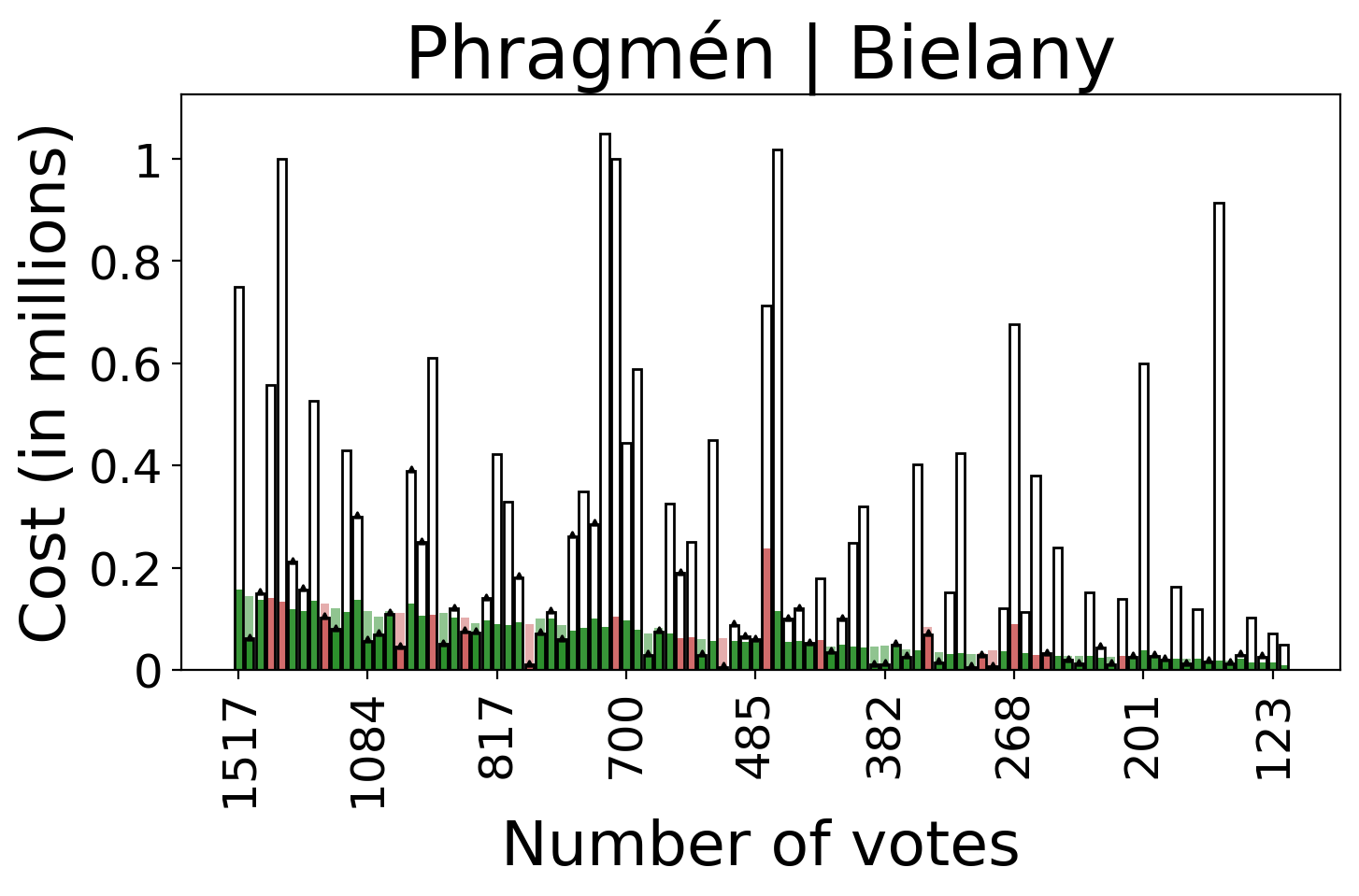}%
         \hspace{\hgap}%
     \includegraphics[width=\marginplotwidth]{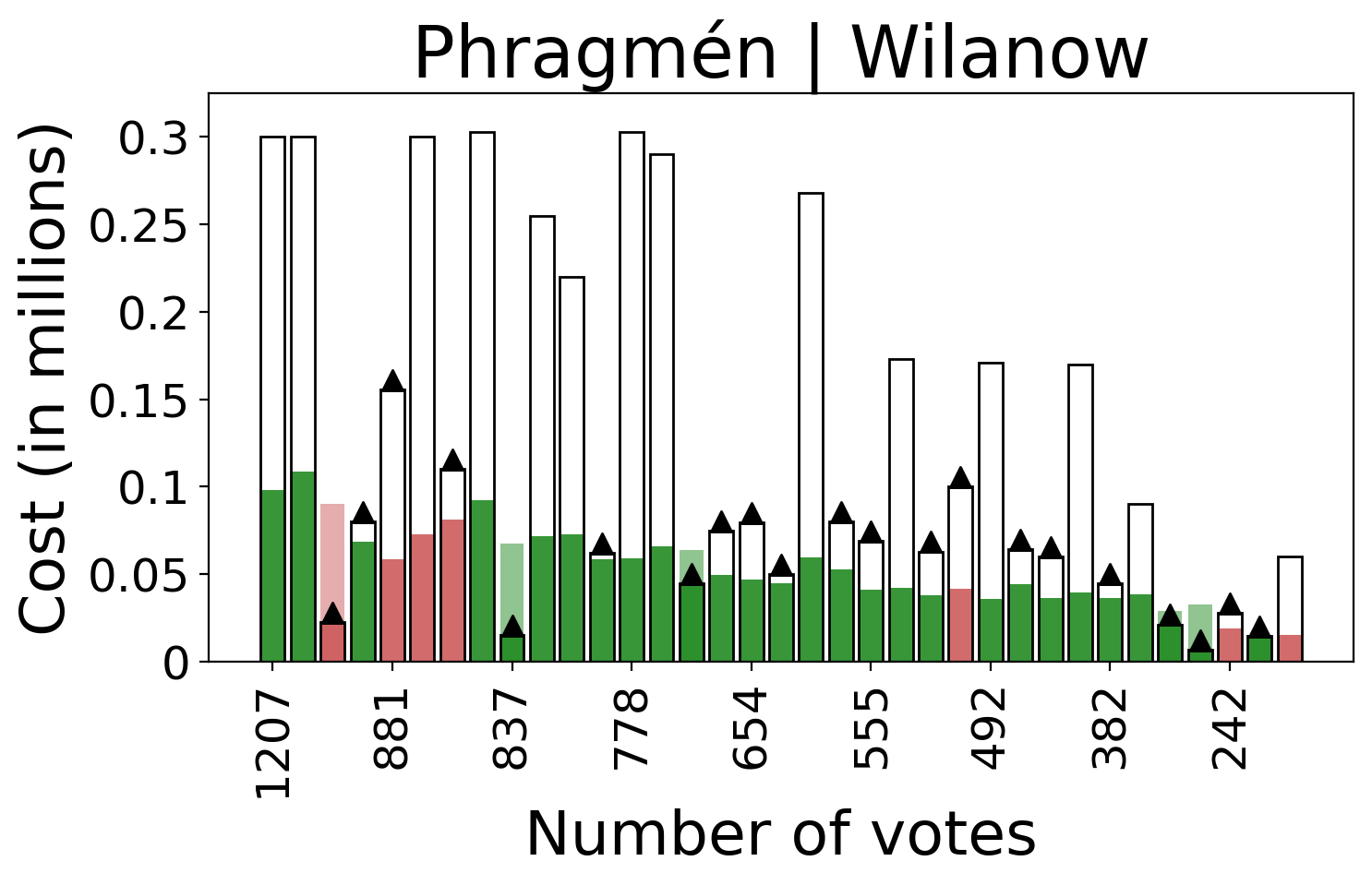}%
         \hspace{\hgap}%
     \includegraphics[width=\marginplotwidth]{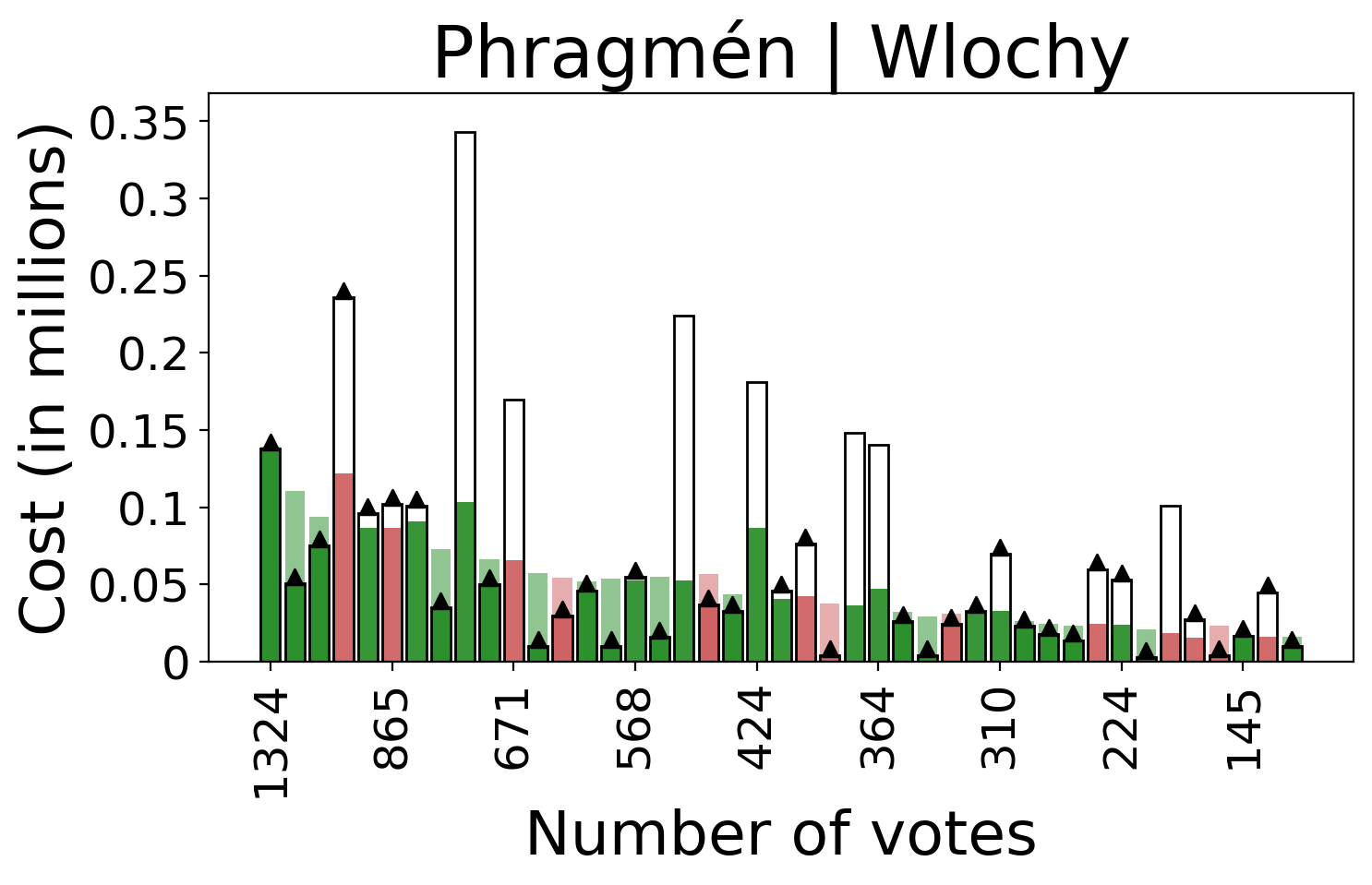}%

     \includegraphics[width=\marginplotwidth]{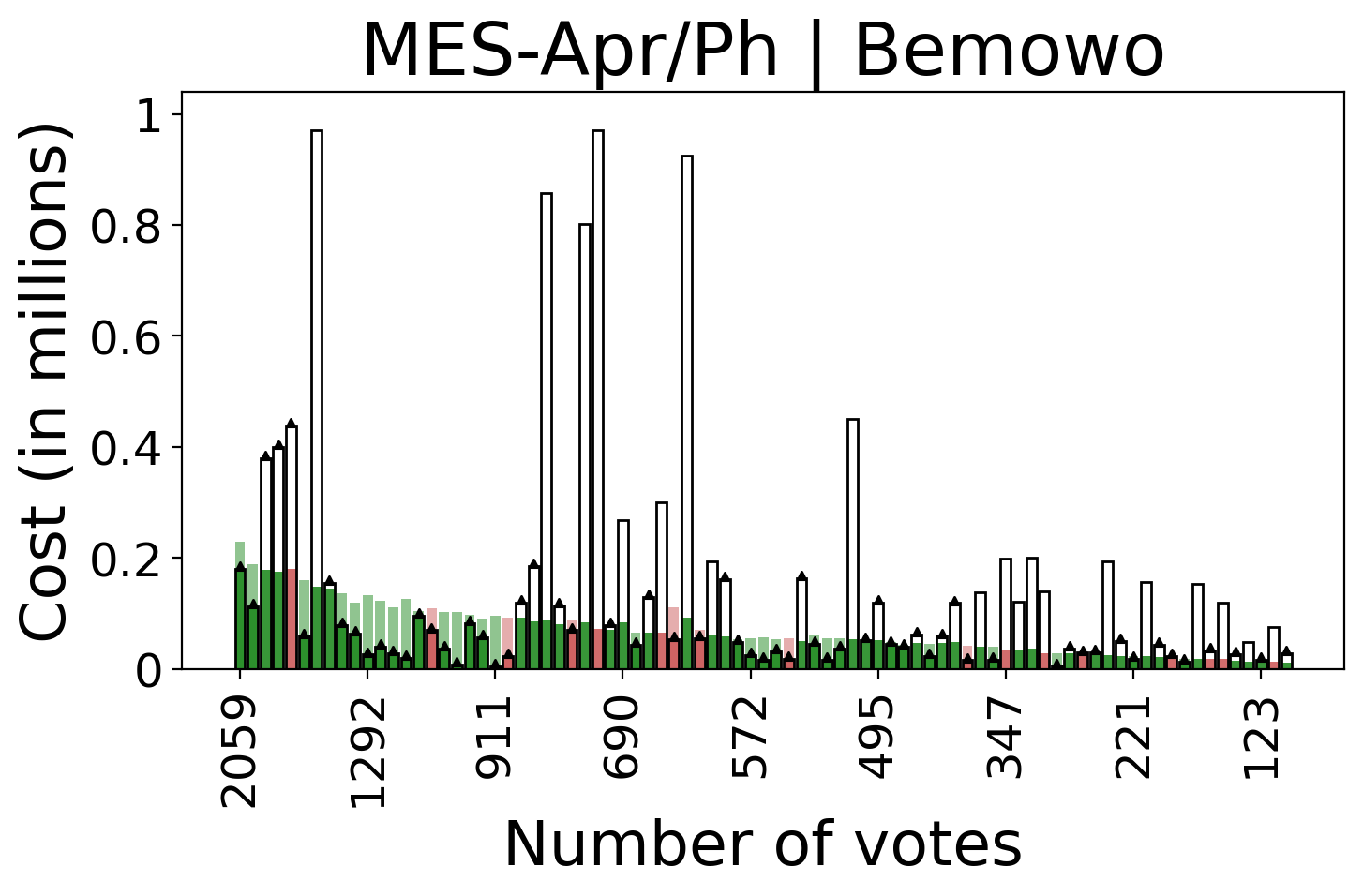}%
         \hspace{\hgap}%
     \includegraphics[width=\marginplotwidth]{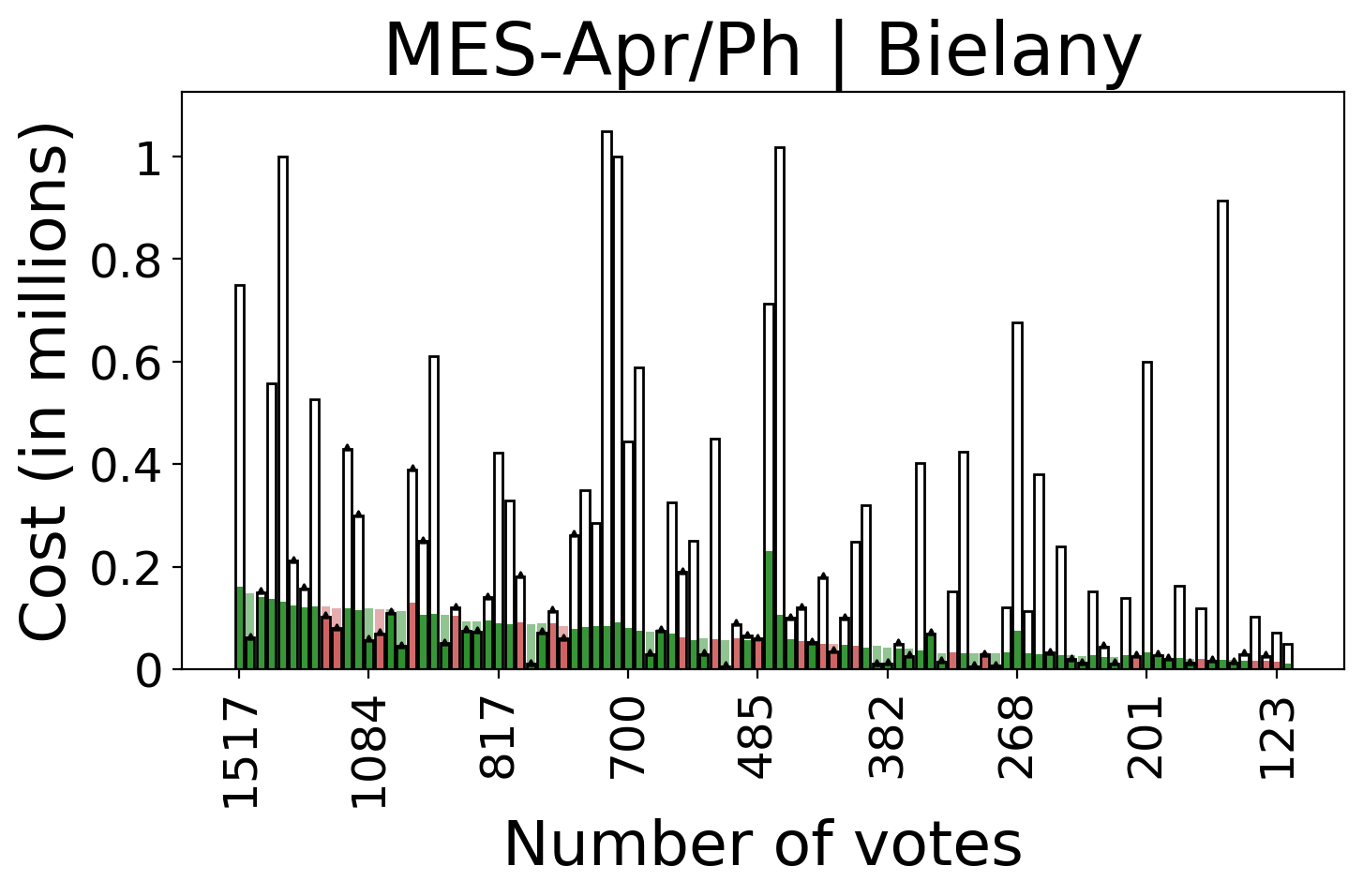}%
         \hspace{\hgap}%
     \includegraphics[width=\marginplotwidth]{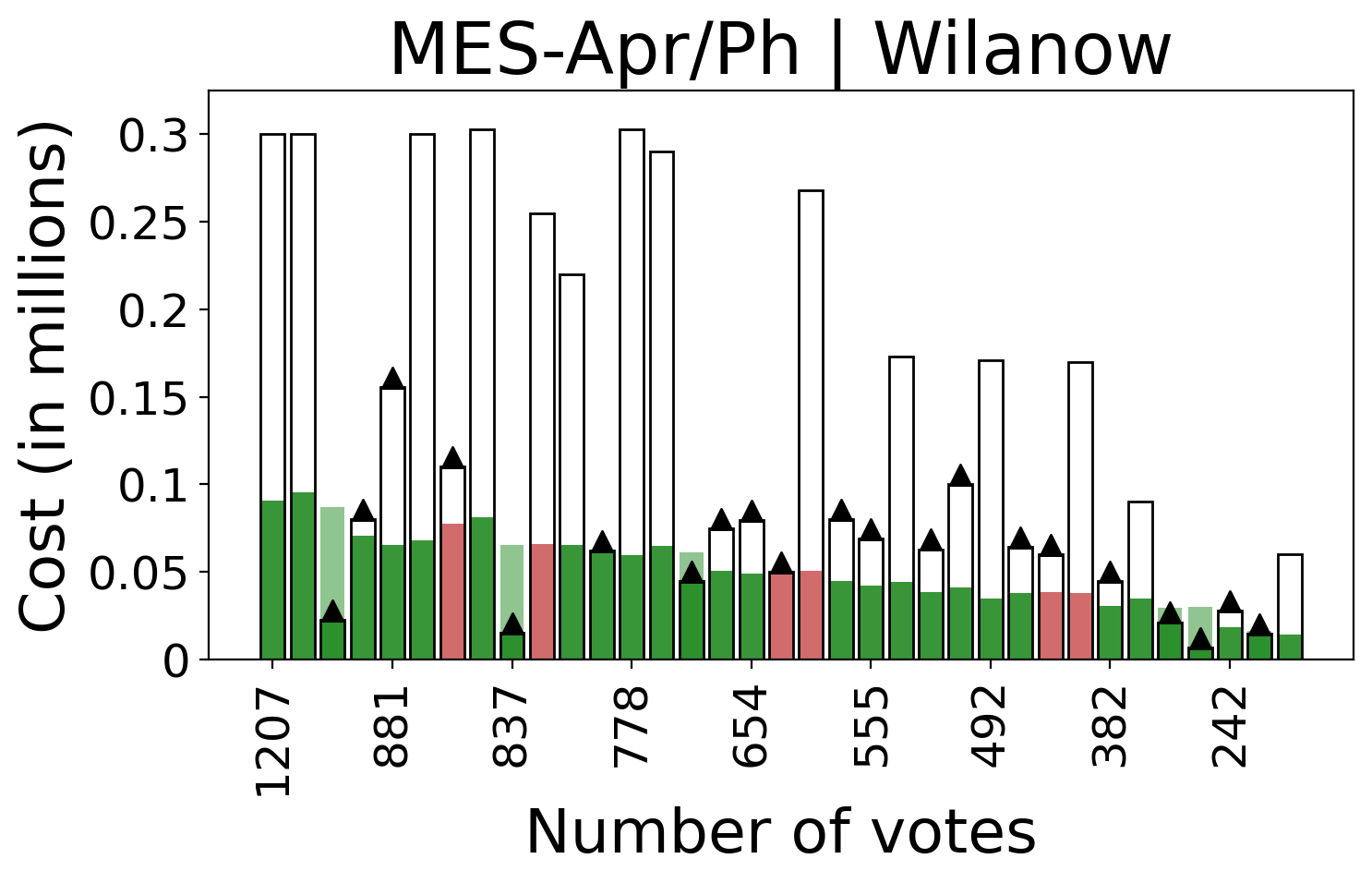}%
         \hspace{\hgap}%
     \includegraphics[width=\marginplotwidth]{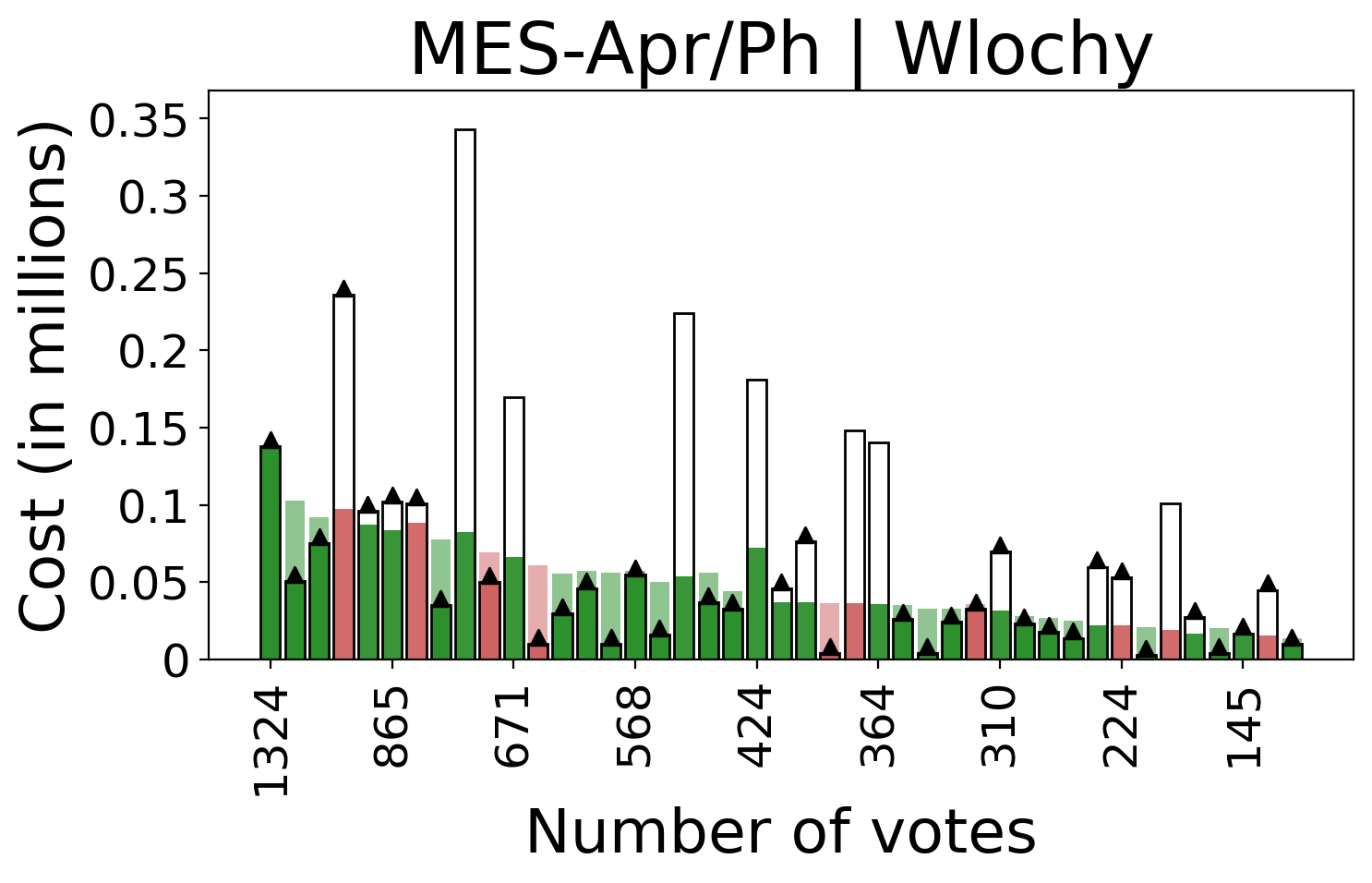}%

     \includegraphics[width=\marginplotwidth]{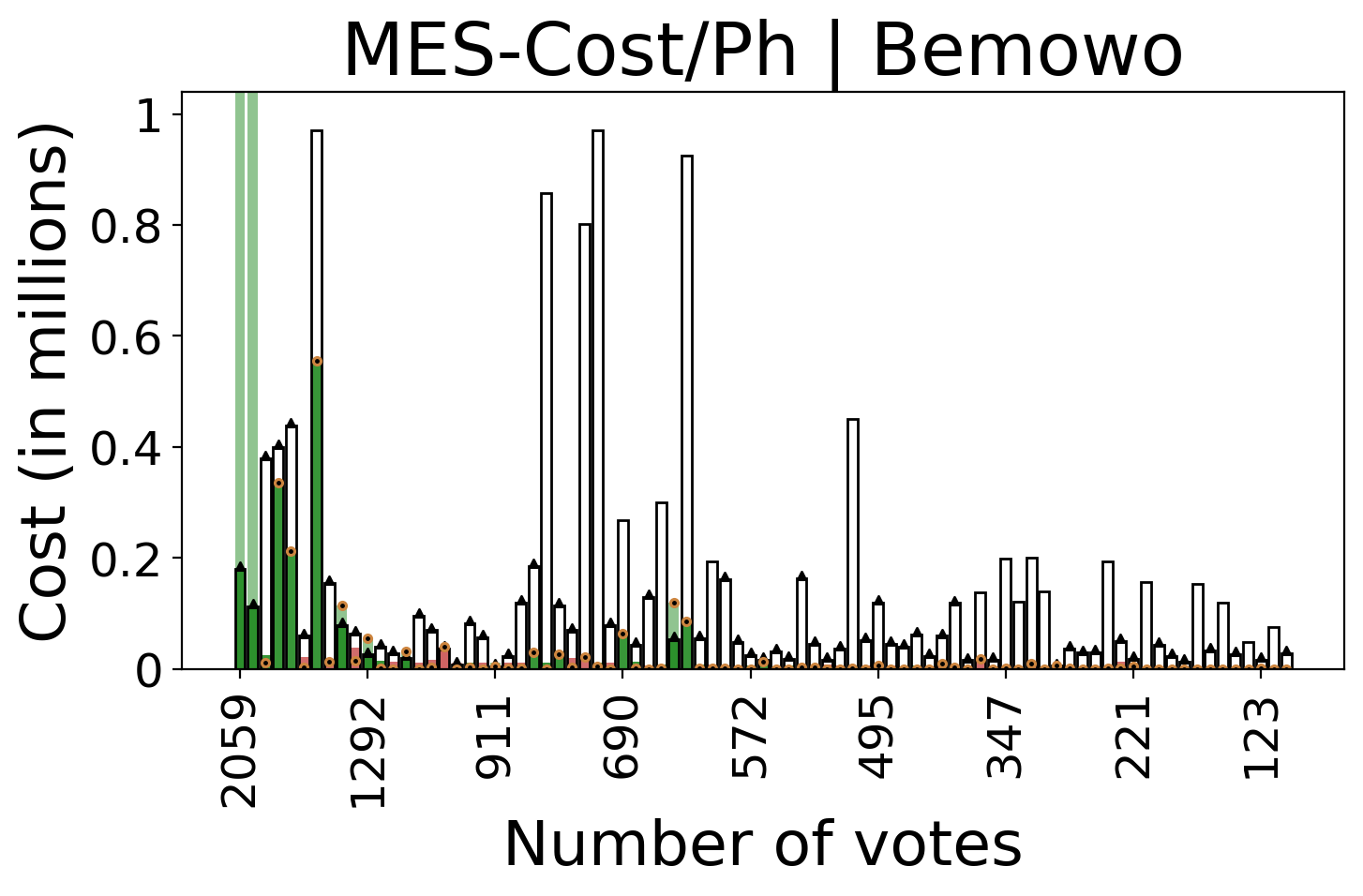}%
         \hspace{\hgap}%
     \includegraphics[width=\marginplotwidth]{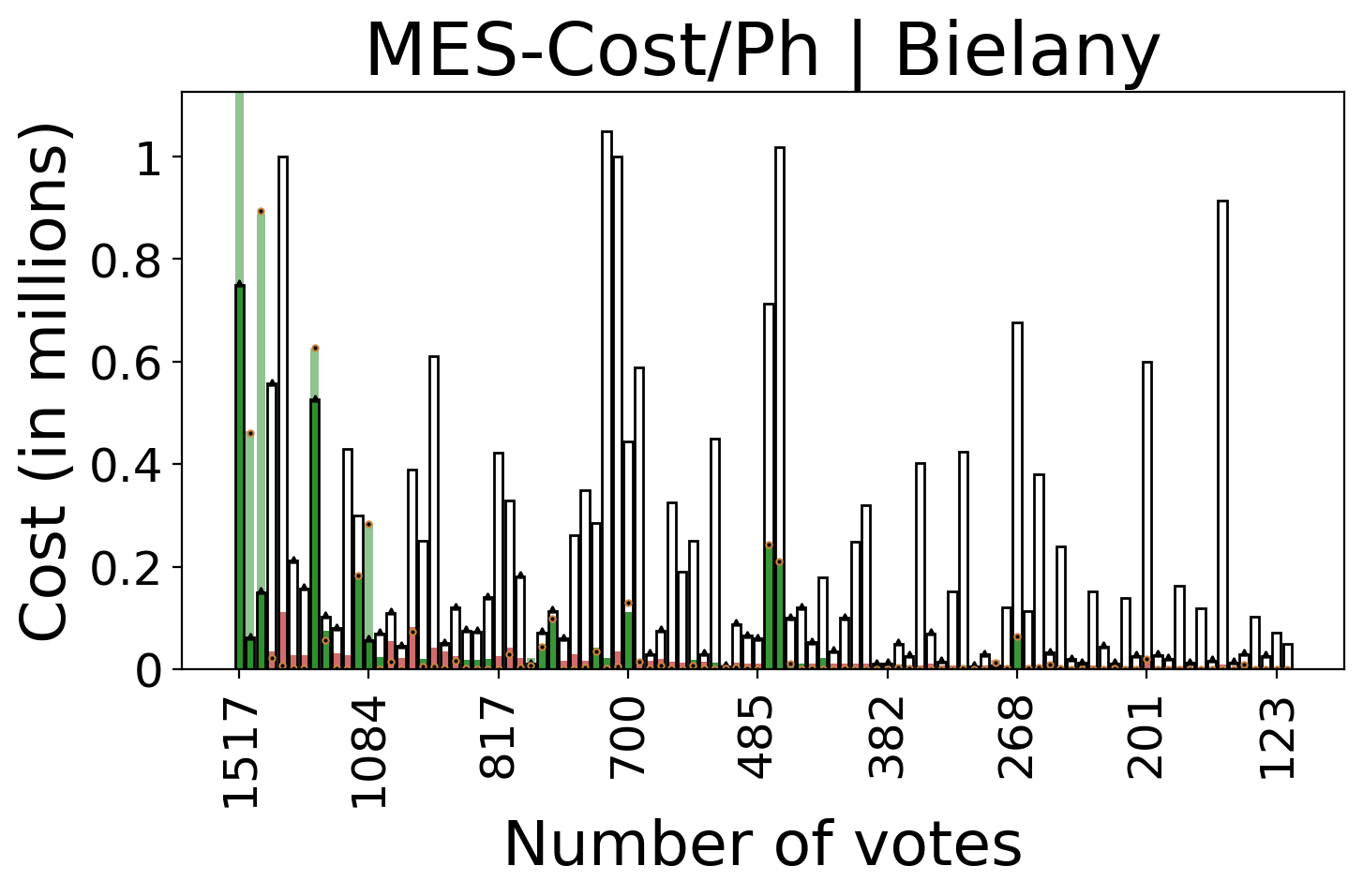}%
         \hspace{\hgap}%
     \includegraphics[width=\marginplotwidth]{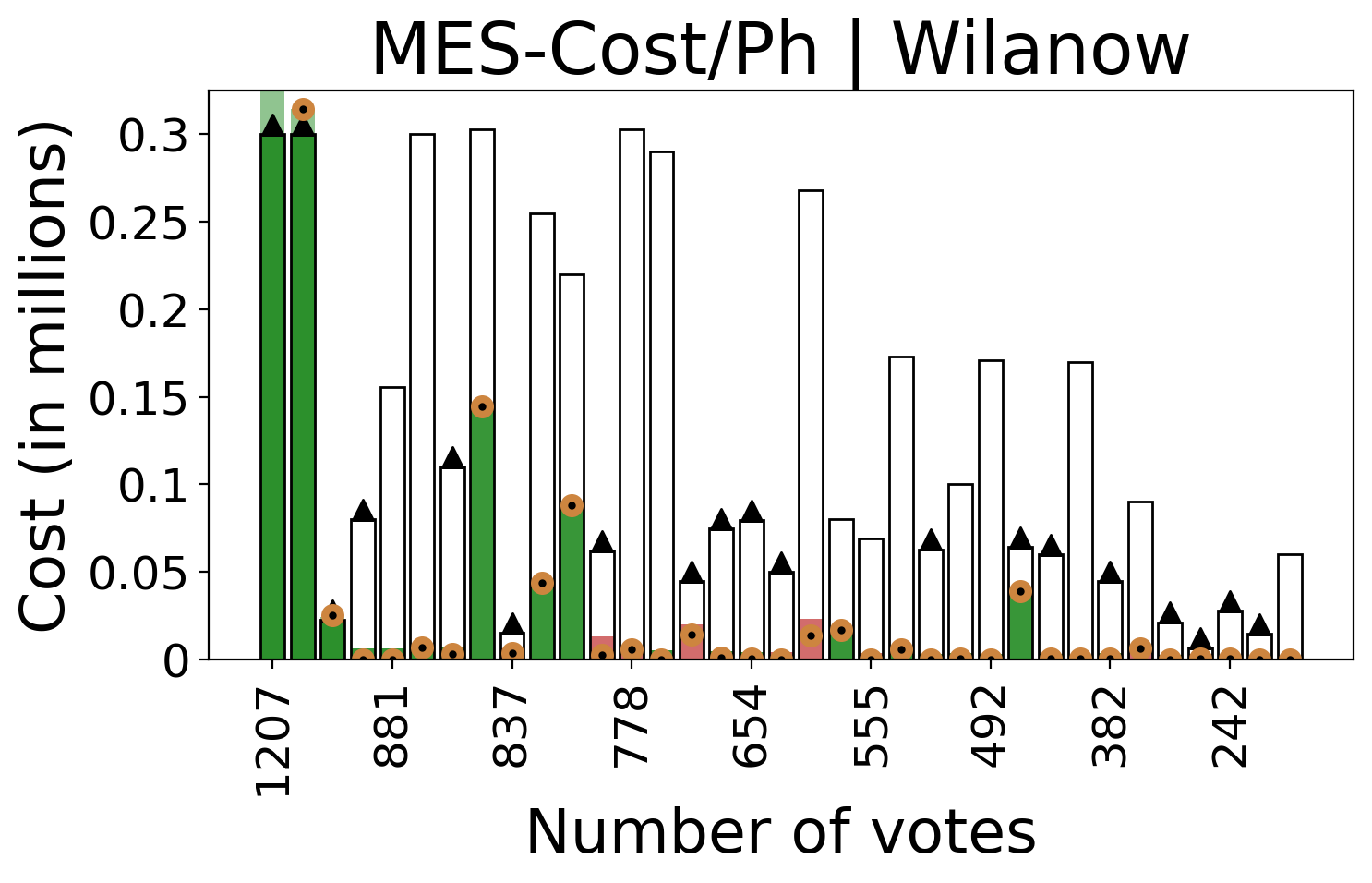}%
         \hspace{\hgap}%
     \includegraphics[width=\marginplotwidth]{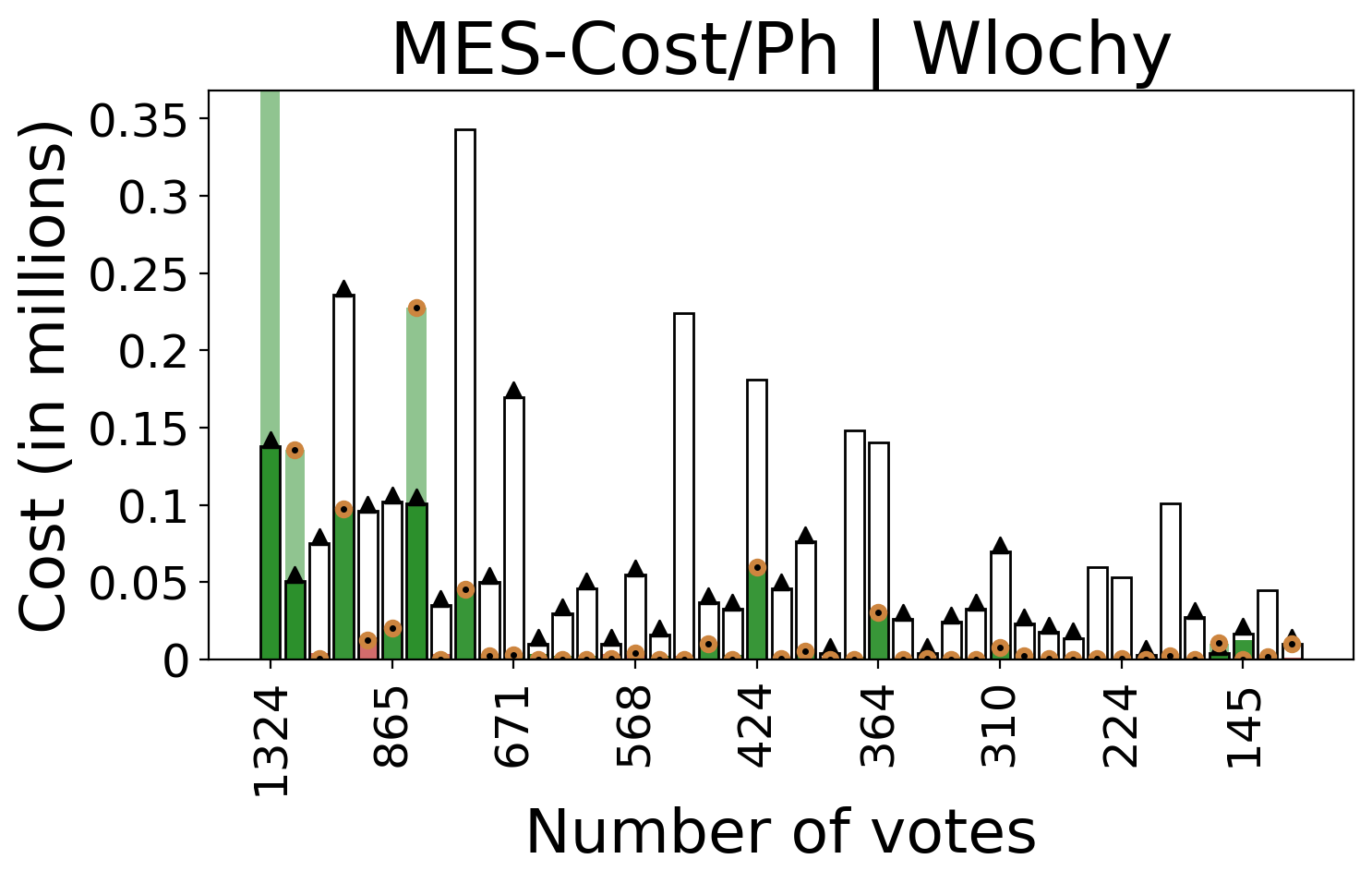}%

    \caption{ The results of the game after 10000 iterations. The green bars denote the final costs of the projects that are winning in the 10000th iteration (with the brighter part emphasizing the increase in comparison to the original costs). The red bars denote the final costs of projects that are losing in the 10000th iteration. Black outlines denote the original costs of the projects. Black triangles mark the projects that were winning in the 0th iteration (i.e., the original PB). Brown circles are denoting the NE.}\label{fig:games_apdx}
\end{figure*}

\end{document}